\documentclass[11pt,a4paper]{article}
\usepackage{fullpage}
\usepackage{url}
\usepackage{hyperref}

\usepackage{graphicx}
\usepackage{amsmath}
\usepackage{amssymb}
\usepackage{amsthm}
\usepackage{thm-restate}

\usepackage{tikz}
\usepackage{tikz-cd}
\usetikzlibrary{calc}
\usetikzlibrary{arrows.meta, positioning}

\usepackage{breqn}
\usepackage{array}
\usepackage{tabularx}
\usepackage[ruled,linesnumbered,noend]{algorithm2e}
\SetKwComment{Comment}{/* }{ */}
\SetKwInOut{Input}{input}
\SetKwInOut{Output}{output}
\usepackage{diagbox}
\usepackage{makecell}
\usepackage{comment}
\usepackage{color}

\usepackage{mdframed}

\usepackage{placeins}

\newmdenv[
  linewidth=1pt,
  roundcorner=5pt,
  linecolor=black,
  backgroundcolor=gray!10,
  skipabove=10pt,
  skipbelow=10pt,
  innertopmargin=5pt,
  innerbottommargin=5pt,
  innerleftmargin=10pt,
  innerrightmargin=10pt
]{definitionstyle}

\usepackage{float}
\usepackage{graphicx}
\usepackage{caption}
\usepackage{subcaption}

\usepackage[noend]{algpseudocode}

\newcommand{\vectorfromset}{\mathrm{vec}}

\usepackage[disable]{todonotes}

\makeatletter
\def\renewtheorem#1{%
  \expandafter\let\csname#1\endcsname\relax
  \expandafter\let\csname c@#1\endcsname\relax
  \gdef\renewtheorem@envname{#1}
  \renewtheorem@secpar
}
\def\renewtheorem@secpar{\@ifnextchar[{\renewtheorem@numberedlike}{\renewtheorem@nonumberedlike}}
\def\renewtheorem@numberedlike[#1]#2{\newtheorem{\renewtheorem@envname}[#1]{#2}}
\def\renewtheorem@nonumberedlike#1{
\def\renewtheorem@caption{#1}
\edef\renewtheorem@nowithin{\noexpand\newtheorem{\renewtheorem@envname}{\renewtheorem@caption}}
\renewtheorem@thirdpar
}
\def\renewtheorem@thirdpar{\@ifnextchar[{\renewtheorem@within}{\renewtheorem@nowithin}}
\def\renewtheorem@within[#1]{\renewtheorem@nowithin[#1]}
\makeatother

\newtheorem{theorem}{Theorem}
\newtheorem{lemma}[theorem]{Lemma}
\newtheorem{definition}[theorem]{Definition}
\newtheorem{observation}[theorem]{Observation}
\newtheorem{proposition}[theorem]{Proposition}
\newtheorem{conjecture}[theorem]{Conjecture}

\newtheorem{corollary}[theorem]{Corollary}

\newcommand{\falsevalue}{0}
\newcommand{\truevalue}{1}

\newcommand{\treeautomaton}{A}

\newcommand{\statestreeautomaton}{Q}
\newcommand{\transitionstreeautomaton}{\Delta}
\newcommand{\finalstatestreeautomaton}{F}
\newcommand{\lang}{\mathcal{L}}

\newcommand{\treewidthvalue}{k}

\newcommand{\alphabet}{\Sigma}
\newcommand{\analgorithm}{\mathfrak{A}}
\newcommand{\decompositiongraph}[1]{\mathcal{G}(#1)}
\newcommand{\cliquenumber}{\omega}
\newcommand{\maxDegree}{\Delta}
\newcommand{\chromaticnumber}{\chi}
\newcommand{\agraph}{G}
\newcommand{\bgraph}{H}
\newcommand{\Nplus}{\mathbb{N}_{+}}
\newcommand{\N}{\mathbb{N}}

\newcommand{\abag}{\mathfrak{b}}
\newcommand{\bagset}{B}

\newcommand{\allsubterms}{\mathrm{Sub}}

\newcommand{\aterm}{\tau}

\newcommand{\nodes}[1]{\mathsf{Nodes}(#1)}
\newcommand{\anode}{p}

\newcommand{\vertexset}[1]{V(#1)}
\newcommand{\edgeset}[1]{E(#1)}

\newcommand{\incidencerelationname}{\rho}
\newcommand{\incidencerelation}[1]{\incidencerelationname_{#1}}
\newcommand{\edgeendpointsname}{\mathrm{endpts}}
\newcommand{\invariant}{\mathcal{I}}

\newcommand{\anedge}{e}
\newcommand{\avertex}{v}

\newcommand{\introvertextype}[1]{\mathtt{IntroVertex}\{#1\}}
\newcommand{\introedgetype}[2]{\mathtt{IntroEdge}\{#1,#2\}}
\newcommand{\forgetvertextype}[1]{\mathtt{ForgetVertex}\{#1\}}
\newcommand{\jointype}{\mathtt{Join}}
\newcommand{\leaftype}{\mathtt{Leaf}}

\newcommand{\introvertexgeneric}[1]{\mathtt{IntroVertex}\{#1\}}
\newcommand{\introedgegeneric}[2]{\mathtt{IntroEdge}\{#1,#2\}}
\newcommand{\forgetvertexgeneric}[1]{\mathtt{ForgetVertex}\{#1\}}

\newcommand{\initialsetgeneric}{\mathtt{Leaf}}
\newcommand{\finalwitnessgeneric}{\mathtt{Final}}

\newcommand{\joingenericcore}{\mathtt{Join}}
\newcommand{\initialsetgenericcore}{\mathtt{Leaf}}
\newcommand{\finalwitnessgenericcore}{\mathtt{Final}}
\newcommand{\cleaningfunctioncore}{\mathtt{Clean}}
\newcommand{\accepteddecompositions}[1]{\mathsf{Accepted}(#1)}

\newcommand{\asymbol}{a}
\newcommand{\arity}{\mathfrak{a}}
\newcommand{\labelingfunction}{\lambda}

\newcommand{\topbagname}{B}
\newcommand{\topbag}[1]{\topbagname(#1)}

\newcommand{\allabstractdecompositionspathwidth}[1]{\mathsf{IPD}_{#1}}
\newcommand{\allabstractdecompositionstreewidth}[1]{\mathsf{ITD}_{#1}}
\newcommand{\abstractdecomposition}{\tau}
\newcommand{\sigmaabstractdecomposition}{\sigma}

\newcommand{\defeq}{~\dot{=}~}

\newcommand{\vertexone}{u}
\newcommand{\vertextwo}{v}

\newcommand{\allwitnesses}{\mathcal{W}}
\newcommand{\awitness}{\mathsf{w}}

\newcommand{\abstractalphabet}[1]{\Sigma_{#1}}
\newcommand{\atree}{T}

\newcommand{\finitepowerset}[1]{\mathcal{P}_{\mathrm{fin}}(#1)}
\newcommand{\powerset}[1]{\mathcal{P}(#1)}
\newcommand{\powersetchoosek}[2]{\mathcal{P}(#1,#2)}

\newcommand{\Terms}[1]{\mathrm{Terms}(#1)}

\newcommand{\graphproperty}{\mathbb{P}}
\newcommand{\dpcoregraphproperty}[1]{\mathbb{G}(#1)}

\newcommand{\asetgraphs}{S}

\newcommand{\isomorphic}{\sim}
\newcommand{\isomorphism}{\varphi}
\newcommand{\isomorphismclosure}[1]{\mathsf{ISO}(#1)}

\newcommand{\dpcore}{\mathsf{D}}

\newcommand{\witnessset}{S}

\newcommand{\allgraphs}{\textsc{Graphs}}

\newcommand{\allgraphstreewidth}[1]{\textsc{GraphsTw}[#1]}
\newcommand{\allgraphspathwidth}[1]{\textsc{GraphsPw}[#1]}

\newcommand{\isomorphismvertices}{\phi}
\newcommand{\isomorphismedges}{\nu}

\newcommand{\dpNamePart}[1]{\mbox{\normalfont\ttfamily #1}}
\newcommand{\hasMultiEdgeProperty}{\dpNamePart{Prop}\_\allowbreak\dpNamePart{HasMultipleEdges}}
\newcommand{\predicateHasMultiEdge}[1]{\dpNamePart{Pred}\_\allowbreak\dpNamePart{HasMultipleEdges}(#1)}

\newcommand{\simpleCliqueNumberAtLeastProperty}[1]{\dpNamePart{Prop}\_\allowbreak\dpNamePart{SimpleCliqueNumber}\_\allowbreak\dpNamePart{AtLeast}(#1)}
\newcommand{\maxDegreeAtLeastProperty}[1]{\dpNamePart{Prop}\_\allowbreak\dpNamePart{MaximumDegree}\_\allowbreak\dpNamePart{AtLeast}(#1)}

\newcommand{\chromaticNumberAtMostProperty}[1]{\dpNamePart{Prop}\_\allowbreak\dpNamePart{ChromaticNumber}\_\allowbreak\dpNamePart{AtMost}(#1)}

\newcommand{\reedProperty}{\textsc{Reed}}

\newcommand{\hasMultiEdgeCore}{\dpNamePart{HasMultipleEdges}}

\newcommand{\simpleCliqueNumberAtLeastCore}[1]{\dpNamePart{SimpleCliqueNumber}\_\allowbreak\dpNamePart{AtLeast}(#1)}
\newcommand{\maxDegreeAtLeastCore}[1]{\dpNamePart{MaximumDegree}\_\allowbreak\dpNamePart{AtLeast}(#1)}
\newcommand{\predicateMaxDegree}[2]{\dpNamePart{Pred}\_\allowbreak\dpNamePart{MaximumDegree}\_\allowbreak\dpNamePart{AtLeast}(#1,#2)}

\newcommand{\topmapname}{\theta}
\newcommand{\topmap}[1]{\topmapname[#1]}

\newcommand{\cell}{c}

\newcommand{\invariantCore}{\mathtt{Inv}}

\newcommand{\dynamizationfunctionname}[1]{\Gamma[{#1}]}
\newcommand{\dynamizationfunction}[2]{\dynamizationfunctionname{#1,#2}}

\newcommand{\xvertex}{x}

\newcommand{\domain}{dom}

\newcommand{\numbercolors}{r}

\newcommand{\predicateColorable}[2]{\dpNamePart{Pred}\_\allowbreak\dpNamePart{ChromaticNumber}\_\allowbreak\dpNamePart{AtMost}(#1,#2)}
\newcommand{\chromaticnumberAtMostCore}[1]{\dpNamePart{ChromaticNumber}\_\allowbreak\dpNamePart{AtMost}(#1)}
\newcommand{\chromaticNumberAtMostCore}[1]{\dpNamePart{ChromaticNumber}\_\allowbreak\dpNamePart{AtMost}(#1)}

\newcommand{\cliqueparameter}{\omega}

\newcommand{\cliquemap}{\gamma}

\newcommand{\cliquesize}{s}
\newcommand{\cliquefound}{\mathfrak{b}}
\newcommand{\predicateSimpleCliqueNumberAtLeast}[2]{\dpNamePart{Pred}\_\allowbreak\dpNamePart{SimpleCliqueNumber}\_\allowbreak\dpNamePart{AtLeast}(#1,#2)}
\newcommand{\maxdegreeparameter}{d}
\newcommand{\maxdegreefound}{b}

\newcommand{\realizationclass}{\mathcal{A}}
\newcommand{\dprefutation}{R}

\newcommand{\combinator}{\mathcal{C}}

\newcommand{\vertexsetname}{V}
\newcommand{\edgesetname}{E}

\newcommand{\arityvalue}{r}

\newcommand{\multiedgefound}{\mathrm{b}}

\newcommand{\graphfunction}{\mathcal{G}}

\newcommand{\alphabetclass}{\Sigma}
\newcommand{\transitionsdpcore}{\Delta}

\newcommand{\subtermToBag}{T}

\newcommand{\mxdp}{\predicateMaxDegree{\maxdegreeparameter}{\treewidthvalue}}

\newcommand{\cnc}{\chromaticnumberAtMostCore{\numbercolors}}
\newcommand{\cnp}{\predicateColorable{\numbercolors}{\treewidthvalue}}

\newcommand{\vertexsub}{S}
\newcommand{\vdegree}{d}

\newcommand{\relabelingfunction}{f}

\newcommand{\relabelclass}{\mathcal{F}}
\newcommand{\action}{\rho}

\newcommand{\canonic}{\mathrm{CAN}}

\newcommand{\relabelsequence}{\textsc{F}}

\title{State Canonization and Early Pruning in\\Width-Based Automated Theorem Proving}

\author{Mateus de Oliveira Oliveira$^{1,2}$ and Sam Urmian$^{2,3}$ \\
$^1$Department of Computer and Systems Sciences, Stockholm University, Sweden \\
$^2$Department of Informatics, University of Bergen, Norway\\
$^3$Centre for the Science of Learning \& Technology, University of Bergen, Norway\\
oliveira@dsv.su.se \hspace{0.5cm} sam.urmian@uib.no}

\begin{document}

\maketitle

\begin{abstract}
Width-based automated theorem proving is a framework where counterexamples to graph-theoretic conjectures are searched width-wise relative to some graph width measure, such as treewidth or pathwidth.
In a recent work it has been shown that dynamic programming algorithms operating on tree decompositions can be combined together with the purpose of width-based theorem proving. This approach can be used to show that several long-standing conjectures in graph theory can be tested in time $2^{2^{k^{O(1)}}}$ on the class of graphs of treewidth at most $k$. In this work, we give the first steps towards evaluating the viability of this framework from a practical standpoint. At the same time, we advance the framework in two directions.  First, we
introduce a state-canonization technique that significantly reduces the number of states evaluated during the search for a counterexample of the conjecture. Second, we introduce an early-pruning technique that can be applied in the study of conjectures of the form $\graphproperty_1\rightarrow \graphproperty_2$, for graph properties $\graphproperty_1$ and $\graphproperty_2$, where $\graphproperty_1$ is a property closed under subgraphs.

As a concrete application, we use our framework in the study of graph-theoretic conjectures related to coloring triangle-free graphs. In particular, our algorithm is able to show that Reed's conjecture for triangle-free graphs is valid on the class of graphs of pathwidth at most 5, and on graphs of treewidth at most 3. Perhaps more interestingly, our algorithm is able to construct in a completely automated way counterexamples to invalid strengthenings of Reed's conjecture. These are the first results showing that width-based automated theorem proving is a promising avenue in the study of graph-theoretic conjectures.
\end{abstract}

\section{Introduction}

Width-based automated theorem proving is a framework where parameterized algorithms are employed to search for counterexamples to graph-theoretic conjectures.
Within this framework, the search for counterexamples is conducted width-wise, relative to some specific width measure for graphs, such as treewidth or pathwidth. More specifically, given a conjecture $C$ and a positive integer $k$, the objective is to determine whether $C$ holds on the class of graphs of width at most $k$. If $C$ does not hold on this class of graphs, a counterexample of width at most $k$ that invalidates the conjecture should be produced.

This approach is relevant for two main reasons. First, many interesting classes of graphs have small width with respect to some graph width measure. For example, trees and forests have treewidth at most $1$, series-parallel graphs have treewidth at most $2$ and outerplanar graphs have treewidth at most $2$ \cite{kloks1994treewidth,bodlaender1986classes,brandstadt1999graph}.
It has also been shown that $\treewidthvalue$-outerplanar graphs have treewidth $O(\treewidthvalue)$ \cite{biedl2015triangulating,kammer2007determining}, and that $\treewidthvalue$-caterpillars have pathwidth at most $\treewidthvalue$ \cite{proskurowski1996bandwidth}. Second, many important conjectures in graph theory are not known to hold on classes of graphs of small treewidth or pathwidth, and it is therefore natural to try to determine whether such conjectures hold when restricted to such classes of graphs.
In many cases, structural properties of graphs of small treewidth or pathwidth have been used to produce analytic proofs of special cases of several well studied conjectures  \cite{guan1999game,campos2013dominating,heldt2014bend,hell2000circular,pinlou2006oriented,zhou2007upper,pan2002tight}. The framework of width-based automated theorem proving is an avenue for automating this approach for certain classes of conjectures.

In a recent work \cite{de2023width}, a new approach for width-based automated theorem proving was introduced. Within this approach, instead of specifying graph properties using logical formulas, such properties are specified using dynamic programming algorithms operating on graph decompositions. More specifically, it was shown in \cite{de2023width} that given dynamic programming cores $\dpcore_1,\dpcore_2,\dots,\dpcore_r$ specifying graph properties $\graphproperty_1,\graphproperty_2,\dots,\graphproperty_r$, and given a Boolean combination $\graphproperty$ of these properties, there is an algorithm $\analgorithm$ that takes a number $\treewidthvalue$ as input and decides whether all graphs of treewidth at most $\treewidthvalue$ belong to $\graphproperty$. Furthermore, if not all graphs of treewidth at most $\treewidthvalue$ belong to $\graphproperty$, then the algorithm outputs a certificate that can be used to extract a counterexample.
This approach takes double-exponential time with respect to the number of bits necessary to represent local witnesses used by the dynamic programming algorithms. Given that many interesting graph properties have DP-algorithms that use local witnesses of size $\treewidthvalue^{O(1)}$ when processing tree decompositions of width at most $\treewidthvalue$, the dynamic programming approach developed in \cite{de2023width} implies that many interesting conjectures can be tested in time double-exponential in $\treewidthvalue^{O(1)}$ on the class of graphs of treewidth at most $\treewidthvalue$. This includes long-standing conjectures such as Hadwiger's conjecture (for a fixed number of colors) \cite{bollobas1980hadwiger}, Tutte's flow conjectures \cite{jacobsen2013five,peres2021tutte}, and Reed's conjecture for triangle-free graphs \cite{reed1998omega}.

\subsection{Our Results}

In this work, we take the first steps towards evaluating the viability of the width-based automated theorem proving framework introduced in \cite{de2023width} from a practical perspective. At the same time, we advance this framework by introducing a new width-based deduction algorithm that produces significantly fewer states than the algorithm originally introduced in \cite{de2023width}. Our new algorithm leverages two techniques. First, we introduce a suitable notion of state canonization. Second, we introduce an early pruning technique that can be used in the study of conjectures of the form $\graphproperty_1\rightarrow \graphproperty_2$, where $\graphproperty_1$ is a graph property closed under subgraphs. As a concrete case study, we use our implementation to
study graph-theoretic statements related to colorings of graphs. Our algorithm produced nontrivial counterexamples for false statements and confirmed in a completely automated way a well-known conjecture due to Reed \cite{reed1998omega} on the class of graphs of pathwidth at most $5$, and on the class of graphs of treewidth at most $3$.
Together, our results provide the first evidence that the width-based ATP framework introduced in \cite{de2023width} may be a promising avenue in the study of graph-theoretic conjectures.

The approach for width-based automated theorem proving introduced in \cite{de2023width}, and further developed in the present work, is heavily based on dynamic programming algorithms deciding graph-theoretic properties. Another suitable approach for the study of width-based automated theorem proving is based on logic \cite{courcelle1990monadic,seese1991structure,courcelle2012automata}. In this context, instead of using dynamic-programming algorithms to define graph properties, such graph properties are defined using logical formulas. For example, it can be shown that for each formula $\varphi$ in the monadic second-order logic of graphs, there is an algorithm that takes a positive integer $k$ as input and determines whether every graph of treewidth at most $k$ satisfies $\varphi$ \cite{seese1991structure,courcelle1990monadic}. The drawback with this approach is that the running time of the deduction algorithm is governed by a function that grows as a tower of exponentials on the quantifier depth of the formula $\varphi$. It is worth noting that several interesting graph-theoretic properties have dynamic programming algorithms that produce local witnesses of size $k^{O(1)}$ when processing tree decompositions of width at most $k$, while at the same time requiring large quantifier depth to be expressed in MSO logic \cite{de2023width}. For such properties, while the time complexity implied by the logic approach is upper bounded by a tower of exponentials on $k$, the time complexity implied by the dynamic-programming approach is upper bounded by a double exponential in $k^{O(1)}$.

Over the years, a large variety of general formalisms have been developed with the aim of solving combinatorial problems on graphs of bounded treewidth. This includes combinatorial approaches based on dynamic programming \cite{bodlaender1988dynamic,bodlaender2008combinatorial}, approaches based on answer set programming \cite{fichte2017answer}, logic-based model checking \cite{courcelle1990monadic,seese1991structure,courcelle2012automata}, etc.

A complementary line of research studies decomposition-first dynamic programming frameworks and declarative systems
that exploit supplied decompositions, often achieving strong instance-level performance for
satisfaction/optimization tasks~\cite{fichte2017answer,bannach2023structure}. Related declarative toolchains
(e.g., ASP systems such as clingo) provide extensible backends that can host decomposition-based methods
in practice~\cite{gebser2016theory}.
Additionally, first-order (FO) model checking has well-studied tractability frontiers tied to locality and structural restrictions, and in practice it is frequently supported by SAT/SMT and first-order ATP backends~\cite{de2008z3,barrett2011cvc4,kovacs2013first,weidenbach2009spass}.
Nevertheless, the primary focus of these approaches remains per-instance evaluation rather than the evaluation of a given property on a whole class of graphs of bounded width.

\section{Preliminaries}
\label{section:Preliminaries}

\par\medskip\noindent\textbf{Basic notation.}
We denote by $\N \defeq \{0,1,\dots\}$ the set of natural numbers, and by $\Nplus \defeq \N \setminus \{0\}$ the set of positive natural numbers.
For any $n \in \N$, we define the interval $[n] \defeq \{1,\dots,n\}$, noting in particular that $[0] = \emptyset$.
Given a set $S$, we write $\finitepowerset{S}$ for the set of all \emph{finite} subsets of $S$, and $\powerset{S}$ for the \emph{power set} of $S$, i.e., the set of all subsets of $S$.
For $k\in\N$, we write $\powersetchoosek{S}{k}$ for the set of all $k$-element subsets of $S$.
Let $f : X_1 \times X_2 \times \dots \times X_r \rightarrow X$ be a function, and let $S_i \subseteq X_i$ for each $i = 1,\dots,r$. We define the image of $f$ over the product $S_1 \times S_2 \times \dots \times S_r$ as
\[
f(S_1, S_2, \dots, S_r) \defeq \{f(s_1, s_2, \dots, s_r) \mid s_i \in S_i \text{ for } i = 1,\dots,r\}.
\]
A quick reference for the notation used in this paper is provided in Appendix~\ref{appendix:notation-quickref}.


\par\medskip\noindent\textbf{Graphs.}
In this work, an \emph{undirected} \emph{graph} is defined as a triple
$\agraph = (\vertexsetname, \edgesetname, \incidencerelationname)$,
where $\vertexsetname \subseteq \mathbb{N}$ is a \emph{finite} set of \emph{vertices},
$\edgesetname \subseteq \mathbb{N}$ is a finite set of \emph{edges}, and
$\incidencerelationname \subseteq \edgesetname \times \vertexsetname$ is an \emph{incidence relation}
with the property that each edge is incident to exactly two vertices.
For each edge $\anedge \in \edgesetname$, we define its \emph{endpoints} as $\edgeendpointsname(\anedge) = \{\avertex \in \vertexsetname \mid (\anedge, \avertex) \in \incidencerelationname\}$.
In what follows, we may write $\vertexset{\agraph}$, $\edgeset{\agraph}$, and
$\incidencerelation{\agraph}$ to denote $\vertexsetname$, $\edgesetname$, and
$\incidencerelationname$, respectively. The \emph{size} of a graph $\agraph$ is defined as $|\agraph| = |\vertexset{\agraph}| + |\edgeset{\agraph}|$.
We denote by $\allgraphs$ the set of all graphs. The \emph{empty graph} is the graph
$(\emptyset, \emptyset, \emptyset)$ with no vertices and no edges.
Note that vertices and edges are treated as distinct sorts, even though both are encoded as natural numbers; thus it is harmless if the same integer occurs in both $\vertexset{\agraph}$ and $\edgeset{\agraph}$.

\par\medskip\noindent\textbf{Loops, multi-edges, and simplicity.}
We consider loopless graphs; multiedges are allowed.
We say that a graph is \emph{simple} if it has no multiedge, i.e., whenever $e\neq e'$ are distinct edges then $\edgeendpointsname(e)\neq \edgeendpointsname(e')$.

\par\medskip\noindent\textbf{Disjoint union.}
Let $\agraph$ be a graph and let $i\in\{0,1\}$. We define the renamed copy $\agraph^{(i)}$ by
\[
\vertexset{\agraph^{(i)}} \defeq \{2v+i : v\in \vertexset{\agraph}\},\qquad
\edgeset{\agraph^{(i)}} \defeq \{2e+i : e\in \edgeset{\agraph}\},
\]
and
\[
\incidencerelation{\agraph^{(i)}} \defeq \{(2e+i,2v+i) : (e,v)\in \incidencerelation{\agraph}\}.
\]
For two graphs $\agraph$ and $\bgraph$, we define their disjoint union by
\[
\agraph \uplus \bgraph \defeq \agraph^{(0)} \cup \bgraph^{(1)}.
\]
That is, $\vertexset{\agraph\uplus\bgraph}=\vertexset{\agraph^{(0)}}\cup \vertexset{\bgraph^{(1)}}$,
$\edgeset{\agraph\uplus\bgraph}=\edgeset{\agraph^{(0)}}\cup \edgeset{\bgraph^{(1)}}$, and
$\incidencerelation{\agraph\uplus\bgraph}=\incidencerelation{\agraph^{(0)}}\cup \incidencerelation{\bgraph^{(1)}}$.
Whenever we say ``disjoint union'' of two graphs, we mean this fixed construction.


\par\medskip\noindent\textbf{Graph isomorphisms.}
An {\em isomorphism} from a graph $\agraph$ to a graph $\bgraph$ is a pair
$\isomorphism  = (\isomorphismvertices,\isomorphismedges)$ where
$\isomorphismvertices:\vertexset{\agraph}\rightarrow \vertexset{\bgraph}$ is a
bijection from the vertices of $\agraph$ to the vertices of $\bgraph$ and
$\isomorphismedges:\edgeset{\agraph}\rightarrow \edgeset{\bgraph}$ is a
bijection from the edges of $\agraph$ to the edges of $\bgraph$ with the
property that for each vertex $\avertex\in \vertexset{\agraph}$ and each edge
$\anedge \in \edgeset{\agraph}$, $(\anedge,\avertex)\in
\incidencerelation{\agraph}$ if and only if
$(\isomorphismedges(\anedge),\isomorphismvertices(\avertex))\in \incidencerelation{\bgraph}$.
If such a bijection exists, we say that $\agraph$
and $\bgraph$ are {\em isomorphic}, and denote this fact by $\agraph\sim \bgraph$.

\par\medskip\noindent\textbf{Graph embeddings.}
Let $H=(V_H,E_H,\incidencerelation{H})$ and $G=(V_G,E_G,\incidencerelation{G})$ be graphs.
An \emph{embedding} of $H$ into $G$ is a pair $(\phi,\nu)$ where $\phi:V_H\to V_G$ and $\nu:E_H\to E_G$ are injective functions such that for every $v\in V_H$ and every $e\in E_H$,
\[
(e,v)\in \incidencerelation{H} \iff (\nu(e),\phi(v))\in \incidencerelation{G}.
\]
We say that $H$ \emph{embeds as a subgraph} of $G$ if such an embedding exists.

\par\medskip\noindent\textbf{Graph properties.}
A {\em graph property} is any subset $\graphproperty \subseteq \allgraphs$
closed under isomorphisms. That is to say, for each two isomorphic graphs
$\agraph$ and $\bgraph$ in $\allgraphs$, $\agraph \in\graphproperty$ if and
only if $\bgraph\in\graphproperty$. Note that the sets $\emptyset$ and
$\allgraphs$ are graph properties. Given a set $\asetgraphs$ of graphs, the
{\em isomorphism closure} of $\asetgraphs$ is defined as the set
$\isomorphismclosure{\asetgraphs} = \{\agraph\in \allgraphs \;:\; \exists
\bgraph \in \asetgraphs, \agraph\isomorphic \bgraph\}.$
We use the standard Boolean operations on graph properties: for graph properties $P,Q\subseteq \allgraphs$, we write
$\neg P\defeq \allgraphs\setminus P$,
$P\wedge Q\defeq P\cap Q$,
$P\vee Q\defeq P\cup Q$,
and $P\rightarrow Q \defeq \neg P\vee Q$.
All of these are again graph properties (they are closed under isomorphism).

\par\medskip\noindent\textbf{Ranked alphabets and terms.}
A \emph{ranked alphabet} is a finite set $\alphabet$ equipped with a function $\arity: \alphabet \to \mathbb{N}$, which assigns to each symbol in $\alphabet$ a non-negative integer called its \emph{arity}.
A \emph{term} over a ranked alphabet $\alphabet$ is a pair $\aterm = (\atree, \labelingfunction)$, where $\atree$ is a finite rooted tree and $\labelingfunction : \nodes{\atree} \to \alphabet$ is a labeling function that assigns to each node $\anode \in \nodes{\atree}$ a symbol from $\alphabet$. The set $\nodes{\atree}$ denotes the set of nodes in the tree $\atree$, with the root node included, and where each node may have ordered children.
The labeling must respect arities: if $\labelingfunction(\anode)$ has arity $r$, then the node $\anode$ must have exactly $r$ children $\anode_1, \dots, \anode_r$. In particular, leaf nodes are labeled with symbols of arity $0$. We assume that the children of each node are ordered from left to right, so it makes sense to speak of the $i$-th child of a node.
We may write $\nodes{\aterm}$ to refer to $\nodes{\atree}$, and we use $|\aterm|$ to denote the number of nodes in the term, i.e., $|\nodes{\atree}|$. The \emph{height} of $\aterm$ is defined as the height of the tree $\atree$. The set of all terms over $\alphabet$ is denoted by $\Terms{\alphabet}$.
If $\aterm_1 = (\atree_1, \labelingfunction_1), \dots, \aterm_{\arityvalue} = (\atree_{\arityvalue}, \labelingfunction_{\arityvalue})$ are terms in $\Terms{\alphabet}$, and $\asymbol \in \alphabet$ is a symbol of arity $\arityvalue$, then the term $\asymbol(\aterm_1, \dots, \aterm_{\arityvalue})$ is defined after replacing the child terms by isomorphic copies with pairwise disjoint node sets. Thus reusing the same written term in two child positions means taking two fresh copies. The resulting pair is $\aterm = (\atree, \labelingfunction)$ where $u$ is a fresh node not appearing in any of the copied trees, and the node set of $\atree$ is $\{u\} \cup \nodes{\atree_1} \cup \dots \cup \nodes{\atree_{\arityvalue}}$. The root of $\atree$ is $u$, and we define $\labelingfunction(u) = \asymbol$. For each $j \in [\arityvalue]$, the labeling $\labelingfunction$ agrees with $\labelingfunction_j$ on the nodes of $\atree_j$, that is, $\labelingfunction|_{\nodes{\atree_j}} = \labelingfunction_j$.

\section{Instructive Dynamic Programming Cores}
\label{section:InstructiveTreeDecompositions}

In this work, we represent graphs of treewidth at most $\treewidthvalue$ using $\treewidthvalue$-instructive tree decompositions ($k$-ITDs) \cite{de2023width}. For each $\treewidthvalue\in \N$, consider the following alphabet
\begin{equation}
\label{equation:InstructiveAlphabet}
\begin{array}{l}
\abstractalphabet{\treewidthvalue} \; = \Big\{\leaftype,\; \introvertextype{\vertexone}, \introedgetype{\vertexone}{\vertextwo},
\\ \hspace{5.cm}
\forgetvertextype{\vertexone},
  \; \jointype
  \;:\; \vertexone,\vertextwo\in [\treewidthvalue+1], \vertexone\neq \vertextwo \Big\},
\end{array}
\end{equation}
\noindent where $\leaftype$ is a symbol of arity $0$, $\introvertextype{\vertexone}$, $\forgetvertextype{\vertexone}$
and $\introedgetype{\vertexone}{\vertextwo}$ are symbols of arity $1$, and $\jointype$ is a symbol of arity $2$. We call $\abstractalphabet{\treewidthvalue}$ the {\em $k$-instructive alphabet}. Intuitively, elements of $\abstractalphabet{\treewidthvalue}$ represent instructions that can be used for the construction of a graph of treewidth at most $\treewidthvalue$, together with a set $\abag\subseteq [\treewidthvalue+1]$ of {\em active labels}, where each active label labels exactly one vertex of the graph.

\begin{enumerate}
\item In the base case, the instruction $\leaftype$ creates an empty graph with an empty set of active labels.
\item Now, let $\agraph$ be a graph with set of active labels $\abag$.
\begin{enumerate}
\item For each $\vertexone\in [\treewidthvalue+1]\backslash \abag$, the instruction $\introvertextype{\vertexone}$ adds a new vertex
to $G$, labels this vertex with $\vertexone$, and adds $\vertexone$ to $\abag$.
\item For each $\vertexone\in \abag$, the
instruction $\forgetvertextype{\vertexone}$ erases the label from the current vertex labeled with $\vertexone$, and
removes $\vertexone$ from $\abag$. The intuition is that the label $\vertexone$ is now free and may be used later in the
creation of another vertex.
\item For each two distinct $\vertexone,\vertextwo\in \abag$, the instruction $\introedgetype{\vertexone}{\vertextwo}$
introduces a new edge between the current vertex labeled with $\vertexone$ and the current vertex labeled with $\vertextwo$.
We note that multiedges are allowed in our graphs.
\end{enumerate}
\item Finally, if $\agraph$ and $\agraph'$ are two graphs, each having $\abag$ as the set of active labels, then the instruction
$\jointype$ creates a new graph by identifying, for each $\vertexone\in \abag$, the vertex of $\agraph$ labeled
with $\vertexone$ with the vertex of $\agraph'$ labeled with $\vertexone$.
\end{enumerate}

Such a construction process can be formalized using a term $\abstractdecomposition$ over $\abstractalphabet{\treewidthvalue}$, which specifies an inductive construction from the leaves towards the root. More specifically, leaves are labeled with the  $\leaftype$ instruction, nodes with a single child are labeled with an instruction of type $\introvertextype{\vertexone}$,
$\forgetvertextype{\vertexone}$, or
$\introedgetype{\vertexone}{\vertextwo}$, and nodes with two children are labeled with the $\jointype$ instruction. We let $\decompositiongraph{\abstractdecomposition}$ be the graph associated with the root of $\abstractdecomposition$. We let $\topbag{\abstractdecomposition}$ be the set of active labels after processing all operations in $\abstractdecomposition$, and let
 $\topmap{\abstractdecomposition}:\topbag{\abstractdecomposition} \rightarrow \vertexset{\decompositiongraph{\abstractdecomposition}}$ be the map that sends each label in
$\topbag{\abstractdecomposition}$ to its corresponding vertex in $\decompositiongraph{\abstractdecomposition}$.
We refer to Appendix~\ref{appendix:GraphFromDecomposition} for a formal inductive definition of $\decompositiongraph{\abstractdecomposition}$.
 We note that not all terms
over $\abstractalphabet{\treewidthvalue}$ give rise to legal graphs. For instance, if the set $\abag$ does not contain a label $\vertexone$ then the instruction $\forgetvertextype{\vertexone}$ is not well defined. Similarly, if $\vertexone$ is already in $\abag$, then the instruction $\introvertextype{\vertexone}$ is not well defined. In order to specify the set of all terms over $\abstractalphabet{\treewidthvalue}$
that do correspond to graphs, we may use a tree automaton $\treeautomaton_{\treewidthvalue}$.
More specifically, we let $\treeautomaton_{\treewidthvalue} =  (\alphabet_k,\statestreeautomaton_k,\finalstatestreeautomaton_k,\transitionstreeautomaton_k)$ be the tree automaton where $\statestreeautomaton_k = \finalstatestreeautomaton_k = \powerset{[\treewidthvalue+1]}$, and

\[
\begin{array}{lcl}
\transitionstreeautomaton_k &=& \{ \leaftype \rightarrow \emptyset\}\;  \\
& \cup& \{\introvertextype{\vertexone}(\abag) \rightarrow \abag \cup \{\vertexone\} \;|\;  \vertexone\in [\treewidthvalue+1]\backslash \abag\} \\
&\cup& \{\forgetvertextype{\vertexone}(\abag) \rightarrow \abag\backslash \{\vertexone\} \;|\; \vertexone\in \abag\} \\
&\cup&\{\introedgetype{\vertexone}{\vertextwo}(\abag) \rightarrow \abag\;|\; \vertexone,\vertextwo\in \abag, \vertexone\neq \vertextwo\} \\
&\cup& \{\jointype(\abag,\abag')\rightarrow \abag \;|\; \abag = \abag'\}.
\end{array}
\]


We let $\allabstractdecompositionstreewidth{k}  = \lang(\treeautomaton_{\treewidthvalue})$ where $\lang(\treeautomaton_{\treewidthvalue})$ is the set of terms accepted by
$\treeautomaton_{\treewidthvalue}$. The terms in $\allabstractdecompositionstreewidth{k}$ are called {\em $\treewidthvalue$-instructive decompositions}. It turns out that graphs that can be represented by
$\treewidthvalue$-instructive tree decompositions are precisely the graphs of treewidth at most $\treewidthvalue$.

\begin{lemma}
[\cite{de2023width}]
\label{lemma:TreewidthIsTreelikeConfirmation}
Let $\agraph\in \allgraphs$ and $\treewidthvalue\in \N$. Then $\agraph$ has
treewidth at most $\treewidthvalue$ if and only if there exists a $\treewidthvalue$-instructive tree decomposition $\abstractdecomposition$ such that $\graphfunction(\abstractdecomposition) \simeq \agraph$.
\end{lemma}

We also use the join-free specialization of this representation for pathwidth. Let $\allabstractdecompositionspathwidth{\treewidthvalue}$ be the subset of $\allabstractdecompositionstreewidth{\treewidthvalue}$ consisting of terms with no occurrence of the join symbol $\jointype$. These are the $\treewidthvalue$-instructive path decompositions. As in the standard nice-decomposition representation, a graph has pathwidth at most $\treewidthvalue$ if and only if it is represented by some term in $\allabstractdecompositionspathwidth{\treewidthvalue}$. All search algorithms below specialize to pathwidth by using this join-free transition system, i.e., by omitting join successors.

Dynamic programming algorithms operating on tree decompositions can be formalized using the notion of an {\em instructive dynamic programming core} (instructive DP-core), as defined below.

\begin{definition}[Instructive DP-Cores]
\label{definition:DynamicSignature}
An \emph{instructive dynamic programming core} is a sequence of $6$-tuples
$\dpcore = \{(\alphabetclass_{\treewidthvalue},\allwitnesses_{\treewidthvalue},\finalwitnessgenericcore_{\treewidthvalue},\transitionsdpcore_{\treewidthvalue},\cleaningfunctioncore_{\treewidthvalue}, \invariantCore_{\treewidthvalue})\}_{\treewidthvalue\in \N}$ where for each $\treewidthvalue\in \N$,
\begin{enumerate}
\item $\alphabetclass_{\treewidthvalue}$ is the $\treewidthvalue$-instructive alphabet;
\item $\allwitnesses_{\treewidthvalue}\subseteq \{0,1\}^*$ is a decidable subset of $\{0,1\}^*$;
\item $\finalwitnessgenericcore_{\treewidthvalue}: \allwitnesses_{\treewidthvalue}\rightarrow \{\falsevalue,\truevalue\}$ is a function;
\item $\transitionsdpcore_{\treewidthvalue}$ is a set containing
\begin{itemize}
\item A {\em finite} subset $\initialsetgenericcore\subseteq \allwitnesses_{\treewidthvalue}$.
\item A function $\introvertexgeneric{\vertexone}:\allwitnesses_{\treewidthvalue}\rightarrow \finitepowerset{\allwitnesses_{\treewidthvalue}}$ for each $\vertexone\in [\treewidthvalue+1]$.
\item A function $\forgetvertexgeneric{\vertexone}:\allwitnesses_{\treewidthvalue} \rightarrow \finitepowerset{\allwitnesses_{\treewidthvalue}}$ for each
		$\vertexone\in [\treewidthvalue+1]$.
\item A function $\introedgegeneric{\vertexone}{\vertextwo}:\allwitnesses_{\treewidthvalue}\rightarrow\finitepowerset{\allwitnesses_{\treewidthvalue}}$ for each $\{\vertexone,\vertextwo\}\in \powersetchoosek{[\treewidthvalue+1]}{2}$.
\item A function $\joingenericcore:\allwitnesses_{\treewidthvalue} \times \allwitnesses_{\treewidthvalue} \rightarrow \finitepowerset{\allwitnesses_{\treewidthvalue}}$.
\end{itemize}
\item $\cleaningfunctioncore_{\treewidthvalue}:\finitepowerset{\allwitnesses_{\treewidthvalue}} \rightarrow \finitepowerset{\allwitnesses_{\treewidthvalue}}$ is a function;
\item $\invariantCore_{\treewidthvalue}:\finitepowerset{\allwitnesses_{\treewidthvalue}} \rightarrow \{0,1\}^*$ is a function.
\end{enumerate}
\end{definition}

For each $\treewidthvalue\in \N$, we let
$\dpcore[\treewidthvalue]= (\alphabetclass_{\treewidthvalue},\allwitnesses_{\treewidthvalue},\finalwitnessgenericcore_{\treewidthvalue},
\transitionsdpcore_{\treewidthvalue},\cleaningfunctioncore_{\treewidthvalue}, \invariantCore_{\treewidthvalue})$
denote the
$\treewidthvalue$-th tuple of $\dpcore$.
We may write $\dpcore[\treewidthvalue].\alphabetclass$ to denote the set
$\alphabetclass_{\treewidthvalue}$, $\dpcore[\treewidthvalue].\allwitnesses$ to denote the set $\allwitnesses_{\treewidthvalue}$,
and so on.
Intuitively, for each $\treewidthvalue$, $\dpcore[\treewidthvalue]$ is a description of a dynamic programming algorithm that operates on $\treewidthvalue$-instructive tree decompositions. Such an algorithm processes a $\treewidthvalue$-instructive tree decomposition $\abstractdecomposition$ from the leaves
towards the root, and assigns a set of local witnesses to each node of $\abstractdecomposition$, depending on which instruction labels the node and on the sets assigned to the children of the node.
The transition functions above are defined on individual witnesses. Whenever we apply a transition to finite witness sets, we use the following union liftings:
\[
\introvertexgeneric{\vertexone}(S)\defeq \bigcup_{\awitness\in S}\introvertexgeneric{\vertexone}(\awitness),
\]
\[
\forgetvertexgeneric{\vertexone}(S)\defeq \bigcup_{\awitness\in S}\forgetvertexgeneric{\vertexone}(\awitness),
\]
\[
\introedgegeneric{\vertexone}{\vertextwo}(S)\defeq \bigcup_{\awitness\in S}\introedgegeneric{\vertexone}{\vertextwo}(\awitness),
\]
\[
\joingenericcore(S,S')\defeq \bigcup_{\awitness\in S,\;\awitness'\in S'}\joingenericcore(\awitness,\awitness').
\]
All occurrences of $\introvertexgeneric{\vertexone}$, $\forgetvertexgeneric{\vertexone}$, $\introedgegeneric{\vertexone}{\vertextwo}$, and $\joingenericcore$ with finite witness sets as arguments use these lifted operations.
Some dynamic programming algorithms use a function that removes redundant local witnesses from the set of local witnesses constructed at each node. In our framework, this is formalized by the
function $\cleaningfunctioncore_{\treewidthvalue}$. Unless a DP-core explicitly specifies a nontrivial cleaning function, we use the identity cleaning function.
Finally, the function $\invariantCore_{\treewidthvalue}$ is used whenever we want to use dynamic programming algorithms to compute graph invariants. In this work, we will not be concerned
with the computation of invariants, and therefore, we assume that
$\invariantCore_{\treewidthvalue}$ is the Boolean function that assigns $1$ to a set of local witnesses if and only if it contains some final witness.
This process is formalized by the notion of {\em dynamization}, which we define below.

\begin{definition}[Dynamization]
\label{definition:Dynamization}
Let $\treewidthvalue\in \N$ and $\dpcore$ be an instructive DP-core.
The $\treewidthvalue$-dynamization of $\dpcore$ is the function
$\dynamizationfunctionname{\dpcore,\treewidthvalue}:\allabstractdecompositionstreewidth{\treewidthvalue}\rightarrow \finitepowerset{\dpcore[\treewidthvalue].\allwitnesses}$
inductively
defined as follows for each $\abstractdecomposition\in \allabstractdecompositionstreewidth{\treewidthvalue}$.
\begin{enumerate}
\setlength\itemsep{0em}
\item If $\abstractdecomposition = \leaftype$, then
$\dynamizationfunction{\dpcore}{\treewidthvalue}(\abstractdecomposition)=\dpcore[\treewidthvalue].\initialsetgeneric.$
\item If $\abstractdecomposition = \introvertextype{\vertexone}(\sigmaabstractdecomposition)$, then
$
\dynamizationfunction{\dpcore}{\treewidthvalue}(\abstractdecomposition)= 
\dpcore[\treewidthvalue].\cleaningfunctioncore(
\introvertextype{\vertexone}(
\dynamizationfunction{\dpcore}{\treewidthvalue}(\sigmaabstractdecomposition)
)
).
$
%
\item If $\abstractdecomposition = \forgetvertextype{\vertexone}(\sigmaabstractdecomposition)$, then
$
 \dynamizationfunction{\dpcore}{\treewidthvalue}(\abstractdecomposition)=
  \dpcore[\treewidthvalue].\cleaningfunctioncore(
  \forgetvertextype{\vertexone}(
 \dynamizationfunction{\dpcore}{\treewidthvalue}(\sigmaabstractdecomposition)
)
).
$
\item If $\abstractdecomposition = \introedgetype{\vertexone}{\vertextwo}(\sigmaabstractdecomposition)$, then
$\dynamizationfunction{\dpcore}{\treewidthvalue}(\abstractdecomposition)=
\dpcore[\treewidthvalue].\cleaningfunctioncore(
\introedgetype{\vertexone}{\vertextwo}(
\dynamizationfunction{\dpcore}{\treewidthvalue}(\sigmaabstractdecomposition))).$
%
\item If $\abstractdecomposition = \jointype(\sigmaabstractdecomposition_1,\sigmaabstractdecomposition_2)$, then
$
\dynamizationfunction{\dpcore}{\treewidthvalue}(\abstractdecomposition)=
\dpcore[\treewidthvalue].\cleaningfunctioncore(\jointype(
\dynamizationfunction{\dpcore}{\treewidthvalue}(\sigmaabstractdecomposition_1), \dynamizationfunction{\dpcore}{\treewidthvalue}(\sigmaabstractdecomposition_2)
)
).$
\end{enumerate}
\end{definition}

We say that $\dpcore[\treewidthvalue]$ accepts $\abstractdecomposition$ if
$\dynamizationfunction{\dpcore}{\treewidthvalue}(\abstractdecomposition)$ has a final local witness,
i.e. a local witness $\awitness$ with the property that $\dpcore[\treewidthvalue].\finalwitnessgenericcore(\awitness)=1$.
We let  $\dpcoregraphproperty{\dpcore[\treewidthvalue]} = \isomorphismclosure{\{\decompositiongraph{\abstractdecomposition}\;:\;
\abstractdecomposition \mbox{ is accepted by }\dpcore[\treewidthvalue]\}}$ be the
isomorphism closure of the set of graphs associated with terms accepted by $\dpcore[\treewidthvalue]$.
We note that $\dpcoregraphproperty{\dpcore[\treewidthvalue]}$ is a graph property, and that all graphs in  $\dpcoregraphproperty{\dpcore[\treewidthvalue]}$ have treewidth at most $\treewidthvalue$.
 We let $\dpcoregraphproperty{\dpcore}=\bigcup_{\treewidthvalue\in \N} \dpcoregraphproperty{\dpcore[\treewidthvalue]}$ be the graph property defined by $\dpcore$.

\begin{definition}[Coherency]
\label{definition:CoherencyDPCore}
Let $\dpcore$ be an instructive DP-core.
We say that $\dpcore$ is coherent if for each $\treewidthvalue,\treewidthvalue'\in \N$,
each $\abstractdecomposition\in \allabstractdecompositionstreewidth{\treewidthvalue}$, and
each $\abstractdecomposition'\in \allabstractdecompositionstreewidth{\treewidthvalue'}$
if $\decompositiongraph{\abstractdecomposition} \simeq \decompositiongraph{\abstractdecomposition'}$
then $\dpcore[\treewidthvalue]$ accepts $\abstractdecomposition$ if and only if $\dpcore[\treewidthvalue']$ accepts $\abstractdecomposition'$.
\end{definition}

Let $\dpcore$ be a coherent instructive DP-core, and $\treewidthvalue\in \N$. A $(\treewidthvalue,\dpcore)$-state is a pair of the form $(\abag,\witnessset)$ where $\abag\subseteq [\treewidthvalue+1]$ and $\witnessset\subseteq \dpcore[\treewidthvalue].\allwitnesses$. Such a state is said to be {\em $(\treewidthvalue,\dpcore)$-inconsistent}, if $\witnessset$ has no final local witness. The initial $(\treewidthvalue,\dpcore)$-state is the pair $(\emptyset,\dpcore[\treewidthvalue].\initialsetgeneric)$.

\begin{definition}[$(\treewidthvalue,\dpcore)$-Refutation]\label{definition:DPRefutation}
Let $\dpcore$ be an instructive DP-core with identity cleaning. A $(\treewidthvalue,\dpcore)$-refutation is a sequence of pairs
$\dprefutation \equiv (\abag_0,\witnessset_0)(\abag_1,\witnessset_1)\dots (\abag_m,\witnessset_m)$ satisfying the following conditions.
\begin{enumerate}
\setlength\itemsep{0.5em}
\item \label{sequence-relabeling-condition-oneA}
$(\abag_0,\witnessset_0) = (\emptyset,\dpcore[\treewidthvalue].\initialsetgeneric)$.
\item \label{sequence-relabeling-condition-twoA}
$(\abag_m,\witnessset_m)$ is $(\treewidthvalue,\dpcore)$-inconsistent.
\item \label{sequence-relabeling-condition-threeA}
For each $i\in [m]$, there is some $j\in \{0,\dots,i-1\}$, such that $(\abag_i,\witnessset_i)$ is equal to
one of the following pairs.
\begin{enumerate}
\item \label{sequence-IntroVertexA}
$( \abag_j\cup \{\vertexone\}, \introvertextype{\vertexone}(\witnessset_j))$ with $\vertexone\notin \abag_j$.
\item \label{sequence-ForgetVertexA}
$(\abag_j\setminus \{\vertexone\}, \forgetvertextype{\vertexone}(\witnessset_j))$ with $\vertexone\in \abag_j$.
\item \label{sequence-relabelinn-IntroEdgeA}
$(\abag_j, \introedgegeneric{\vertexone}{\vertextwo}(\witnessset_j))$ with distinct $\vertexone,\vertextwo\in \abag_j$.
\item \label{sequence-relabeling-JoinA}
$(\abag_j, \jointype(\witnessset_j,\witnessset_l))$ with $l\in \{0,\dots,i-1\}$ and $\abag_l=\abag_j$.
\end{enumerate}
\end{enumerate}
\end{definition}

Intuitively, a $(\treewidthvalue,\dpcore)$-refutation is a certificate that some inconsistent $(\treewidthvalue,\dpcore)$-state is reachable from the initial $(\treewidthvalue,\dpcore)$-state. This transition-only definition is the version used throughout the search algorithms, where the relevant DP-cores have identity cleaning. For a DP-core with non-identity cleaning, each transition output in Definition~\ref{definition:DPRefutation} has to be replaced by its image under the corresponding cleaning function. It turns out that if $\dpcore$ is a coherent instructive DP-core with identity cleaning, then constructing a $(\treewidthvalue,\dpcore)$-refutation is equivalent to showing that $\dpcoregraphproperty{\dpcore}$ does not contain all graphs of treewidth at most $\treewidthvalue$.

\begin{theorem}[\cite{de2023width}]
\label{theorem:EquivalenceRefutation}
Let $\dpcore$ be a coherent instructive DP-core with identity cleaning. Then there is a $(\treewidthvalue,\dpcore)$-refutation if and only if some graph of treewidth at most $\treewidthvalue$ does not belong to the graph property $\dpcoregraphproperty{\dpcore}$.
\end{theorem}

\section{Example: An Instructive DP-Core for Chromatic Number at Most \texorpdfstring{$\numbercolors$}{r}}
\label{section:ChromaticNumberAtMost}

Let $S$ be a finite set, and $\numbercolors\in \Nplus$. An $\numbercolors$-partition of $S$ is a partition of $S$ with at most $\numbercolors$ cells.
Let $\agraph$ be a graph. We say that $\agraph$ is $\numbercolors$-colorable if there is an $\numbercolors$-partition of
$\vertexset{\agraph}$ such that for each edge $\anedge\in \edgeset{\agraph}$, the endpoints of $\anedge$ belong to distinct cells.
We use the following notation for this property and its DP-core:
\[
\begin{array}{rcl}
\text{property} &:& \chromaticNumberAtMostProperty{\numbercolors},\\
\text{DP-core} &:& \chromaticnumberAtMostCore{\numbercolors}.
\end{array}
\]
We start by defining the corresponding local witnesses for each $\treewidthvalue\in \N$.

\begin{definition}
\label{definition:ChromaticLocalWitness}
Let $\treewidthvalue\in \N$. A $\chromaticnumberAtMostCore{\numbercolors}[\treewidthvalue]$ local witness is any $\numbercolors$-partition of a subset of $[\treewidthvalue+1]$.
\end{definition}



\begin{definition}
\label{definition:ColorableCore}
Let $\numbercolors\in \Nplus$. We let $\chromaticnumberAtMostCore{\numbercolors}$ be the instructive DP-core $\dpcore$ specified below. For each $\treewidthvalue\in \N$,
we define $\chromaticnumberAtMostCore{\numbercolors}[\treewidthvalue] = \dpcore[\treewidthvalue]$. We let $\vertexone,\vertextwo\in [\treewidthvalue+1]$,
$\awitness$ and $\awitness'$ be  $\chromaticnumberAtMostCore{\numbercolors}[\treewidthvalue]$ local witnesses, and $\witnessset$ be a set of such local witnesses.

\begin{enumerate}
\item \label{definition:ColorableCore-all-witness} $\dpcore[\treewidthvalue].\allwitnesses = \{\awitness: \awitness \text{ is a } \chromaticnumberAtMostCore{\numbercolors}[\treewidthvalue]  \text{ local witness }\}$.
\item \label{definition:ColorableCore-initial} $\dpcore[\treewidthvalue].\initialsetgenericcore = \{\emptyset\}$.
\item \label{definition:ColorableCore-intro-vertex}
The transition $\dpcore[\treewidthvalue].\introvertextype{\vertexone}(\awitness)$ is defined by the following cases.
\begin{enumerate}
\item If $\vertexone\in \bigcup_{p\in\awitness}p$, then
\[
\dpcore[\treewidthvalue].\introvertextype{\vertexone}(\awitness)=\emptyset.
\]
\item If $\vertexone\notin \bigcup_{p\in\awitness}p$ and $|\awitness|=\numbercolors$, then
\[
\dpcore[\treewidthvalue].\introvertextype{\vertexone}(\awitness)
=
\{(\awitness\setminus \{p\}) \cup \{p\cup\{\vertexone\}\} : p \in \awitness \}.
\]
\item If $\vertexone\notin \bigcup_{p\in\awitness}p$ and $|\awitness|<\numbercolors$, then
\[
\begin{aligned}
\dpcore[\treewidthvalue].\introvertextype{\vertexone}(\awitness)
= {}& \{\awitness\cup\{\{\vertexone\}\}\}\\
&{}\cup
\{(\awitness\setminus \{p\}) \cup \{p\cup\{\vertexone\}\} : p \in \awitness \}.
\end{aligned}
\]
\end{enumerate}
\item \label{definition:ColorableCore-forget-vertex} $\dpcore[\treewidthvalue].\forgetvertextype{\vertexone}(\awitness) =
\begin{cases}
\{\{p\setminus\{\vertexone\}\;:\; p\in \awitness \} \backslash \{\emptyset\}\} & \quad \text{if }\vertexone\in \bigcup_{p\in\awitness}p,\\
\emptyset & \quad \text{otherwise}
\end{cases}$
\footnote{There is at most one cell containing $\vertexone$. If this cell is a singleton, the whole cell is deleted from $\awitness$. If no cell contains $\vertexone$, the implementation returns no witness.}.
\item \label{definition:ColorableCore-intro-edge}$\dpcore[\treewidthvalue].\introedgegeneric{\vertexone}{\vertextwo}(\awitness) =
\begin{cases}
\{\awitness\} & \quad \text{if } \vertexone \text{ and }\vertextwo \text{ are not in a same cell,}\\
\emptyset & \quad \mbox{Otherwise.}
\end{cases}$
\item \label{definition:ColorableCore-join} $\dpcore[\treewidthvalue].\joingenericcore(\awitness,\awitness')=
\begin{cases}
\{\awitness\} &\quad \text{if }\awitness = \awitness' \text{,} \\
\emptyset &\quad \text{Otherwise.}
\end{cases}$
\item $\dpcore[\treewidthvalue].\finalwitnessgenericcore(\awitness)=1$ for every $\awitness\in \dpcore[\treewidthvalue].\allwitnesses$.

\item $\dpcore[\treewidthvalue].\cleaningfunctioncore(\witnessset)=\witnessset$ for every
$\witnessset \subseteq \dpcore[\treewidthvalue].\allwitnesses$.
\item $\dpcore[\treewidthvalue].\invariantCore(\witnessset)=
\begin{cases}
1 & \quad \text{ if } \witnessset \text{ has a final witness,} \\
0 & \quad \text{ Otherwise.}
\end{cases}
$
\end{enumerate}
\end{definition}

Next, we define a predicate relating $\treewidthvalue$-instructive tree decompositions
with local witnesses.

\begin{definition}
\label{definition:ColorableCorePredicate}
We let $\predicateColorable{\numbercolors}{\treewidthvalue}$ be the predicate that is true on a pair $(\abstractdecomposition,\awitness) \in \allabstractdecompositionstreewidth{\treewidthvalue}\times \chromaticnumberAtMostCore{\numbercolors}[\treewidthvalue].\allwitnesses$
if and only if the following conditions are satisfied.
\begin{enumerate}
    \item $\bigcup\limits_{\cell \in \awitness} \cell= \domain(\topmap{\abstractdecomposition})$.
\item There is a proper $\numbercolors$-coloring $\alpha$ of $\decompositiongraph{\abstractdecomposition}$, viewed as an $\numbercolors$-partition of $\vertexset{\decompositiongraph{\abstractdecomposition}}$, such that for every $\vertexone,\vertextwo \in \topbag{\abstractdecomposition}$,
$\topmap{\abstractdecomposition}(\vertexone)$ and $\topmap{\abstractdecomposition}(\vertextwo)$ belong to the same cell in $\alpha$ if and only if $\vertexone$ and $\vertextwo$ belong to the same cell in $\awitness$.
%
\end{enumerate}
\end{definition}

\begin{proposition}
\label{proposition:PredicateColorable}
For each $\abstractdecomposition\in \allabstractdecompositionstreewidth{\treewidthvalue}$
and each local witness $\awitness$,
\[
\begin{aligned}
&\awitness \in \dynamizationfunction{\cnc}{\treewidthvalue}(\abstractdecomposition)\\
&\Longleftrightarrow\quad
\cnp(\abstractdecomposition,\awitness)=\truevalue.
\end{aligned}
\]
\end{proposition}
\begin{proof}
We prove the claim by induction on $\abstractdecomposition$. For a proper $\numbercolors$-coloring $\alpha$ of $\decompositiongraph{\abstractdecomposition}$, write $\alpha|_{\topbag{\abstractdecomposition}}$ for the partition of $\topbag{\abstractdecomposition}$ in which two active labels are in the same cell exactly when their current vertices receive the same color. Thus $\cnp(\abstractdecomposition,\awitness)$ states precisely that $\awitness=\alpha|_{\topbag{\abstractdecomposition}}$ for some proper $\numbercolors$-coloring $\alpha$.

The leaf case is immediate: the only dynamized witness is the empty partition, and the empty graph has the empty coloring. For an introduce-vertex node $\introvertextype{\vertexone}(\sigmaabstractdecomposition)$, the new vertex is isolated and $\vertexone\notin\topbag{\sigmaabstractdecomposition}$ by the legality of instructive decompositions, so the guard in the transition is never triggered on a reachable predecessor witness. Starting from a coloring for a child witness, the transition either inserts $\vertexone$ into an existing color cell or, when fewer than $\numbercolors$ cells are used, creates a new singleton color cell; both choices exactly describe all possible restrictions of proper colorings after adding the isolated vertex. Conversely, any proper coloring of the parent restricts to a proper coloring of the child, and deleting $\vertexone$ from its active color cell gives a predecessor witness generated by the introduce-vertex transition.

For a forget-vertex node $\forgetvertextype{\vertexone}(\sigmaabstractdecomposition)$, the graph is unchanged and the active interface only removes $\vertexone$. Hence the transition, which deletes $\vertexone$ from its cell and removes an empty cell if necessary, exactly matches restriction of the same proper coloring to the smaller active bag. The converse uses the same coloring on the unchanged graph and re-inserts $\vertexone$ into the unique color cell determined by its color, or into a singleton cell if no remaining active label has that color.

For an introduce-edge node $\introedgegeneric{\vertexone}{\vertextwo}(\sigmaabstractdecomposition)$, the graph is obtained by adding the edge between the two current active vertices. A child coloring remains proper after this edge is added if and only if $\vertexone$ and $\vertextwo$ are in distinct cells of the active partition, which is exactly the transition condition. The converse follows because every proper coloring of the parent is also a proper coloring of the child and must color the two endpoints differently.

Finally consider a join node $\joingenericcore(\sigmaabstractdecomposition_1,\sigmaabstractdecomposition_2)$. The join transition keeps a witness exactly when the two child witnesses are equal. If both children have the same active partition $\awitness$, choose proper colorings for the two children witnessing this partition. After naming the colors by $[\numbercolors]$, the color-name correspondence on the active cells is a partial bijection, and we extend it to a permutation of $[\numbercolors]$; applying this permutation to the second coloring makes the two colorings agree on the identified active vertices. Their union is then a proper coloring of the joined graph: the two colorings agree on the identified bag vertices, and Lemma~\ref{lemma:join-no-cross-edges} rules out edges between non-interface vertices coming from different children. This union coloring induces $\awitness$ on the root bag. Conversely, any proper coloring of the joined graph restricts to proper colorings of both child graphs, and both restrictions induce the same active partition at the common bag. By the induction hypothesis, the equal child witnesses are present, and the join transition produces $\awitness$.
\end{proof}

The next corollary records the accepted-language consequence of
Proposition~\ref{proposition:PredicateColorable}.

\begin{corollary}
\label{corollary:CorrectnessColorability}
Let $\abstractdecomposition$ be a $\treewidthvalue$-instructive tree
decomposition. Then $\decompositiongraph{\abstractdecomposition}$
is $\numbercolors$-colorable if and only if $\abstractdecomposition\in
\accepteddecompositions{\chromaticnumberAtMostCore{\numbercolors}[\treewidthvalue]}$.
\end{corollary}
\begin{proof}
All local witnesses of $\chromaticnumberAtMostCore{\numbercolors}[\treewidthvalue]$ are final. Thus $\abstractdecomposition$ is accepted exactly when its dynamization contains some witness. By Proposition~\ref{proposition:PredicateColorable}, this is equivalent to the existence of a proper $\numbercolors$-coloring of $\decompositiongraph{\abstractdecomposition}$ whose restriction to the active bag induces that witness.
\end{proof}

Since a $\chromaticnumberAtMostCore{\numbercolors}[\treewidthvalue]$ local witness is an $\numbercolors$-partition of a subset of  $[\treewidthvalue+1]$, we can represent such a partition using $O(\treewidthvalue\cdot\log(\numbercolors+1))$ bits.

\begin{observation}
\label{observation:ComplexityMeasuresColorability}
$\chromaticnumberAtMostCore{\numbercolors}[\treewidthvalue]$ has bit-length
$O(\treewidthvalue\cdot\log(\numbercolors+1))$.
\end{observation}

\section{Width-Based ATP with Symmetry Breaking}
\label{WBATPSymmetryBreaking}
\label{sec:canonization}
In this section, we introduce our main technical result. More specifically, we introduce a width-based automated deduction algorithm endowed with a symmetry-breaking procedure that allows us to remove redundant states during the search for a counterexample to a given conjecture. At the core of our technique lies the notion of a {\em witness action}.
Intuitively, functions that satisfy the axioms of a witness action can be used to define permuted versions of local witnesses generated by a DP-core. This allows us to define the notion of a canonical form of a state generated during the search process. Instead of keeping track of all inferred states, we only keep their canonical forms. This leads to a significant reduction of the search space because states with the same canonical form are identified. Our main theorem (Theorem~\ref{theorem:ProvabilityPreserving}) states that this process preserves provability.

\subsection{Relabelings and Witness Actions}
A DP-state is indexed by an interface bag $\abag \subseteq [\treewidthvalue+1]$ and a set of local witnesses $\witnessset$. The labels in $[\treewidthvalue+1]$ are abstract names: permuting them does not change the underlying partial graph represented by an instructive tree decomposition prefix. Consequently, a naive search explores many states that are identical up to relabeling. State canonization quotients the search space by these relabelings by mapping each state to a canonical representative, so that each equivalence class is explored once.

A \emph{$k$-relabeling function} is an injective function $\relabelingfunction: \abag \rightarrow [\treewidthvalue+1]$, where $\abag \subseteq [\treewidthvalue+1]$. We let $\relabelclass_\treewidthvalue$ denote the set of all such functions.
We use injections (rather than permutations) because the active bag size varies across the decomposition; thus relabelings behave like partial permutations whose domain is the current bag.
%
Given a $\treewidthvalue$-relabeling function
$\relabelingfunction\in \relabelclass_{\treewidthvalue}$, and a subset $\abag\subseteq [\treewidthvalue+1]$, we let $\relabelingfunction(\abag) = \{\relabelingfunction(\vertexone)\;:\; \vertexone \in \abag\}$ be the image of $\abag$ under $\relabelingfunction$. Next, we define the notion of a witness action.

For relabeling functions $f,g\in \relabelclass_{\treewidthvalue}$, we write $f\circ g$ for the partial function with domain
\[
\{u\in \domain(g) : g(u)\in \domain(f)\}
\]
given by $(f\circ g)(u)\defeq f(g(u))$. In particular, if $g(\domain(g))\subseteq \domain(f)$ then $f\circ g$ is again a $\treewidthvalue$-relabeling function with domain $\domain(g)$.

\begin{lemma}[Injective extension]
\label{lemma:injective-extension}
Let $\abag\subseteq [\treewidthvalue+1]$ and let $\vertexone\in \abag$.
If $h:\abag\setminus\{\vertexone\}\to [\treewidthvalue+1]$ is injective, then there exists an injective extension $\hat h:\abag\to[\treewidthvalue+1]$ of $h$, i.e., $\hat h|_{\abag\setminus\{\vertexone\}}=h$.
\end{lemma}
\begin{proof}
Since $h$ is injective, its image $h(\abag\setminus\{\vertexone\})$ has size $|\abag|-1\le \treewidthvalue$.
Therefore there exists a label $x\in [\treewidthvalue+1]\setminus h(\abag\setminus\{\vertexone\})$.
Define $\hat h(\vertexone)\defeq x$ and $\hat h(\vertextwo)\defeq h(\vertextwo)$ for every $\vertextwo\in \abag\setminus\{\vertexone\}$.
Then $\hat h$ is injective and extends $h$.
\end{proof}

\begin{lemma}[Injective image containment]
\label{lemma:injective-image-containment}
Let $f\in \relabelclass_{\treewidthvalue}$ and let $X,Y\subseteq \domain(f)$. If $f(X)\subseteq f(Y)$, then $X\subseteq Y$.
\end{lemma}
\begin{proof}
Let $x\in X$. Then $f(x)\in f(Y)$, so there exists $y\in Y$ such that $f(y)=f(x)$. Since $f$ is injective, we have $y=x$, and therefore $x\in Y$. Hence $X\subseteq Y$.
\end{proof}

Formally, the set $\dpcore[\treewidthvalue].\allwitnesses$ is a decidable subset of $\{0,1\}^\ast$, i.e., witnesses are binary strings. For readability, we freely identify each witness with the mathematical object it encodes (e.g., a partition, a map, a set), and we view the DP-core operations and witness actions as operating on these objects via a fixed encoding.

Each witness $\awitness\in \dpcore[\treewidthvalue].\allwitnesses$ encodes a finite object whose representation mentions a finite set of interface labels in $[\treewidthvalue+1]$.
We write $\mathrm{Lbl}_{\treewidthvalue}(\awitness)\subseteq [\treewidthvalue+1]$ for the set of labels occurring in $\awitness$.
For a finite witness set $\witnessset$ we define
\[
\mathrm{Lbl}_{\treewidthvalue}(\witnessset)\defeq \bigcup_{\awitness\in \witnessset}\mathrm{Lbl}_{\treewidthvalue}(\awitness).
\]
A $(\treewidthvalue,\dpcore)$-state $(\abag,\witnessset)$ is \emph{well-formed} if $\mathrm{Lbl}_{\treewidthvalue}(\witnessset)\subseteq \abag$.
We will only apply relabelings to well-formed states.
To ensure that well-formedness is preserved by reachability and by dynamization, we restrict attention to interface-respecting DP-cores (Definition~\ref{definition:interface-respecting-dpcore}); under this mild discipline, all refutations and all dynamization states are well-formed (Lemma~\ref{lemma:well-formedness-preserved}).

\begin{definition}[Interface-respecting DP-cores]
\label{definition:interface-respecting-dpcore}
Let $\dpcore$ be a DP-core and let $\treewidthvalue\in \N$. We say that $\dpcore[\treewidthvalue]$ is \emph{interface-respecting} if the following hold for every $\abag\subseteq [\treewidthvalue+1]$ and all $\witnessset,\witnessset_1,\witnessset_2\subseteq \dpcore[\treewidthvalue].\allwitnesses$.
\begin{enumerate}
\setlength\itemsep{0.25em}
\item \label{interface-respecting-initial}
$\mathrm{Lbl}_{\treewidthvalue}(\dpcore[\treewidthvalue].\initialsetgeneric)\subseteq \emptyset$.
\item \label{interface-respecting-introvertex}
If $\mathrm{Lbl}_{\treewidthvalue}(\witnessset)\subseteq \abag$ and $\vertexone\notin \abag$, then
$\mathrm{Lbl}_{\treewidthvalue}(\introvertextype{\vertexone}(\witnessset))\subseteq \abag\cup\{\vertexone\}$.
\item \label{interface-respecting-forgetvertex}
If $\mathrm{Lbl}_{\treewidthvalue}(\witnessset)\subseteq \abag$ and $\vertexone\in \abag$, then
$\mathrm{Lbl}_{\treewidthvalue}(\forgetvertextype{\vertexone}(\witnessset))\subseteq \abag\setminus\{\vertexone\}$.
\item \label{interface-respecting-introedge}
If $\mathrm{Lbl}_{\treewidthvalue}(\witnessset)\subseteq \abag$ and $\vertexone,\vertextwo\in \abag$, then
$\mathrm{Lbl}_{\treewidthvalue}(\introedgegeneric{\vertexone}{\vertextwo}(\witnessset))\subseteq \abag$.
\item \label{interface-respecting-join}
If $\mathrm{Lbl}_{\treewidthvalue}(\witnessset_1)\subseteq \abag$ and $\mathrm{Lbl}_{\treewidthvalue}(\witnessset_2)\subseteq \abag$, then
$\mathrm{Lbl}_{\treewidthvalue}(\jointype(\witnessset_1,\witnessset_2))\subseteq \abag$.
\item \label{interface-respecting-clean}
If $\mathrm{Lbl}_{\treewidthvalue}(\witnessset)\subseteq \abag$, then
$\mathrm{Lbl}_{\treewidthvalue}(\dpcore[\treewidthvalue].\cleaningfunctioncore(\witnessset))\subseteq \abag$.
\end{enumerate}
We say that $\dpcore$ is interface-respecting if $\dpcore[\treewidthvalue]$ is interface-respecting for every $\treewidthvalue\in \N$.
\end{definition}

\begin{lemma}[Well-formedness is preserved]
\label{lemma:well-formedness-preserved}
Let $\dpcore$ be an interface-respecting DP-core and let $\treewidthvalue\in \N$.
\begin{enumerate}
\setlength\itemsep{0.25em}
\item Every $(\treewidthvalue,\dpcore)$-refutation (Definition~\ref{definition:DPRefutation}) is well-formed, i.e., each state $(\abag_i,\witnessset_i)$ satisfies $\mathrm{Lbl}_{\treewidthvalue}(\witnessset_i)\subseteq \abag_i$.
\item For every $\treewidthvalue$-instructive tree decomposition $\abstractdecomposition$ and every subterm $\sigmaabstractdecomposition$ of $\abstractdecomposition$, the state
\[
\bigl(\topbag{\sigmaabstractdecomposition},\,\dynamizationfunction{\dpcore}{\treewidthvalue}(\sigmaabstractdecomposition)\bigr)
\]
is well-formed.
\end{enumerate}
\end{lemma}
\begin{proof}
For (1), let $\dprefutation=(\abag_0,\witnessset_0)\dots(\abag_m,\witnessset_m)$ be a $(\treewidthvalue,\dpcore)$-refutation. By Condition~\ref{interface-respecting-initial}, the initial state $(\abag_0,\witnessset_0)=(\emptyset,\dpcore[\treewidthvalue].\initialsetgeneric)$ is well-formed. For each step, if $(\abag_i,\witnessset_i)$ is obtained from earlier state(s) by one of the refutation rules (Definition~\ref{definition:DPRefutation}), then the corresponding interface-respecting condition (Definition~\ref{definition:interface-respecting-dpcore}.\ref{interface-respecting-introvertex}--\ref{interface-respecting-join}) shows that $\mathrm{Lbl}_{\treewidthvalue}(\witnessset_i)\subseteq \abag_i$. Hence all states are well-formed.

	For (2), follow the bottom-up computation of $\dynamizationfunction{\dpcore}{\treewidthvalue}$ along $\abstractdecomposition$ (Definition~\ref{definition:Dynamization}) and use induction on the term structure. The raw transition output has support contained in the new top bag by Conditions~\ref{interface-respecting-introvertex}--\ref{interface-respecting-join}; applying $\dpcore[\treewidthvalue].\cleaningfunctioncore$ preserves this inclusion by Condition~\ref{interface-respecting-clean}. The base case follows from Condition~\ref{interface-respecting-initial} and $\topbag{\leaftype}=\emptyset$.
\end{proof}



\begin{definition}[Witness Action]\label{definition:witnessaction}
Let $\dpcore$ be a DP-core and $k\in \N$. A witness action for $\dpcore[\treewidthvalue]$ is a partial function
$\action_{\dpcore}^{\treewidthvalue}$ from $\relabelclass_\treewidthvalue \times \dpcore[\treewidthvalue].\allwitnesses$ to $\dpcore[\treewidthvalue].\allwitnesses$, defined exactly on the supported pairs
\[
(\relabelingfunction,\awitness)
\quad\text{with}\quad
\mathrm{Lbl}_{\treewidthvalue}(\awitness)\subseteq \domain(\relabelingfunction).
\]
We say that $\action_{\dpcore}^{\treewidthvalue}$ is an action for $\dpcore[\treewidthvalue]$ if
the following conditions are satisfied for each $\relabelingfunction,\relabelingfunction'\in \relabelclass_{\treewidthvalue}$, and
each $\awitness \in \dpcore[\treewidthvalue].\allwitnesses$, whenever the displayed expressions are supported and hence defined.
Since $\relabelingfunction$ is injective, $\relabelingfunction^{-1}$ denotes the inverse map on its image $\relabelingfunction(\domain(\relabelingfunction))$ (so $\relabelingfunction^{-1}\in \relabelclass_\treewidthvalue$ with domain $\relabelingfunction(\domain(\relabelingfunction))$).
\begin{enumerate}
	    \item \label{condition:relabeldpcore-condition-one}
	    Finality invariance:
	    \[
	    \dpcore[\treewidthvalue].\finalwitnessgenericcore(\awitness)=1
	    \quad \text{if and only if}\quad
	    \dpcore[\treewidthvalue].\finalwitnessgenericcore(\action_{\dpcore}^{\treewidthvalue}(\relabelingfunction,\awitness))=1.
	    \]
	    \item \label{condition:relabeldpcore-condition-support-transport}
	    Label support transport:
	    \[
	    \mathrm{Lbl}_{\treewidthvalue}(\action_{\dpcore}^{\treewidthvalue}(\relabelingfunction,\awitness))
	    = \relabelingfunction(\mathrm{Lbl}_{\treewidthvalue}(\awitness)).
	    \]
	    \item Inverse relabeling: \label{condition:relabeldpcore-condition-two} \[\action_{\dpcore}^{\treewidthvalue}(\relabelingfunction^{-1},\action_{\dpcore}^{\treewidthvalue}(\relabelingfunction,\awitness)) = \awitness.\]
	    \item Composition: \label{condition:relabeldpcore-condition-three} \[\action_{\dpcore}^{\treewidthvalue}(\relabelingfunction\circ \relabelingfunction',\awitness) =
	    \action_{\dpcore}^{\treewidthvalue}(\relabelingfunction,\action_{\dpcore}^{\treewidthvalue}(\relabelingfunction',\awitness)).\]
	    \item \label{condition:relabeldpcore-condition-extension-invariance}
	    Extension invariance:
	    if $\relabelingfunction'\in \relabelclass_{\treewidthvalue}$ satisfies $\domain(\relabelingfunction)\subseteq \domain(\relabelingfunction')$ and $\relabelingfunction'|_{\domain(\relabelingfunction)}=\relabelingfunction$, then
		    \[
		    \action_{\dpcore}^{\treewidthvalue}(\relabelingfunction',\awitness)=\action_{\dpcore}^{\treewidthvalue}(\relabelingfunction,\awitness).
		    \]

We extend $\action_{\dpcore}^{\treewidthvalue}$ to witness sets pointwise whenever the set is supported by the relabeling: for each $\relabelingfunction\in \relabelclass_{\treewidthvalue}$ and each $\witnessset\subseteq \dpcore[\treewidthvalue].\allwitnesses$ with $\mathrm{Lbl}_{\treewidthvalue}(\witnessset)\subseteq\domain(\relabelingfunction)$, we let
\[
	\action_{\dpcore}^{\treewidthvalue}(\relabelingfunction,\witnessset)\defeq
	\{\action_{\dpcore}^{\treewidthvalue}(\relabelingfunction,\awitness)\;:\; \awitness\in \witnessset\}.
	\]
	We use the same notation for the witness-level action and its pointwise extension, as the intended meaning is clear from the type of the second argument.
	Accordingly, the transition equivariance axioms below are equalities of finite witness sets whenever a DP-transition produces a finite set; the action on such an output is the pointwise extension just defined.

			    \item DP-core Consistency Axioms:
				    \begin{enumerate}
				    \item \label{condition:relabeldpcore-condition-four}
				    For each $\vertexone\in \domain(\relabelingfunction)$,
				    \[
				    \begin{aligned}
				    \action_{\dpcore}^{\treewidthvalue}(\relabelingfunction,\introvertextype{\vertexone}(\awitness))
				    &=
				    \introvertextype{\relabelingfunction(\vertexone)}\bigl(\action_{\dpcore}^{\treewidthvalue}(\relabelingfunction,\awitness)\bigr).
				    \end{aligned}
				    \]

		    \item \label{condition:relabeldpcore-condition-five}
				    For each $\vertexone\in \domain(\relabelingfunction)$,
				    \[
				    \begin{aligned}
				    \action_{\dpcore}^{\treewidthvalue}(\relabelingfunction,\forgetvertextype{\vertexone}(\awitness))
				    &=
				    \forgetvertextype{\relabelingfunction(\vertexone)}\bigl(\action_{\dpcore}^{\treewidthvalue}(\relabelingfunction,\awitness)\bigr).
				    \end{aligned}
				    \]
		    \item \label{condition:relabeldpcore-condition-six}
				    For each two distinct $\vertexone,\vertextwo\in \domain(\relabelingfunction)$,
				    \[
				    \begin{aligned}
				    \action_{\dpcore}^{\treewidthvalue}(\relabelingfunction,\introedgegeneric{\vertexone}{\vertextwo}(\awitness))
				    &=
				    \introedgegeneric{\relabelingfunction(\vertexone)}{\relabelingfunction(\vertextwo)}\\
				    &\qquad\bigl(\action_{\dpcore}^{\treewidthvalue}(\relabelingfunction,\awitness)\bigr).
				    \end{aligned}
				    \]
	    \item \label{condition:relabeldpcore-condition-seven}
				    For each $\awitness'\in  \dpcore[\treewidthvalue].\allwitnesses$ with $\mathrm{Lbl}_{\treewidthvalue}(\awitness')\subseteq \domain(\relabelingfunction)$,
				    \[
				    \begin{aligned}
				    \action_{\dpcore}^{\treewidthvalue}(\relabelingfunction,\jointype(\awitness,\awitness'))
				    &=
				    \jointype\Bigl(\action_{\dpcore}^{\treewidthvalue}(\relabelingfunction,\awitness),\\
				    &\qquad\action_{\dpcore}^{\treewidthvalue}(\relabelingfunction,\awitness')\Bigr).
				    \end{aligned}
				    \]
			    \end{enumerate}
\end{enumerate}
\end{definition}

\begin{lemma}
\label{lemma:witness-action-finality-sets}
Let $\dpcore$ be a DP-core, let $\action_{\dpcore}^{\treewidthvalue}$ be a witness action for $\dpcore[\treewidthvalue]$, and let $\relabelingfunction\in \relabelclass_{\treewidthvalue}$. For every finite witness set $\witnessset\in \finitepowerset{\dpcore[\treewidthvalue].\allwitnesses}$ with $\mathrm{Lbl}_{\treewidthvalue}(\witnessset)\subseteq \domain(\relabelingfunction)$, the set $\witnessset$ contains a final witness if and only if $\action_{\dpcore}^{\treewidthvalue}(\relabelingfunction,\witnessset)$ contains a final witness.
\end{lemma}
\begin{proof}
By definition of the pointwise extension, a witness $\awitness$ belongs to $\action_{\dpcore}^{\treewidthvalue}(\relabelingfunction,\witnessset)$ if and only if $\awitness=\action_{\dpcore}^{\treewidthvalue}(\relabelingfunction,\awitness')$ for some $\awitness'\in \witnessset$.
Finality invariance at the witness level (Condition~\ref{condition:relabeldpcore-condition-one}) implies that $\awitness'$ is final if and only if $\action_{\dpcore}^{\treewidthvalue}(\relabelingfunction,\awitness')$ is final. The claim follows.
\end{proof}

\begin{lemma}[Identity relabeling]
\label{lemma:witness-action-identity}
Let $\dpcore$ be a DP-core, let $\action_{\dpcore}^{\treewidthvalue}$ be a witness action for $\dpcore[\treewidthvalue]$, and let $\abag\subseteq [\treewidthvalue+1]$. Let $\mathrm{id}_{\abag}:\abag\to [\treewidthvalue+1]$ be the identity injection, i.e., $\mathrm{id}_{\abag}(u)=u$ for all $u\in \abag$. If $\awitness\in \dpcore[\treewidthvalue].\allwitnesses$ satisfies $\mathrm{Lbl}_{\treewidthvalue}(\awitness)\subseteq \abag$, then $\action_{\dpcore}^{\treewidthvalue}(\mathrm{id}_{\abag},\awitness)=\awitness$. Consequently, for each $\witnessset\subseteq \dpcore[\treewidthvalue].\allwitnesses$ with $\mathrm{Lbl}_{\treewidthvalue}(\witnessset)\subseteq \abag$, we have $\action_{\dpcore}^{\treewidthvalue}(\mathrm{id}_{\abag},\witnessset)=\witnessset$.
\end{lemma}
\begin{proof}
Let $\awitness$ be as stated. By composition (Condition~\ref{condition:relabeldpcore-condition-three}) with $\relabelingfunction=\relabelingfunction'=\mathrm{id}_{\abag}$,
\[
\action_{\dpcore}^{\treewidthvalue}(\mathrm{id}_{\abag},\awitness)=\action_{\dpcore}^{\treewidthvalue}(\mathrm{id}_{\abag},\action_{\dpcore}^{\treewidthvalue}(\mathrm{id}_{\abag},\awitness)).
\]
By inverse relabeling (Condition~\ref{condition:relabeldpcore-condition-two}) and $\mathrm{id}_{\abag}^{-1}=\mathrm{id}_{\abag}$, the right-hand side equals $\awitness$. The statement for witness sets follows by the pointwise extension of $\action_{\dpcore}^{\treewidthvalue}$.
\end{proof}

Intuitively, a witness action lifts an injective relabeling of the interface to the witness domain. Condition~\ref{condition:relabeldpcore-condition-one} ensures that finality (and hence acceptance) is invariant under relabeling. Condition~\ref{condition:relabeldpcore-condition-support-transport} specifies how label support is transported: the labels mentioned by a witness are mapped by the relabeling. Conditions~\ref{condition:relabeldpcore-condition-two}--\ref{condition:relabeldpcore-condition-three} state that relabeling is invertible on its image and respects composition, so $\action_{\dpcore}^{\treewidthvalue}$ behaves like an action. Condition~\ref{condition:relabeldpcore-condition-extension-invariance} (extension invariance) guarantees that $\action_{\dpcore}^{\treewidthvalue}(\relabelingfunction,\awitness)$ depends only on the values of $\relabelingfunction$ on the labels that actually occur in $\awitness$. Finally, Conditions~\ref{condition:relabeldpcore-condition-four}--\ref{condition:relabeldpcore-condition-seven} express equivariance with respect to the DP transitions: relabeling commutes with $\introvertextype{\cdot}$, $\forgetvertextype{\cdot}$, $\introedgetype{\cdot}{\cdot}$, and $\jointype$.

As an example, consider the DP-core $\chromaticnumberAtMostCore{\numbercolors}$ from Section~\ref{section:ChromaticNumberAtMost}, where witnesses are $\numbercolors$-partitions of subsets of $[\treewidthvalue+1]$.
Let $\action_{\chromaticnumberAtMostCore{\numbercolors}}^\treewidthvalue$ be the action defined as follows for each $\relabelingfunction \in \relabelclass_\treewidthvalue$ and each witness $\awitness$:
\[
\action_{\chromaticnumberAtMostCore{\numbercolors}}^\treewidthvalue(\relabelingfunction, \awitness) = \left\{ \relabelingfunction(c) \mid c \in \awitness \right\},
\]
where $\relabelingfunction(c) = \{ \relabelingfunction(\vertexone) \mid \vertexone \in c \}$ for a cell $c\in \awitness$.

\begin{lemma}
\label{lemma:witness-action-colorability}
For each $\treewidthvalue\in \N$ and each $\numbercolors\in \Nplus$, the function $\action_{\chromaticnumberAtMostCore{\numbercolors}}^\treewidthvalue$ defined above is a witness action for $\chromaticnumberAtMostCore{\numbercolors}[\treewidthvalue]$ in the sense of Definition~\ref{definition:witnessaction}.
\end{lemma}
\begin{proof}
Let $\relabelingfunction\in \relabelclass_\treewidthvalue$ and let $\awitness$ be a $\numbercolors$-partition with $\mathrm{Lbl}_{\treewidthvalue}(\awitness)\subseteq \domain(\relabelingfunction)$. Since $\relabelingfunction$ is injective, the images of the cells of $\awitness$ remain nonempty and pairwise disjoint; hence the relabeled witness is again a $\numbercolors$-partition.

Finality invariance holds because, by definition of $\chromaticnumberAtMostCore{\numbercolors}[\treewidthvalue]$, every local witness is final. For label support transport, note that the label support of a partition is its carrier, so
\[
\mathrm{Lbl}_{\treewidthvalue}\big(\action_{\chromaticnumberAtMostCore{\numbercolors}}^\treewidthvalue(\relabelingfunction,\awitness)\big)
= \bigcup_{c\in \awitness}\relabelingfunction(c)
= \relabelingfunction\Big(\bigcup_{c\in \awitness}c\Big)
= \relabelingfunction\big(\mathrm{Lbl}_{\treewidthvalue}(\awitness)\big).
\]
	Inverse relabeling and composition follow from $\relabelingfunction^{-1}(\relabelingfunction(c))=c$ and $(\relabelingfunction\circ \relabelingfunction')(c)=\relabelingfunction(\relabelingfunction'(c))$ for each cell $c\in \awitness$, and extension invariance holds because $\awitness$ mentions only labels in $\domain(\relabelingfunction)$.

	For the DP-core consistency axioms, write $\Phi(\cdot)\defeq \action_{\chromaticnumberAtMostCore{\numbercolors}}^\treewidthvalue(\relabelingfunction,\cdot)$ and check each transition.
	\begin{itemize}
\item \textbf{Introduce vertex.} The condition $\vertexone\in\bigcup_{c\in\awitness}c$ is equivalent to $\relabelingfunction(\vertexone)\in\bigcup_{c\in\Phi(\awitness)}c$, so the new guard is preserved by relabeling. In the nonempty-output case, each witness in $\introvertextype{\vertexone}(\awitness)$ is obtained either by inserting $\vertexone$ into some existing cell $c\in\awitness$, or (if $|\awitness|<\numbercolors$) by adding a new singleton cell $\{\vertexone\}$. Applying $\relabelingfunction$ to labels sends these constructions exactly to the corresponding ones for $\introvertextype{\relabelingfunction(\vertexone)}(\action_{\chromaticnumberAtMostCore{\numbercolors}}^\treewidthvalue(\relabelingfunction,\awitness))$.
\item \textbf{Forget vertex.} The condition $\vertexone\in\bigcup_{c\in\awitness}c$ is equivalent to $\relabelingfunction(\vertexone)\in\bigcup_{c\in\Phi(\awitness)}c$. In the nonempty-output case, removing $\vertexone$ from each cell and deleting empty cells commutes with applying $\relabelingfunction$, and yields $\forgetvertextype{\relabelingfunction(\vertexone)}(\action_{\chromaticnumberAtMostCore{\numbercolors}}^\treewidthvalue(\relabelingfunction,\awitness))$.
\item \textbf{Introduce edge.} For $\vertexone,\vertextwo\in \domain(\relabelingfunction)$, the labels $\vertexone$ and $\vertextwo$ lie in the same cell of $\awitness$ if and only if $\relabelingfunction(\vertexone)$ and $\relabelingfunction(\vertextwo)$ lie in the same cell of $\action_{\chromaticnumberAtMostCore{\numbercolors}}^\treewidthvalue(\relabelingfunction,\awitness)$. Hence the $\introedgegeneric{\vertexone}{\vertextwo}$ filter condition is preserved under relabeling.
	\item \textbf{Join.} The join transition returns $\{\awitness\}$ if $\awitness=\awitness'$ and $\emptyset$ otherwise. Since $\Phi$ is injective on partitions, we have $\awitness=\awitness'$ if and only if $\Phi(\awitness)=\Phi(\awitness')$. Therefore relabeling commutes with join.
	\end{itemize}
Therefore $\action_{\chromaticnumberAtMostCore{\numbercolors}}^\treewidthvalue$ satisfies Conditions~\ref{condition:relabeldpcore-condition-four}--\ref{condition:relabeldpcore-condition-seven}.
\end{proof}

\subsection{Relabeled Refutations}

A standard $(\treewidthvalue,\dpcore)$-refutation (Definition~\ref{definition:DPRefutation}) shows that an inconsistent state is reachable from the initial state by DP-operations. In the presence of symmetry, we also want to allow intermediate states to be renamed (canonized) on the fly. Relabeled refutations formalize exactly this: each inference step applies a DP-operation and then (optionally, via the identity relabeling) applies an admissible relabeling transported to witnesses by the action.



We next define relabeled refutations (Definition~\ref{definition:RelabeledRefutation}), which allow intermediate states to be renamed via relabelings transported to witnesses by the action.
In the join clause, we additionally allow an internal permutation of the shared interface bag to align the second witness set with the labeling used by the first before applying $\jointype$.
For a bag $\abag\subseteq [\treewidthvalue+1]$, we write $\mathrm{Perm}(\abag)$ for the set of all bijections $\pi:\abag\to \abag$.

\begin{definition}[$\relabelsequence$-Relabeled Refutation]\label{definition:RelabeledRefutation}
Let $\relabelsequence = (\relabelingfunction_1,\dots,\relabelingfunction_m)$ be a sequence of relabeling functions in $\relabelclass_\treewidthvalue$ and $\dpcore$ be a DP-core. An $\relabelsequence$-relabeled $(\treewidthvalue,\dpcore)$-refutation is a sequence of pairs
$$\dprefutation \equiv (\abag_0,\witnessset_0)(\abag_1,\witnessset_1)\dots (\abag_m,\witnessset_m)$$ satisfying the following conditions:
\begin{enumerate}
\setlength\itemsep{0.5em}
\item \label{sequence-relabeling-condition-one}
$(\abag_0,\witnessset_0) = (\emptyset,\dpcore[\treewidthvalue].\initialsetgeneric)$.
\item \label{sequence-relabeling-condition-two}
$(\abag_m,\witnessset_m)$ is $(\treewidthvalue,\dpcore)$-inconsistent, i.e., $\witnessset_m$ has no final local witness.
	\item \label{sequence-relabeling-condition-wellformed}
	For each $i\in\{0,1,\dots,m\}$, the state $(\abag_i,\witnessset_i)$ is well-formed, i.e., $\witnessset_i\in \finitepowerset{\dpcore[\treewidthvalue].\allwitnesses}$ and $\mathrm{Lbl}_{\treewidthvalue}(\witnessset_i)\subseteq \abag_i$.
					\item \label{sequence-relabeling-condition-three}
					For each $i\in [m]$, there is some $j\in \{0\}\cup[i-1]$ such that $(\abag_i,\witnessset_i)$ is equal to one of the following pairs.
					In each alternative, $\relabelingfunction_i$ is taken with domain equal to the pre-relabeling output bag displayed in the first component: respectively $\abag_j\cup\{\vertexone\}$, $\abag_j\setminus\{\vertexone\}$, $\abag_j$, and $\abag_j$.
				\begin{enumerate}
	\item \label{sequence-relabeling-IntroVertex}
	$(\relabelingfunction_i(\abag_j\cup \{\vertexone\}), \action_{\dpcore}^{\treewidthvalue}(\relabelingfunction_i,\introvertextype{\vertexone}(\witnessset_j)))$
for some $\vertexone\notin \abag_j$.
\item \label{sequence-relabeling-ForgetVertex}
$(\relabelingfunction_i(\abag_j\setminus \{\vertexone\}), \action_{\dpcore}^{\treewidthvalue}(\relabelingfunction_i,\forgetvertextype{\vertexone}(\witnessset_j)))$
for some $\vertexone\in \abag_j$.
\item \label{sequence-relabelinn-IntroEdge}
$(\relabelingfunction_i(\abag_j), \action_{\dpcore}^{\treewidthvalue}(\relabelingfunction_i,\introedgegeneric{\vertexone}{\vertextwo}(\witnessset_j)))$
 for some distinct $\vertexone,\vertextwo\in \abag_j$.
	\item \label{sequence-relabeling-Join}
	$(\relabelingfunction_i(\abag_j), \action_{\dpcore}^{\treewidthvalue}(\relabelingfunction_i,\jointype(\witnessset_j,\action_{\dpcore}^{\treewidthvalue}(\pi,\witnessset_l))))$
	for some $l\in \{0\}\cup[i-1]$ with $\abag_l=\abag_j$ and some $\pi\in \mathrm{Perm}(\abag_j)$.

		\end{enumerate}
		\end{enumerate}
		\end{definition}

	If the sequence $\dprefutation$ in Definition \ref{definition:RelabeledRefutation} satisfies Conditions
	\ref{sequence-relabeling-condition-one}, \ref{sequence-relabeling-condition-wellformed}, and \ref{sequence-relabeling-condition-three}, but not Condition
	\ref{sequence-relabeling-condition-two}, then we say that $\dprefutation$ is a \emph{semi $\relabelsequence$-relabeled $(\treewidthvalue,\dpcore)$-refutation}.
Below, given an instructive tree decomposition $\abstractdecomposition$, we let
$\allsubterms(\abstractdecomposition)$ denote the set of all subterm occurrences of $\abstractdecomposition$ (equivalently, the nodes of the term tree together with the rooted subterm at each node).

An $\relabelsequence$-relabeled refutation generalizes Definition~\ref{definition:DPRefutation} by allowing relabelings and witness actions to be applied at intermediate steps, without changing the existence of refutations. Lemma~\ref{lemma:tree decomposition-existence-sequence-relabel-semi-refutation} below is the technical bridge: it shows that semi-relabeled refutations still arise from the dynamization of some instructive tree decomposition (up to relabeling), and therefore still yield concrete counterexamples. In this section we restrict to DP-cores with identity cleaning, so dynamization coincides with iterating the DP-transitions without an additional cleaning step.

		\begin{restatable}{lemma}{treedecexistssemiref}
		\label{lemma:tree decomposition-existence-sequence-relabel-semi-refutation}
		Let $\dpcore$ be a DP-core with identity cleaning and equipped with a witness action, and let $\relabelsequence=(\relabelingfunction_1,\dots,\relabelingfunction_m)$ be a sequence of $k$-relabeling functions. Let  $\dprefutation =(\abag_0,\witnessset_0)(\abag_1,\witnessset_1)\dots(\abag_m,\witnessset_m)$ be a semi $\relabelsequence$-relabeled $(\treewidthvalue,\dpcore)$-refutation, and $g:\abag_m\to [\treewidthvalue+1]$ be a $\treewidthvalue$-relabeling.
	Then, there is a $\treewidthvalue$-instructive tree decomposition $\abstractdecomposition_{\dprefutation}\in \allabstractdecompositionstreewidth{\treewidthvalue}$ and a function
	$\subtermToBag: \allsubterms(\abstractdecomposition_{\dprefutation}) \to \finitepowerset{\dpcore[\treewidthvalue].\allwitnesses}$
	such that the following conditions are satisfied for each subterm $\abstractdecomposition$ of $\abstractdecomposition_{\dprefutation}$.
\begin{enumerate}
    \item\label{existenceZero} $\topbag{\abstractdecomposition_{\dprefutation}} = g(\abag_m)$.
    \item\label{existenceOne} If $\abstractdecomposition=\abstractdecomposition_{\dprefutation}$, then $\subtermToBag(\abstractdecomposition) = \action_{\dpcore}^{\treewidthvalue}(g,\witnessset_m)$.
    \item\label{existenceTwo} if $\abstractdecomposition = \leaftype$, then
    $\subtermToBag(\abstractdecomposition) = \dpcore[\treewidthvalue].\initialsetgeneric$.
    \item\label{existenceThree} if $\abstractdecomposition = \introvertextype{\vertexone}(\abstractdecomposition_1)$ for some subterm $\abstractdecomposition_1$, then
    $\subtermToBag(\abstractdecomposition) = \introvertextype{\vertexone}(\subtermToBag(\abstractdecomposition_1))$.
    \item\label{existenceFour} if $\abstractdecomposition = \forgetvertextype{\vertexone}(\abstractdecomposition_1)$ for some subterm $\abstractdecomposition_1$, then
    $\subtermToBag(\abstractdecomposition) = \forgetvertextype{\vertexone}(\subtermToBag(\abstractdecomposition_1))$
    \item\label{existenceFive} if $\abstractdecomposition = \introedgegeneric{\vertexone}{\vertextwo}(\abstractdecomposition_1)$ for some subterm $\abstractdecomposition_1$, then
    $\subtermToBag(\abstractdecomposition) = \introedgegeneric{\vertexone}{\vertextwo}(\subtermToBag(\abstractdecomposition_1))$.
    \item\label{existenceSix} if $\abstractdecomposition = \jointype(\abstractdecomposition_1,\abstractdecomposition_2)$ for some subterms $\abstractdecomposition_1$ and $\abstractdecomposition_2$, then $\subtermToBag(\abstractdecomposition) = \jointype(\subtermToBag(\abstractdecomposition_1),\subtermToBag(\abstractdecomposition_2))$.
\end{enumerate}
\end{restatable}

	\begin{proof}
	We prove the lemma by induction on $m$.

Base case ($m=0$): Then $\dprefutation=(\abag_0,\witnessset_0)$ and by Condition~\ref{sequence-relabeling-condition-one} we have $(\abag_0,\witnessset_0)=(\emptyset,\dpcore[\treewidthvalue].\initialsetgeneric)$.
Let $\abstractdecomposition_{\dprefutation}\defeq \leaftype$ and set $\subtermToBag(\leaftype)\defeq \dpcore[\treewidthvalue].\initialsetgeneric$.
Since $\abag_0=\emptyset$, the relabeling $g$ is the unique empty function. For every witness $\awitness$ with $\mathrm{Lbl}_{\treewidthvalue}(\awitness)=\emptyset$, applying Definition~\ref{definition:witnessaction}.\ref{condition:relabeldpcore-condition-two} and Definition~\ref{definition:witnessaction}.\ref{condition:relabeldpcore-condition-three} yields $\action_{\dpcore}^{\treewidthvalue}(g,\awitness)=\awitness$; hence (by pointwise extension) $\action_{\dpcore}^{\treewidthvalue}(g,\witnessset_0)=\witnessset_0$.
Moreover, $\topbag{\leaftype}=\emptyset=g(\abag_0)$, so Condition~\ref{existenceZero} holds. Therefore Conditions~\ref{existenceZero}--\ref{existenceSix} hold.

Induction step: Assume $m\ge 1$ and that the statement holds for all semi-relabeled refutations of length $<m$.
Let $\dprefutation=(\abag_0,\witnessset_0)\dots(\abag_m,\witnessset_m)$ be a semi $\relabelsequence$-relabeled $(\treewidthvalue,\dpcore)$-refutation and fix $g:\abag_m\to [\treewidthvalue+1]$.
By Condition~\ref{sequence-relabeling-condition-three} of Definition~\ref{definition:RelabeledRefutation} (applied to $i=m$), there exists $j\in \{0\}\cup[m-1]$ such that $(\abag_m,\witnessset_m)$ is obtained from $(\abag_j,\witnessset_j)$ by one of the rules (and in the join case also some $l\in \{0\}\cup[m-1]$). We treat the cases.

\begin{enumerate}
\item \textbf{Introduce vertex.}
Assume
\[
(\abag_m,\witnessset_m)
=
\bigl(\relabelingfunction_m(\abag_j\cup\{\vertexone\}),\,\action_{\dpcore}^{\treewidthvalue}(\relabelingfunction_m,\introvertextype{\vertexone}(\witnessset_j))\bigr)
	\]
	for some $\vertexone\notin \abag_j$.
	Let
	\[
	h \defeq g\circ \bigl(\relabelingfunction_m|_{\abag_j\cup\{\vertexone\}}\bigr),
	\]
	which is well-defined because $\relabelingfunction_m(\abag_j\cup\{\vertexone\})=\abag_m=\domain(g)$.
	Apply the induction hypothesis to the prefix $(\abag_0,\witnessset_0)\dots(\abag_j,\witnessset_j)$ with relabeling $h|_{\abag_j}$.
	This yields $\abstractdecomposition_j$ and $\subtermToBag_j$ satisfying Conditions~\ref{existenceZero}--\ref{existenceSix} and such that
	\[
	\subtermToBag_j(\abstractdecomposition_j)=\action_{\dpcore}^{\treewidthvalue}(h|_{\abag_j},\witnessset_j).
	\]
	Define $\abstractdecomposition_{\dprefutation}\defeq \introvertextype{h(\vertexone)}(\abstractdecomposition_j)$.
	Extend $\subtermToBag_j$ by setting
	\[
	\subtermToBag(\abstractdecomposition_{\dprefutation})\defeq \introvertextype{h(\vertexone)}(\subtermToBag_j(\abstractdecomposition_j)).
	\]
	By the induction hypothesis, $\abstractdecomposition_j\in \allabstractdecompositionstreewidth{\treewidthvalue}$ and $\topbag{\abstractdecomposition_j}=h(\abag_j)$ (Condition~\ref{existenceZero} applied to the prefix).
	Since $\vertexone\notin \abag_j$ and $h$ is injective, we have
	\[
	h(\vertexone)\notin \topbag{\abstractdecomposition_j}.
	\]
	Hence $\introvertextype{h(\vertexone)}(\abstractdecomposition_j)\in \allabstractdecompositionstreewidth{\treewidthvalue}$ and
	\[
	\topbag{\abstractdecomposition_{\dprefutation}}
	= h(\abag_j\cup\{\vertexone\})
	= g(\abag_m).
	\]
		Thus Condition~\ref{existenceZero} holds, and Conditions~\ref{existenceTwo}--\ref{existenceSix} are immediate.
		For Condition~\ref{existenceOne}, since the refutation is well-formed (Condition~\ref{sequence-relabeling-condition-wellformed}), $\mathrm{Lbl}_{\treewidthvalue}(\witnessset_j)\subseteq \abag_j$.
		Therefore Definition~\ref{definition:witnessaction}.\ref{condition:relabeldpcore-condition-extension-invariance} gives
		\[
		\action_{\dpcore}^{\treewidthvalue}(h|_{\abag_j},\witnessset_j)=\action_{\dpcore}^{\treewidthvalue}(h,\witnessset_j).
		\]
		Let $S^\star\defeq \introvertextype{\vertexone}(\witnessset_j)$.
		Since $(\abag_m,\witnessset_m)$ is well-formed and $\abag_m=\relabelingfunction_m(\abag_j\cup\{\vertexone\})$, we have $\mathrm{Lbl}_{\treewidthvalue}(\witnessset_m)\subseteq \relabelingfunction_m(\abag_j\cup\{\vertexone\})$.
	Moreover, $\witnessset_m=\action_{\dpcore}^{\treewidthvalue}(\relabelingfunction_m,S^\star)$ by assumption, and this action application is well-defined by Condition~\ref{sequence-relabeling-condition-three}; hence label support transport (Definition~\ref{definition:witnessaction}.\ref{condition:relabeldpcore-condition-support-transport}, applied pointwise to sets) yields
	\[
	\mathrm{Lbl}_{\treewidthvalue}(\witnessset_m)=\relabelingfunction_m(\mathrm{Lbl}_{\treewidthvalue}(S^\star)).
	\]
	Therefore $\relabelingfunction_m(\mathrm{Lbl}_{\treewidthvalue}(S^\star))\subseteq \relabelingfunction_m(\abag_j\cup\{\vertexone\})$, and since $\relabelingfunction_m$ is injective we conclude $\mathrm{Lbl}_{\treewidthvalue}(S^\star)\subseteq \abag_j\cup\{\vertexone\}$ by Lemma~\ref{lemma:injective-image-containment}.
	Hence extension invariance yields $\action_{\dpcore}^{\treewidthvalue}(\relabelingfunction_m|_{\abag_j\cup\{\vertexone\}},S^\star)=\action_{\dpcore}^{\treewidthvalue}(\relabelingfunction_m,S^\star)=\witnessset_m$, and
		\[
		\begin{aligned}
		\subtermToBag(\abstractdecomposition_{\dprefutation})
		&= \introvertextype{h(\vertexone)}(\action_{\dpcore}^{\treewidthvalue}(h,\witnessset_j))\\
		&= \action_{\dpcore}^{\treewidthvalue}(h,S^\star)\\
		&= \action_{\dpcore}^{\treewidthvalue}\bigl(g,\,\action_{\dpcore}^{\treewidthvalue}(\relabelingfunction_m|_{\abag_j\cup\{\vertexone\}},S^\star)\bigr)\\
		&= \action_{\dpcore}^{\treewidthvalue}(g,\witnessset_m),
		\end{aligned}
		\]
	where the second equality follows from Definition~\ref{definition:witnessaction}.\ref{condition:relabeldpcore-condition-four} and the third from Definition~\ref{definition:witnessaction}.\ref{condition:relabeldpcore-condition-three}.

\item \textbf{Forget vertex.}
Assume $(\abag_m,\witnessset_m)=(\relabelingfunction_m(\abag_j\setminus\{\vertexone\}),\,\action_{\dpcore}^{\treewidthvalue}(\relabelingfunction_m,\forgetvertextype{\vertexone}(\witnessset_j)))$ for some $\vertexone\in \abag_j$.
Let $h\defeq g\circ \relabelingfunction_m|_{\abag_j\setminus\{\vertexone\}}$, which is well-defined since $\relabelingfunction_m(\abag_j\setminus\{\vertexone\})=\abag_m$.
Choose any injective extension $\hat h:\abag_j\to[\treewidthvalue+1]$ of $h$ (which exists by Lemma~\ref{lemma:injective-extension}).
Apply the induction hypothesis to the prefix ending at $j$ with relabeling $\hat h$ to obtain $\abstractdecomposition_j$ and $\subtermToBag_j$ with $\subtermToBag_j(\abstractdecomposition_j)=\action_{\dpcore}^{\treewidthvalue}(\hat h,\witnessset_j)$.
Define $\abstractdecomposition_{\dprefutation}\defeq \forgetvertextype{\hat h(\vertexone)}(\abstractdecomposition_j)$ and extend $\subtermToBag_j$ by setting $\subtermToBag(\abstractdecomposition_{\dprefutation})\defeq \forgetvertextype{\hat h(\vertexone)}(\subtermToBag_j(\abstractdecomposition_j))$.
By the induction hypothesis, $\abstractdecomposition_j\in \allabstractdecompositionstreewidth{\treewidthvalue}$ and $\topbag{\abstractdecomposition_j}=\hat h(\abag_j)$.
Since $\vertexone\in \abag_j$, we have $\hat h(\vertexone)\in \topbag{\abstractdecomposition_j}$, so $\forgetvertextype{\hat h(\vertexone)}(\abstractdecomposition_j)\in \allabstractdecompositionstreewidth{\treewidthvalue}$ and
\[
\topbag{\abstractdecomposition_{\dprefutation}}
= \hat h(\abag_j\setminus\{\vertexone\})
= h(\abag_j\setminus\{\vertexone\})
= g(\abag_m).
\]
Thus Condition~\ref{existenceZero} holds, and Conditions~\ref{existenceTwo}--\ref{existenceSix} are immediate.
For Condition~\ref{existenceOne}, Definition~\ref{definition:witnessaction}.\ref{condition:relabeldpcore-condition-five} implies
\[
	\subtermToBag(\abstractdecomposition_{\dprefutation})
	= \action_{\dpcore}^{\treewidthvalue}(\hat h,\forgetvertextype{\vertexone}(\witnessset_j)).
	\]
	Let $S^\star\defeq \forgetvertextype{\vertexone}(\witnessset_j)$.
	Since $(\abag_m,\witnessset_m)$ is well-formed (Condition~\ref{sequence-relabeling-condition-wellformed}) and $\abag_m=\relabelingfunction_m(\abag_j\setminus\{\vertexone\})$, we have
	\[
	\mathrm{Lbl}_{\treewidthvalue}(\witnessset_m)\subseteq \relabelingfunction_m(\abag_j\setminus\{\vertexone\}).
	\]
	Moreover, $\witnessset_m=\action_{\dpcore}^{\treewidthvalue}(\relabelingfunction_m,S^\star)$ by assumption, and this action application is well-defined by Condition~\ref{sequence-relabeling-condition-three}; hence label support transport (Definition~\ref{definition:witnessaction}.\ref{condition:relabeldpcore-condition-support-transport}, applied pointwise to sets) yields
	\[
	\mathrm{Lbl}_{\treewidthvalue}(\witnessset_m)=\relabelingfunction_m(\mathrm{Lbl}_{\treewidthvalue}(S^\star)).
	\]
	Therefore $\relabelingfunction_m(\mathrm{Lbl}_{\treewidthvalue}(S^\star))\subseteq \relabelingfunction_m(\abag_j\setminus\{\vertexone\})$, and since $\relabelingfunction_m$ is injective we conclude $\mathrm{Lbl}_{\treewidthvalue}(S^\star)\subseteq \abag_j\setminus\{\vertexone\}=\domain(h)$ by Lemma~\ref{lemma:injective-image-containment}.
	Hence extension invariance yields $\action_{\dpcore}^{\treewidthvalue}(\hat h,S^\star)=\action_{\dpcore}^{\treewidthvalue}(h,S^\star)$, and Definition~\ref{definition:witnessaction}.\ref{condition:relabeldpcore-condition-three} gives
		\[
		\begin{aligned}
		\subtermToBag(\abstractdecomposition_{\dprefutation})
		&= \action_{\dpcore}^{\treewidthvalue}(\hat h,S^\star)\\
		&= \action_{\dpcore}^{\treewidthvalue}(h,S^\star)\\
		&= \action_{\dpcore}^{\treewidthvalue}\bigl(g,\action_{\dpcore}^{\treewidthvalue}(\relabelingfunction_m,S^\star)\bigr)\\
		&= \action_{\dpcore}^{\treewidthvalue}(g,\witnessset_m).
		\end{aligned}
		\]

\item \textbf{Introduce edge.}
Assume
\[
(\abag_m,\witnessset_m)
=
\bigl(\relabelingfunction_m(\abag_j),\,\action_{\dpcore}^{\treewidthvalue}(\relabelingfunction_m,\introedgegeneric{\vertexone}{\vertextwo}(\witnessset_j))\bigr)
\]
for some distinct $\vertexone,\vertextwo\in \abag_j$.
Let $h\defeq g\circ \relabelingfunction_m|_{\abag_j}$ and apply the induction hypothesis to the prefix ending at $j$ with relabeling $h$, obtaining $\abstractdecomposition_j$ and $\subtermToBag_j$ with $\subtermToBag_j(\abstractdecomposition_j)=\action_{\dpcore}^{\treewidthvalue}(h,\witnessset_j)$.
Define $\abstractdecomposition_{\dprefutation}\defeq \introedgegeneric{h(\vertexone)}{h(\vertextwo)}(\abstractdecomposition_j)$ and set $\subtermToBag(\abstractdecomposition_{\dprefutation})\defeq \introedgegeneric{h(\vertexone)}{h(\vertextwo)}(\subtermToBag_j(\abstractdecomposition_j))$.
By the induction hypothesis, $\abstractdecomposition_j\in \allabstractdecompositionstreewidth{\treewidthvalue}$ and $\topbag{\abstractdecomposition_j}=h(\abag_j)$.
Since $\vertexone,\vertextwo\in \abag_j$, we have $h(\vertexone),h(\vertextwo)\in \topbag{\abstractdecomposition_j}$, so $\introedgegeneric{h(\vertexone)}{h(\vertextwo)}(\abstractdecomposition_j)\in \allabstractdecompositionstreewidth{\treewidthvalue}$ and $\topbag{\abstractdecomposition_{\dprefutation}}=\topbag{\abstractdecomposition_j}=h(\abag_j)=g(\abag_m)$.
Thus Condition~\ref{existenceZero} holds.
Let $S^\star\defeq \introedgegeneric{\vertexone}{\vertextwo}(\witnessset_j)$.
Since $(\abag_m,\witnessset_m)$ is well-formed and $\abag_m=\relabelingfunction_m(\abag_j)$, we have
\[
\mathrm{Lbl}_{\treewidthvalue}(\witnessset_m)\subseteq \relabelingfunction_m(\abag_j).
\]
Moreover, $\witnessset_m=\action_{\dpcore}^{\treewidthvalue}(\relabelingfunction_m,S^\star)$, so label support transport gives
\[
\mathrm{Lbl}_{\treewidthvalue}(\witnessset_m)=\relabelingfunction_m(\mathrm{Lbl}_{\treewidthvalue}(S^\star)).
\]
By Lemma~\ref{lemma:injective-image-containment}, $\mathrm{Lbl}_{\treewidthvalue}(S^\star)\subseteq \abag_j=\domain(h)$.
Hence
\[
\begin{aligned}
\subtermToBag(\abstractdecomposition_{\dprefutation})
&= \introedgegeneric{h(\vertexone)}{h(\vertextwo)}(\action_{\dpcore}^{\treewidthvalue}(h,\witnessset_j))\\
&= \action_{\dpcore}^{\treewidthvalue}(h,S^\star)\\
&= \action_{\dpcore}^{\treewidthvalue}\bigl(g,\action_{\dpcore}^{\treewidthvalue}(\relabelingfunction_m|_{\abag_j},S^\star)\bigr)\\
&= \action_{\dpcore}^{\treewidthvalue}(g,\witnessset_m),
\end{aligned}
\]
where the second equality is the introduce-edge equivariance axiom and the third is composition. The remaining conditions are immediate.

\item \textbf{Join.}
Assume $(\abag_m,\witnessset_m)=(\relabelingfunction_m(\abag_j),\,\action_{\dpcore}^{\treewidthvalue}(\relabelingfunction_m,\jointype(\witnessset_j,\action_{\dpcore}^{\treewidthvalue}(\pi,\witnessset_l))))$ for some $l\in \{0\}\cup[m-1]$ with $\abag_l=\abag_j$ and $\pi\in\mathrm{Perm}(\abag_j)$.
Let $h\defeq g\circ \relabelingfunction_m|_{\abag_j}$.
Apply the induction hypothesis to the prefixes ending at $j$ and $l$ with relabelings $h$ and $h\circ \pi$, obtaining decompositions $\abstractdecomposition_j,\abstractdecomposition_l$ and maps $\subtermToBag_j,\subtermToBag_l$ with
\[
\subtermToBag_j(\abstractdecomposition_j)=\action_{\dpcore}^{\treewidthvalue}(h,\witnessset_j),
\qquad
\subtermToBag_l(\abstractdecomposition_l)=\action_{\dpcore}^{\treewidthvalue}(h\circ \pi,\witnessset_l).
\]
Define $\abstractdecomposition_{\dprefutation}\defeq \jointype(\abstractdecomposition_j,\abstractdecomposition_l)$, using fresh copies of the two child terms if necessary; combine $\subtermToBag_j$ and $\subtermToBag_l$ on the respective subterms, and set
\[
\subtermToBag(\abstractdecomposition_{\dprefutation}) \defeq \jointype(\subtermToBag_j(\abstractdecomposition_j),\subtermToBag_l(\abstractdecomposition_l)).
\]
By the induction hypothesis, $\abstractdecomposition_j,\abstractdecomposition_l\in \allabstractdecompositionstreewidth{\treewidthvalue}$, $\topbag{\abstractdecomposition_j}=h(\abag_j)$, and $\topbag{\abstractdecomposition_l}=(h\circ \pi)(\abag_l)=(h\circ \pi)(\abag_j)=h(\abag_j)$.
Hence $\jointype(\abstractdecomposition_j,\abstractdecomposition_l)\in \allabstractdecompositionstreewidth{\treewidthvalue}$ and $\topbag{\abstractdecomposition_{\dprefutation}}=h(\abag_j)=g(\abag_m)$, so Condition~\ref{existenceZero} holds.
Because $(\abag_l,\witnessset_l)$ is well-formed and $\pi$ is a permutation of $\abag_j=\abag_l$, label support transport gives
\[
\mathrm{Lbl}_{\treewidthvalue}(\action_{\dpcore}^{\treewidthvalue}(\pi,\witnessset_l))\subseteq \abag_j.
\]
Let $S^\star\defeq \jointype(\witnessset_j,\action_{\dpcore}^{\treewidthvalue}(\pi,\witnessset_l))$.
Since $(\abag_m,\witnessset_m)$ is well-formed, $\abag_m=\relabelingfunction_m(\abag_j)$, and $\witnessset_m=\action_{\dpcore}^{\treewidthvalue}(\relabelingfunction_m,S^\star)$, label support transport and Lemma~\ref{lemma:injective-image-containment} imply
\[
\mathrm{Lbl}_{\treewidthvalue}(S^\star)\subseteq \abag_j=\domain(h).
\]
Therefore
\[
\begin{aligned}
\subtermToBag(\abstractdecomposition_{\dprefutation})
&= \jointype(\action_{\dpcore}^{\treewidthvalue}(h,\witnessset_j),\action_{\dpcore}^{\treewidthvalue}(h\circ \pi,\witnessset_l))\\
&= \jointype(\action_{\dpcore}^{\treewidthvalue}(h,\witnessset_j),\action_{\dpcore}^{\treewidthvalue}(h,\action_{\dpcore}^{\treewidthvalue}(\pi,\witnessset_l)))\\
&= \action_{\dpcore}^{\treewidthvalue}(h,S^\star)\\
&= \action_{\dpcore}^{\treewidthvalue}\bigl(g,\action_{\dpcore}^{\treewidthvalue}(\relabelingfunction_m|_{\abag_j},S^\star)\bigr)\\
&= \action_{\dpcore}^{\treewidthvalue}(g,\witnessset_m),
\end{aligned}
\]
where the second equality uses composition, the third uses join equivariance, and the fourth uses composition again. This proves Condition~\ref{existenceOne}, and the remaining conditions are immediate.
\end{enumerate}
\end{proof}

Theorem~\ref{theorem:counterexample-existence-relabel-refutation} states that if $\dpcore$ is a coherent DP-core, then from each relabeled $(\treewidthvalue,\dpcore)$-refutation, one can extract a $\treewidthvalue$-instructive tree decomposition whose graph does not belong to $\dpcoregraphproperty{\dpcore}$.

Theorem~\ref{theorem:EquivalenceRefutation} corresponds to the special case of Theorem~\ref{theorem:counterexample-existence-relabel-refutation} in which every relabeling function is the identity, so no symmetry rewriting is permitted during the derivation. Theorem~\ref{theorem:counterexample-existence-relabel-refutation} shows that allowing admissible relabelings at intermediate steps does not change the existence of refutations (and hence counterexamples), while enabling the symmetry-aware search procedures developed in this section.

\begin{restatable}{theorem}{theoremcounterexampleexistencerelabelingref}\label{theorem:counterexample-existence-relabel-refutation}
Let $\relabelsequence=(\relabelingfunction_1,\dots,\relabelingfunction_m)$ be a sequence of relabeling functions, and let $\dpcore$ be a coherent DP-core with identity cleaning and equipped with a witness action. If there is an $\relabelsequence$-relabeled $(\treewidthvalue,\dpcore)$-refutation, then there exists a $\treewidthvalue$-instructive tree decomposition $\abstractdecomposition \in \allabstractdecompositionstreewidth{\treewidthvalue}$ where $\decompositiongraph{\abstractdecomposition}\notin \dpcoregraphproperty{\dpcore}$.
\end{restatable}
\begin{proof}
Let $\dprefutation= (\abag_0,\witnessset_0)(\abag_1,\witnessset_1)\dots (\abag_{m},\witnessset_m)$ be a $\relabelsequence$-relabeled $(\treewidthvalue,\dpcore)$-refutation.
Let $\mathrm{id}_{\abag_m}:\abag_m\to[\treewidthvalue+1]$ be the identity relabeling on $\abag_m$.
By Lemma~\ref{lemma:tree decomposition-existence-sequence-relabel-semi-refutation} applied to $\mathrm{id}_{\abag_m}$, there is a $\treewidthvalue$-instructive tree decomposition $\abstractdecomposition$ and a map
$\subtermToBag:\allsubterms(\abstractdecomposition) \to \finitepowerset{\dpcore[\treewidthvalue].\allwitnesses}$ such that Conditions~\ref{existenceOne}--\ref{existenceSix} hold. By Definition~\ref{definition:Dynamization}, Conditions~\ref{existenceTwo}--\ref{existenceSix} imply (by structural induction on subterms) that $\subtermToBag(\abstractdecomposition)=\dynamizationfunction{\dpcore}{\treewidthvalue}(\abstractdecomposition)$.
By the identity-cleaning hypothesis, this structural induction uses only the lifted DP-transitions appearing in Conditions~\ref{existenceTwo}--\ref{existenceSix}.
Moreover, by Lemma~\ref{lemma:witness-action-identity}, we have $\action_{\dpcore}^{\treewidthvalue}(\mathrm{id}_{\abag_m},\witnessset_m)=\witnessset_m$. Therefore,
\[
\dynamizationfunction{\dpcore}{\treewidthvalue}(\abstractdecomposition)
=\subtermToBag(\abstractdecomposition)
=\action_{\dpcore}^{\treewidthvalue}(\mathrm{id}_{\abag_m},\witnessset_m)
=\witnessset_m.
\]
	Since $\dprefutation$ is an $\relabelsequence$-relabeled $(\treewidthvalue,\dpcore)$-refutation, the set $\witnessset_m$ contains no final local witness, and hence $\abstractdecomposition$ is not accepted by $\dpcore[\treewidthvalue]$.
	We claim that $\decompositiongraph{\abstractdecomposition}\notin \dpcoregraphproperty{\dpcore}$.
	Indeed, suppose for contradiction that $\decompositiongraph{\abstractdecomposition}\in \dpcoregraphproperty{\dpcore}$.
	By definition of $\dpcoregraphproperty{\dpcore}$, there exist some $\treewidthvalue'\in\N$ and some $\abstractdecomposition'\in \allabstractdecompositionstreewidth{\treewidthvalue'}$ such that $\decompositiongraph{\abstractdecomposition'}\simeq \decompositiongraph{\abstractdecomposition}$ and $\abstractdecomposition'$ is accepted by $\dpcore[\treewidthvalue']$.
	By coherency (Definition~\ref{definition:CoherencyDPCore}), acceptance is invariant under graph isomorphism across widths, so $\abstractdecomposition$ would be accepted by $\dpcore[\treewidthvalue]$, contradicting that $\witnessset_m=\dynamizationfunction{\dpcore}{\treewidthvalue}(\abstractdecomposition)$ contains no final witness.
	Therefore $\decompositiongraph{\abstractdecomposition}\notin \dpcoregraphproperty{\dpcore}$.

\end{proof}

\subsection{Canonical Pairs and Canonized Refutations}
\label{subsection:CanonicalRefutation}

Let $U$ be an ordered set of elements, and let $X\subseteq U$ be finite. We let $\vectorfromset(X)$ be the vector obtained by
ordering the elements of $X$ from the smallest
value to the largest value. For instance, if $X = \{2,3,5\}$, then $\vectorfromset(X) = (2,3,5)$.
Given two such subsets $X,X'\subseteq U$, we say
that $X<X'$ if $\vectorfromset(X)$ is
lexicographically smaller than $\vectorfromset(X')$.
We use the standard lexicographic order on finite sequences: compare the first position where they differ, and if one sequence is a strict prefix of the other then the shorter sequence is smaller.

Let $\dpcore$ be a DP-core and fix an arbitrary
order for the set $\dpcore[\treewidthvalue].\allwitnesses$. For instance,
this order can be simply the lexicographic order on strings. We say that a pair $(\abag,\witnessset)$
is smaller than $(\abag',\witnessset')$ if the pair
$(\vectorfromset(\abag), \vectorfromset(\witnessset))$ is lexicographically smaller than the pair
$(\vectorfromset(\abag'), \vectorfromset(\witnessset'))$.
Any fixed total order yields a correct canonization procedure; the specific order only affects which representative is chosen.

\begin{definition}[Canonical Pair]\label{definition:canonicalpair}
Let $\abag\subseteq [\treewidthvalue+1]$ and let $\witnessset\in \finitepowerset{\dpcore[\treewidthvalue].\allwitnesses}$ be such that $(\abag,\witnessset)$ is well-formed, i.e., $\mathrm{Lbl}_{\treewidthvalue}(\witnessset)\subseteq \abag$.
$$\canonic_{\dpcore}^k(\abag,\witnessset) = \min \{(\relabelingfunction(\abag),\action_{\dpcore}^{\treewidthvalue}(\relabelingfunction,\witnessset))\;|\; \relabelingfunction\in \relabelclass_\treewidthvalue, \mathit{\domain(\relabelingfunction) = \abag}\}.$$
We call any function $f:\abag\rightarrow [\treewidthvalue+1]$ where the minimum in the above equation is achieved a
{\em canonical relabeling of $(\abag,\witnessset)$}.
We say that $(\abag,\witnessset)$ is \emph{canonical} if $\canonic_{\dpcore}^k(\abag,\witnessset)=(\abag,\witnessset)$.
\end{definition}

\begin{lemma}[Canonical representatives are orbit-invariant]
\label{lemma:canonical-orbit-invariant}
Let $\dpcore[\treewidthvalue]$ be equipped with a witness action. Let $(\abag,\witnessset)$ be well-formed and let $f:\abag\to[\treewidthvalue+1]$ be injective. If
\[
(\abag',\witnessset')=(f(\abag),\action_{\dpcore}^{\treewidthvalue}(f,\witnessset)),
\]
then
\[
\canonic_{\dpcore}^k(\abag',\witnessset')
=
\canonic_{\dpcore}^k(\abag,\witnessset).
\]
\end{lemma}
\begin{proof}
The two minima in Definition~\ref{definition:canonicalpair} are taken over the same relabeling orbit. Indeed, for every relabeling $h:\abag'\to[\treewidthvalue+1]$, the composition $h\circ f:\abag\to[\treewidthvalue+1]$ is a relabeling and, by the composition axiom,
\[
(h(\abag'),\action_{\dpcore}^{\treewidthvalue}(h,\witnessset'))
=
((h\circ f)(\abag),\action_{\dpcore}^{\treewidthvalue}(h\circ f,\witnessset)).
\]
Conversely, for every relabeling $g:\abag\to[\treewidthvalue+1]$, the map $g\circ f^{-1}:\abag'\to[\treewidthvalue+1]$ is a relabeling and inverse relabeling plus composition give
\[
((g\circ f^{-1})(\abag'),\action_{\dpcore}^{\treewidthvalue}(g\circ f^{-1},\witnessset'))
=
(g(\abag),\action_{\dpcore}^{\treewidthvalue}(g,\witnessset)).
\]
Thus the two sets of candidate representatives are equal, and their lexicographic minima are equal.
\end{proof}

By Lemma~\ref{lemma:well-formedness-preserved}, if $\dpcore$ is interface-respecting then all pairs arising in refutations and in dynamization are well-formed, and therefore $\canonic_{\dpcore}^k$ is defined for all pairs considered in this section. Moreover, since the bag component is compared before the witness-set component, canonization normalizes bags: any two canonical pairs with bags of the same size have the same bag component.

Intuitively, given a fixed action $\action_{\dpcore}^{\treewidthvalue}$ the canonical form
$\canonic_{\dpcore}^k(\abag,\witnessset)$
of a pair $(\abag,\witnessset)$ is the lexicographically smallest pair obtained by relabeling the elements of $\abag$ and each string $\awitness\in \witnessset$ according to some relabeling function $\relabelingfunction:\abag\rightarrow [\treewidthvalue+1]$. A \emph{canonical relabeling} of $(\abag,\witnessset)$ is any function $\relabelingfunction:\abag\rightarrow [\treewidthvalue +1]$
with the property that $(\relabelingfunction(\abag),\action_{\dpcore}^{\treewidthvalue}(\relabelingfunction,\witnessset))  = \canonic_{\dpcore}^k(\abag,\witnessset)$.

For example, consider $k=3$, bag $\abag=\{2,4\}$, and a singleton witness set $\witnessset=\{\gamma\}$ where (for example, in the max-degree core) $\gamma$ is a map with $\gamma(2)=0$ and $\gamma(4)=1$.
Let $f_1(2)=1,f_1(4)=2$ and $f_2(2)=2,f_2(4)=1$.
Then $\action_{\dpcore}^{\treewidthvalue}(f_1,\witnessset)=\{\gamma_1\}$ where $\gamma_1(1)=0,\gamma_1(2)=1$, while $\action_{\dpcore}^{\treewidthvalue}(f_2,\witnessset)=\{\gamma_2\}$ where $\gamma_2(1)=1,\gamma_2(2)=0$.
Canonization selects the lexicographically smallest among these relabeled pairs.

Next, we define the notion of a canonized refutation. Intuitively, such a refutation is obtained by canonizing the $(\treewidthvalue,\dpcore)$-pairs occurring in a
$(\treewidthvalue,\dpcore)$-refutation.

	\begin{definition}[Canonized Refutation]\label{definition:canonizedrefutation}
	Let $\dpcore$ be a coherent DP-core, and $\treewidthvalue\in \N$, $F=(f_1,\dots,f_m)$ be a sequence of relabelings, and
	$\dprefutation \equiv (\abag_0,\witnessset_0)(\abag_1,\witnessset_1)\dots (\abag_m,\witnessset_m)$
	be an $F$-relabeled $(\treewidthvalue,\dpcore)$-refutation. We say that $\dprefutation$ is canonized if for each $i\in \{0,1,\dots,m\}$, the pair $(\abag_i,\witnessset_i)$ is canonical.
	\end{definition}

Canonized refutations formalize quotienting the state space by interface relabelings: after each DP step we replace the resulting state by its canonical representative. Witness actions are the key requirement that makes this sound: they ensure that relabelings commute with the DP transitions, so canonizing ``on the fly'' preserves reachability of inconsistent states.

Our main theorem (Theorem~\ref{theorem:ProvabilityPreserving}) states that if $\dpcore$ is a coherent interface-respecting DP-core with identity cleaning and equipped with a witness action, then the existence of a canonized
$(\treewidthvalue,\dpcore)$-refutation is equivalent to the existence of a graph of treewidth
at most $\treewidthvalue$ that does not belong to the property
$\dpcoregraphproperty{\dpcore}$ defined by $\dpcore$.

\begin{restatable}{theorem}{theoremprovabilitypreserving}
\label{theorem:ProvabilityPreserving}
Let $\dpcore$ be a coherent interface-respecting DP-core with identity cleaning and equipped with a witness action and $\treewidthvalue\in \N$.
Then there is a canonized $(\treewidthvalue,\dpcore)$-refutation if and only if some graph of treewidth at most
$\treewidthvalue$ does not belong to $\dpcoregraphproperty{\dpcore}$.
\end{restatable}
\begin{proof}
\noindent($\Rightarrow$)
If there is a canonized $(\treewidthvalue,\dpcore)$-refutation, then in particular there is an $\relabelsequence$-relabeled $(\treewidthvalue,\dpcore)$-refutation. By Theorem~\ref{theorem:counterexample-existence-relabel-refutation}, there exists a $\treewidthvalue$-instructive tree decomposition $\abstractdecomposition$ with $\decompositiongraph{\abstractdecomposition}\notin \dpcoregraphproperty{\dpcore}$.

\noindent($\Leftarrow$)
Assume that some graph of treewidth at most $\treewidthvalue$ does not belong to $\dpcoregraphproperty{\dpcore}$. By Theorem~\ref{theorem:EquivalenceRefutation}, there exists a $(\treewidthvalue,\dpcore)$-refutation
\[
\dprefutation = (\abag_0,\witnessset_0)(\abag_1,\witnessset_1)\dots (\abag_{m},\witnessset_m)
\]
in the sense of Definition~\ref{definition:DPRefutation}.
Since $\dpcore$ is interface-respecting, Lemma~\ref{lemma:well-formedness-preserved} implies that $\dprefutation$ is well-formed, and therefore each application of $\action_{\dpcore}^{\treewidthvalue}(\cdot,\witnessset_i)$ below is well-defined.
For each $i\in \{0,1,\dots,m\}$ choose a relabeling function $c_i:\abag_i\to[\treewidthvalue+1]$ that attains the minimum in Definition~\ref{definition:canonicalpair}, and write
\[
(\bar\abag_i,\bar\witnessset_i)\defeq (c_i(\abag_i),\action_{\dpcore}^{\treewidthvalue}(c_i,\witnessset_i))
=\canonic_{\dpcore}^k(\abag_i,\witnessset_i).
\]
By construction, every $(\bar\abag_i,\bar\witnessset_i)$ is canonical.
Moreover, since $\witnessset_m$ contains no final witness, Lemma~\ref{lemma:witness-action-finality-sets} implies that also $\bar\witnessset_m$ contains no final witness.

We claim that $\bar\dprefutation\defeq (\bar\abag_0,\bar\witnessset_0)\dots(\bar\abag_m,\bar\witnessset_m)$ is an $\relabelsequence$-relabeled refutation for a suitable relabeling sequence $\relabelsequence$, and hence a canonized refutation.
Fix $i\in [m]$ and let $j\in \{0\}\cup[i-1]$ (and in the join case also $l\in \{0\}\cup[i-1]$) witness the derivation of $(\abag_i,\witnessset_i)$ from earlier state(s) in Definition~\ref{definition:DPRefutation}. We show that $(\bar\abag_i,\bar\witnessset_i)$ is derivable from $(\bar\abag_j,\bar\witnessset_j)$ (and $(\bar\abag_l,\bar\witnessset_l)$) according to Definition~\ref{definition:RelabeledRefutation}. We use the DP-core consistency axioms (Definition~\ref{definition:witnessaction}.\ref{condition:relabeldpcore-condition-four}--\ref{condition:relabeldpcore-condition-seven}) and the composition/inverse axioms (Definition~\ref{definition:witnessaction}.\ref{condition:relabeldpcore-condition-two}--\ref{condition:relabeldpcore-condition-three}), applied pointwise to witness sets.

\begin{enumerate}
\item \textbf{Introduce vertex.}
If $(\abag_i,\witnessset_i)=(\abag_j\cup\{\vertexone\},\introvertextype{\vertexone}(\witnessset_j))$ with $\vertexone\notin \abag_j$, choose an injective extension $\hat c_j:\abag_j\cup\{\vertexone\}\to[\treewidthvalue+1]$ of $c_j$ (which exists by Lemma~\ref{lemma:injective-extension}) and let $\vertexone'\defeq \hat c_j(\vertexone)$ (so $\vertexone'\notin \bar\abag_j$). Set $f_i\defeq c_i\circ \hat c_j^{-1}$. Then $\bar\abag_i=f_i(\bar\abag_j\cup\{\vertexone'\})$ and, by the intro-vertex axiom and composition,
\[
\bar\witnessset_i
=\action_{\dpcore}^{\treewidthvalue}(c_i,\introvertextype{\vertexone}(\witnessset_j))
=\action_{\dpcore}^{\treewidthvalue}(f_i,\introvertextype{\vertexone'}(\bar\witnessset_j)).
\]

\item \textbf{Forget vertex.}
If $(\abag_i,\witnessset_i)=(\abag_j\setminus\{\vertexone\},\forgetvertextype{\vertexone}(\witnessset_j))$ with $\vertexone\in \abag_j$, let $\vertexone'\defeq c_j(\vertexone)$ and define
\[
f_i:\bar\abag_j\setminus\{\vertexone'\}\to[\treewidthvalue+1],
\qquad
f_i(x)\defeq c_i(c_j^{-1}(x))
\ \text{ for each }\ x\in \bar\abag_j\setminus\{\vertexone'\}.
\]
Let $S^\star\defeq \forgetvertextype{\vertexone}(\witnessset_j)$. By interface-respectingness (Definition~\ref{definition:interface-respecting-dpcore}.\ref{interface-respecting-forgetvertex}), we have $\mathrm{Lbl}_{\treewidthvalue}(S^\star)\subseteq \abag_j\setminus\{\vertexone\}$.
Moreover, for each $x\in \abag_j\setminus\{\vertexone\}$ we have $c_i(x)=f_i(c_j(x))$, and therefore extension invariance and composition yield
\[
\action_{\dpcore}^{\treewidthvalue}(c_i,S^\star)
=\action_{\dpcore}^{\treewidthvalue}\bigl(f_i\circ (c_j|_{\abag_j\setminus\{\vertexone\}}),S^\star\bigr)
=\action_{\dpcore}^{\treewidthvalue}\bigl(f_i,\action_{\dpcore}^{\treewidthvalue}(c_j,S^\star)\bigr).
\]
By the forget-vertex axiom (Condition~\ref{condition:relabeldpcore-condition-five}),
\[
\action_{\dpcore}^{\treewidthvalue}(c_j,S^\star)
=\forgetvertextype{\vertexone'}(\action_{\dpcore}^{\treewidthvalue}(c_j,\witnessset_j))
=\forgetvertextype{\vertexone'}(\bar\witnessset_j).
\]
Then $\bar\abag_i=f_i(\bar\abag_j\setminus\{\vertexone'\})$ and, using the forget-vertex axiom and composition,
\[
\bar\witnessset_i
=\action_{\dpcore}^{\treewidthvalue}(c_i,\forgetvertextype{\vertexone}(\witnessset_j))
=\action_{\dpcore}^{\treewidthvalue}(f_i,\forgetvertextype{\vertexone'}(\bar\witnessset_j)).
\]

\item \textbf{Introduce edge.}
If $(\abag_i,\witnessset_i)=(\abag_j,\introedgegeneric{\vertexone}{\vertextwo}(\witnessset_j))$ with $\vertexone,\vertextwo\in \abag_j$, let $\vertexone'\defeq c_j(\vertexone)$, $\vertextwo'\defeq c_j(\vertextwo)$ and set $f_i\defeq c_i\circ c_j^{-1}$. Then $\bar\abag_i=f_i(\bar\abag_j)$ and by the intro-edge axiom and composition,
\[
\bar\witnessset_i
=\action_{\dpcore}^{\treewidthvalue}(c_i,\introedgegeneric{\vertexone}{\vertextwo}(\witnessset_j))
=\action_{\dpcore}^{\treewidthvalue}(f_i,\introedgegeneric{\vertexone'}{\vertextwo'}(\bar\witnessset_j)).
\]

\item \textbf{Join.}
If $(\abag_i,\witnessset_i)=(\abag_j,\jointype(\witnessset_j,\witnessset_l))$ for some $l\in \{0\}\cup[i-1]$ with $\abag_l=\abag_j$, note that $\bar\abag_l=\bar\abag_j$ because both are canonical bags and $|\abag_l|=|\abag_j|$. Let $\pi\defeq c_j\circ c_l^{-1}$; then $\pi\in \mathrm{Perm}(\bar\abag_j)$. Set $f_i\defeq c_i\circ c_j^{-1}$. Using the join axiom, inverse, and composition we obtain
\[
\bar\witnessset_i
=\action_{\dpcore}^{\treewidthvalue}(c_i,\jointype(\witnessset_j,\witnessset_l))
=\action_{\dpcore}^{\treewidthvalue}\bigl(f_i,\jointype(\bar\witnessset_j,\action_{\dpcore}^{\treewidthvalue}(\pi,\bar\witnessset_l))\bigr),
\]
which matches the join clause of Definition~\ref{definition:RelabeledRefutation}.
\end{enumerate}

Thus $\bar\dprefutation$ is a canonized $(\treewidthvalue,\dpcore)$-refutation.
\end{proof}

\subsection{Symmetry-Aware Search}\label{subsec:symmetry-search}

We now give a symmetry-aware state-space search procedure for a fixed width parameter $\treewidthvalue$.
Here $\dpcore$ denotes a $\treewidthvalue$-abstract DP-core; in general, $\dpcore$ may itself be obtained as a \emph{$\treewidthvalue$-abstract DP-combination} (i.e., by combining several DP-cores according to a Boolean formula) as in~\cite{de2023width}.
Here ``$\treewidthvalue$-abstract'' means that the DP-core (or combination) is specialized to the fixed interface label universe $[\treewidthvalue+1]$.

We explore the implicit transition system of $(\treewidthvalue,\dpcore)$-states by a standard breadth-first search, starting from the initial state.
Each generated successor is immediately canonized, so the search runs directly on canonical representatives in the quotient space induced by relabelings.
Algorithm~\ref{algorithm:InclusionTestCombination} specifies the exploration and returns a refutation trace as soon as an inconsistent state is reached (otherwise it reports that inclusion holds).
In the join case, the current state is combined with previously processed states with the same bag (in both argument orders), considering all permutations of the shared interface.

	\begin{algorithm}[H]
	\caption{Relabeled Inclusion Test}
	\label{algorithm:InclusionTestCombination}
	\footnotesize
	\Input{A finite, coherent, interface-respecting $\treewidthvalue$-abstract DP-combination $\dpcore$ with identity cleaning and equipped with a witness action}
		\Output{A canonized relabeled $(\treewidthvalue,\dpcore)$-refutation trace $R$ if there exists $G\in \allgraphstreewidth{\treewidthvalue}$ with $G\notin \dpcoregraphproperty{\dpcore}$, and "Inclusion Holds", otherwise.}
	$R \gets []$\;
	$Seen \gets \emptyset$\;
	$Init \gets \canonic_{\dpcore}^k(\emptyset,\dpcore[\treewidthvalue].\initialsetgeneric)$\;
	$Y \gets [Init]$\;
	$Seen \gets Seen \cup \{Init\}$\;
	\While{$Y\neq \emptyset$}{
		Remove the next pair $(\abag,\mathbf{\witnessset})$ from $Y$ and append it to $R$\;
	\eIf{$(\abag,\mathbf{\witnessset})$
	    is an inconsistent $(\treewidthvalue,\dpcore)$-state,}{
	    \Return $R$\;
	}{
		\tcp{Introduce-vertex successors}
		\ForEach{$\vertexone\in [\treewidthvalue+1]\setminus \abag$}{
			$X \leftarrow (\abag\cup \{\vertexone\},\dpcore.\introvertextype{\vertexone}(\mathbf{\witnessset}))$ \;
			$\bar{X} \gets \canonic_{\dpcore}^k(X)$\;
			\If{$\bar{X} \notin Seen$}{
				$Seen \gets Seen \cup \{\bar{X}\}$\;
				Append $\bar{X}$ to $Y$\;
			}
		}
		\tcp{Forget-vertex successors}
		\ForEach{$\vertexone\in \abag$}{
			$X\leftarrow (\abag\setminus \{\vertexone\},\dpcore.\forgetvertextype{\vertexone}(\mathbf{\witnessset}))$ \;
			$\bar{X} \gets \canonic_{\dpcore}^k(X)$\;
			\If{$\bar{X} \notin Seen$}{
				$Seen \gets Seen \cup \{\bar{X}\}$\;
				Append $\bar{X}$ to $Y$\;
			}
		}
		\tcp{Introduce-edge successors}
		\ForEach{$\vertexone,\vertextwo\in \abag$ with $\vertexone\neq \vertextwo$}{
			$X\leftarrow (\abag,\dpcore.\introedgegeneric{\vertexone}{\vertextwo}(\mathbf{\witnessset}))$ \;
			$\bar{X} \gets \canonic_{\dpcore}^k(X)$\;
			\If{$\bar{X} \notin Seen$}{
				$Seen \gets Seen \cup \{\bar{X}\}$\;
				Append $\bar{X}$ to $Y$\;
			}
		}
		\tcp{Join successors with all earlier same-bag states (both argument orders)}
		\ForEach{$(\abag',\mathbf{\witnessset}')\in R$ with $\abag= \abag'$}{
			\ForEach{$\pi\in \mathrm{Perm}(\abag)$}{
				$X\leftarrow (\abag, \dpcore.\jointype(\mathbf{\witnessset},\action_{\dpcore}^{\treewidthvalue}(\pi,\mathbf{\witnessset}')))$\;
				$\bar{X} \gets \canonic_{\dpcore}^k(X)$\;
				\If{$\bar{X} \notin Seen$}{
					$Seen \gets Seen \cup \{\bar{X}\}$\;
					Append $\bar{X}$ to $Y$\;
				}
				$X\leftarrow (\abag, \dpcore.\jointype(\mathbf{\witnessset}',\action_{\dpcore}^{\treewidthvalue}(\pi,\mathbf{\witnessset})))$\;
				$\bar{X} \gets \canonic_{\dpcore}^k(X)$\;
				\If{$\bar{X} \notin Seen$}{
					$Seen \gets Seen \cup \{\bar{X}\}$\;
					Append $\bar{X}$ to $Y$\;
				}
			}
		}
	}
	}
	\Return "Inclusion Holds"
	\end{algorithm}

		\begin{lemma}[Algorithm~\ref{algorithm:InclusionTestCombination} returns a canonized refutation]
		\label{lemma:algorithm-returns-canonized-refutation}
			Assume $\dpcore$ is coherent, interface-respecting, has identity cleaning, and is equipped with a witness action.
	If Algorithm~\ref{algorithm:InclusionTestCombination} returns a list $R$ (i.e., it returns from the inconsistent-pair branch), then there exists a relabeling sequence $\relabelsequence$ such that $R$ is a canonized $(\treewidthvalue,\dpcore)$-refutation in the sense of Definition~\ref{definition:canonizedrefutation}.
		\end{lemma}
		\begin{proof}
		First note the following invariant: every pair inserted into $Seen$ or $Y$ is well-formed, and every raw successor is well-formed before it is canonized.
		The initial raw pair $(\emptyset,\dpcore[\treewidthvalue].\initialsetgeneric)$ is well-formed by Definition~\ref{definition:interface-respecting-dpcore}.\ref{interface-respecting-initial}, and its canonical representative is well-formed by label support transport.
		Assume next that the processed pair $(\abag,\witnessset)$ and, in the join case, the earlier pair $(\abag,\witnessset')$ are well-formed.
		For introduce vertex, forget vertex, and introduce edge, the corresponding interface-respecting clauses imply that the raw successor witness set has support contained in the raw successor bag.
		For join, if $\pi\in \mathrm{Perm}(\abag)$, then $\action_{\dpcore}^{\treewidthvalue}(\pi,\witnessset')$ has support contained in $\abag$ by label support transport, and the join clause of interface-respectingness gives support contained in $\abag$ for the joined witness set.
		Finally, applying $\canonic_{\dpcore}^k$ only relabels a well-formed state, so the result is well-formed again.
		Thus all generated pairs are well-formed and every call to $\canonic_{\dpcore}^k$ is well-defined.
		By construction, every pair appended to $R$ is canonical, because it is obtained as $\canonic_{\dpcore}^k(\abag,\witnessset)$ for some pair $(\abag,\witnessset)$.
	The first element of $R$ is $\canonic_{\dpcore}^k(\emptyset,\dpcore[\treewidthvalue].\initialsetgeneric)$, which equals the initial canonical pair.
	If the algorithm returns $R$, then its last element is inconsistent by the return condition.

		It remains to show that $R$ satisfies the derivation condition of Definition~\ref{definition:RelabeledRefutation}.
		Consider any pair $(\bar\abag,\bar\witnessset)$ in $R$ after the initial element.
		By inspection of the update rules, it is of the form $\canonic_{\dpcore}^k(X)$ where $X$ is obtained from some earlier pair in $R$ by applying one DP-operation (and in the join case, also using a previously processed join partner and a permutation of the shared bag).
		Write $X=(\abag_X,\witnessset_X)$.
		Let $f$ be a canonical relabeling witnessing $\canonic_{\dpcore}^k(X)=(f(\abag_X),\action_{\dpcore}^{\treewidthvalue}(f,\witnessset_X))$ as in Definition~\ref{definition:canonicalpair}.
		Then $(\bar\abag,\bar\witnessset)=(f(\abag_X),\action_{\dpcore}^{\treewidthvalue}(f,\witnessset_X))$ is derivable from the same earlier pair(s) by the corresponding clause of Definition~\ref{definition:RelabeledRefutation}.

	We can therefore choose, for each non-initial element of $R$, some such relabeling function $f_i$ witnessing its derivation from earlier element(s), yielding a relabeling sequence $\relabelsequence$ for which $R$ is an $\relabelsequence$-relabeled refutation.
	Since all pairs in $R$ are canonical, $R$ is canonized.
	\end{proof}

	\begin{theorem}[Correctness of Algorithm~\ref{algorithm:InclusionTestCombination}]
	\label{theorem:algorithm-correctness}
	Assume $\dpcore$ is finite, coherent, interface-respecting, has identity cleaning, and is equipped with a witness action.
	Then Algorithm~\ref{algorithm:InclusionTestCombination} terminates.
	Moreover, it returns ``Inclusion Holds'' if and only if $\allgraphstreewidth{\treewidthvalue}\subseteq \dpcoregraphproperty{\dpcore}$, and otherwise it returns a canonized $(\treewidthvalue,\dpcore)$-refutation.
	\end{theorem}
		\begin{proof}
		Since $\dpcore$ is finite and $\treewidthvalue$ is fixed, there are only finitely many canonical $(\treewidthvalue,\dpcore)$-states.
		The algorithm inserts a canonical state into $Seen$ at most once, so it can process only finitely many canonical states; hence the algorithm terminates.

	If the algorithm returns a list $R$, then by Lemma~\ref{lemma:algorithm-returns-canonized-refutation} there exists a relabeling sequence for which $R$ is a canonized $(\treewidthvalue,\dpcore)$-refutation.
	By Theorem~\ref{theorem:ProvabilityPreserving}, this implies $\allgraphstreewidth{\treewidthvalue}\nsubseteq \dpcoregraphproperty{\dpcore}$, so returning a refutation is sound.

			Otherwise the algorithm returns ``Inclusion Holds'', meaning that the exploration finishes without ever reaching an inconsistent canonical state.
			We claim that in this case no canonized $(\treewidthvalue,\dpcore)$-refutation exists.
			Assume for contradiction that there exists a canonized $(\treewidthvalue,\dpcore)$-refutation
			\[
			R^\star = (\abag_0,\witnessset_0)\cdots(\abag_m,\witnessset_m).
			\]
			We show by induction on $i\in \{0,\dots,m\}$ that $(\abag_i,\witnessset_i)$ is eventually generated by the algorithm.
			The base case $i=0$ holds because $(\abag_0,\witnessset_0)$ is the initial canonical pair, which equals $Init$.
			For the induction step, consider the derivation of $(\abag_i,\witnessset_i)$ from earlier pair(s) as in Definition~\ref{definition:RelabeledRefutation}.
			If it is derived by an introduce-vertex, forget-vertex, or introduce-edge step from some $(\abag_j,\witnessset_j)$ with $j<i$, then when the algorithm processes $(\abag_j,\witnessset_j)$ it enumerates the corresponding operation (over all eligible labels), canonizes the resulting successor, and therefore generates $(\abag_i,\witnessset_i)$ by Lemma~\ref{lemma:canonical-orbit-invariant}, since $(\abag_i,\witnessset_i)$ is already canonical.
			If it is derived by a join step from $(\abag_j,\witnessset_j)$ and $(\abag_\ell,\witnessset_\ell)$ with $j,\ell<i$ and some $\pi\in\mathrm{Perm}(\abag_j)$, then by the induction hypothesis both premise pairs are eventually dequeued; let $(\abag_s,\witnessset_s)$ be the one dequeued later.
			At that point the other premise is already contained in $R$, and the join loop considers all permutations and both argument orders, so it generates the corresponding join successor and inserts its canonization, which equals $(\abag_i,\witnessset_i)$ by Lemma~\ref{lemma:canonical-orbit-invariant}.
			This completes the induction.
			In particular, $(\abag_m,\witnessset_m)$ is eventually generated; since it is inconsistent, the algorithm would return from the inconsistent branch, contradicting that it returns ``Inclusion Holds''.
			Therefore no canonized refutation exists, and by Theorem~\ref{theorem:ProvabilityPreserving} we conclude $\allgraphstreewidth{\treewidthvalue}\subseteq \dpcoregraphproperty{\dpcore}$.
			Conversely, if $\allgraphstreewidth{\treewidthvalue}\subseteq \dpcoregraphproperty{\dpcore}$, then by Theorem~\ref{theorem:ProvabilityPreserving} there is no canonized refutation, so no inconsistent canonical state is reachable and the algorithm cannot return from the inconsistent branch; therefore it returns ``Inclusion Holds''.
			\end{proof}

\section{Early Pruning}
\label{section:EarlyPruning}

Our second main contribution is an early-pruning procedure that can be applied in the study of conjectures of the form $\graphproperty_1\rightarrow \graphproperty_2$, where $\graphproperty_1$ is closed under subgraphs. This reduces the search space even more because it allows us to avoid the computation of states that do not contribute to the search for a counterexample.

A graph property $\graphproperty$ is said to be closed under subgraphs if
whenever a graph $\agraph$ belongs to $\graphproperty$, we have that
every subgraph of $\agraph$ also belongs to $\graphproperty$.
For example, the property of being $3$-colorable is closed under subgraphs.
In contrast, the property of having minimum degree at least $d$ is not closed under subgraphs.
If the graph property $\dpcoregraphproperty{\dpcore}$ specified by a given coherent DP-core $\dpcore$ is of the form $\dpcoregraphproperty{\dpcore_1}\rightarrow \dpcoregraphproperty{\dpcore_2}$ for coherent DP-cores $\dpcore_1$ and $\dpcore_2$, such that $\dpcoregraphproperty{\dpcore_1}$ is closed under subgraphs, then when running our inference algorithm to
determine whether some graph of treewidth at most $\treewidthvalue$ is not contained in $\dpcoregraphproperty{\dpcore}$ we may prune the search earlier.
The following simple, but crucial observation is the basis of our specialized search.

\begin{observation}\label{remark:closed-under-subgraph}
Let $\graphproperty$ be a graph property closed under subgraphs. Let $\abstractdecomposition$ and $\abstractdecomposition'$ be $\treewidthvalue$-instructive tree decompositions such that $\abstractdecomposition$ is a subterm of $\abstractdecomposition'$. Then, if $\decompositiongraph{\abstractdecomposition}\notin \graphproperty$, then  $\decompositiongraph{\abstractdecomposition'}\notin\graphproperty$.
\end{observation}
\begin{proof}
Assume that $\decompositiongraph{\abstractdecomposition}\notin \graphproperty$. Since $\abstractdecomposition$ is a subterm of $\abstractdecomposition'$, Lemma~\ref{lemma:subterm-subgraph} yields that $\decompositiongraph{\abstractdecomposition}$ is isomorphic to a subgraph of $\decompositiongraph{\abstractdecomposition'}$. Suppose $\decompositiongraph{\abstractdecomposition'}\in \graphproperty$. Since $\graphproperty$ is closed under subgraphs, that subgraph belongs to $\graphproperty$; because graph properties are taken up to isomorphism, $\decompositiongraph{\abstractdecomposition}\in \graphproperty$. This contradicts the assumption that $\decompositiongraph{\abstractdecomposition}\notin \graphproperty$.
\end{proof}

Observation \ref{remark:closed-under-subgraph} implies that
in order to determine whether there is a graph of treewidth
at most $\treewidthvalue$ that does not belong to
$\dpcoregraphproperty{\dpcore}$, instead of searching for
inconsistent $(\treewidthvalue,\dpcore)$-pairs, we may instead search for inconsistent {\em $(\treewidthvalue,\dpcore_1,\dpcore_2)$-triples}. Such a triple
is a triple of the form $(\abag,\witnessset_1,\witnessset_2)$ satisfying the following properties:

\begin{enumerate}
    \item $(\abag,\witnessset_1)$ is a $(\treewidthvalue,\dpcore_1)$-pair,
    \item $(\abag,\witnessset_2)$ is a
    $(\treewidthvalue,\dpcore_2)$-pair,
    \item $\witnessset_1$ has a final local witness for $\dpcore_1$
    but $\witnessset_2$ does not have a final local witness for
    $\dpcore_2$.
\end{enumerate}

This allows a more efficient search because, since
$\dpcoregraphproperty{\dpcore_1}$ is assumed to be closed under subgraphs,
as soon as we have reached a $(\treewidthvalue,\dpcore_1,\dpcore_2)$-triple $(\abag,\witnessset_1,\witnessset_2)$ where $(\abag,\witnessset_1)$ is an inconsistent $(\treewidthvalue,\dpcore_1)$-pair, we know that no triple $(\abag,\witnessset_1',\witnessset_2')$ derived from
$(\abag,\witnessset_1,\witnessset_2)$ will be inconsistent (because $\witnessset_1'$ does not contain a final witness for $\dpcore_1$).
Therefore, we do not need to consider $(\treewidthvalue,\dpcore_1,\dpcore_2)$-triples derived from
$(\abag,\witnessset_1,\witnessset_2)$. We formalize this idea with the premise-pruned search described below.

\subsection{Premise Search Algorithm}
\label{section:premise-search-appendix}

Let $\dpcore_1$ and $\dpcore_2$ be coherent DP-cores with identity cleaning and fix $\treewidthvalue\in \N$. We write $P_i \defeq \dpcoregraphproperty{\dpcore_i}$ for $i\in\{1,2\}$ and assume that $P_1$ is closed under subgraphs.

A \emph{$(\treewidthvalue,\dpcore_1,\dpcore_2)$-state} is a triple $(\abag,\witnessset_1,\witnessset_2)$ where $\abag\subseteq [\treewidthvalue+1]$ and $\witnessset_i \in \finitepowerset{\dpcore_i[\treewidthvalue].\allwitnesses}$ for each $i\in\{1,2\}$. The initial state is $(\emptyset,\dpcore_1[\treewidthvalue].\initialsetgeneric,\dpcore_2[\treewidthvalue].\initialsetgeneric)$.

We say that a state $(\abag,\witnessset_1,\witnessset_2)$ \emph{satisfies the premise} if $\witnessset_1$ contains a final local witness for $\dpcore_1[\treewidthvalue]$. We say that it is \emph{inconsistent} if it satisfies the premise but $\witnessset_2$ contains no final local witness for $\dpcore_2[\treewidthvalue]$.
Equivalently (for reachable states), this means that the partial graph constructed so far belongs to $P_1$, since each reachable state arises from the dynamization of some subterm of an instructive tree decomposition.
For example, if $P_1$ is triangle-freeness, once a partial construction already contains a triangle then no extension can restore triangle-freeness, so we can prune immediately.

From a state $(\abag,\witnessset_1,\witnessset_2)$ we may derive successor states using the following deduction rules.
\begin{enumerate}
\item \textbf{Introduce vertex.} For each $\vertexone\in [\treewidthvalue+1]\setminus \abag$,
\[
(\abag\cup\{\vertexone\},\, \dpcore_1[\treewidthvalue].\introvertextype{\vertexone}(\witnessset_1),\, \dpcore_2[\treewidthvalue].\introvertextype{\vertexone}(\witnessset_2)).
\]
\item \textbf{Forget vertex.} For each $\vertexone\in \abag$,
\[
(\abag\setminus\{\vertexone\},\, \dpcore_1[\treewidthvalue].\forgetvertextype{\vertexone}(\witnessset_1),\, \dpcore_2[\treewidthvalue].\forgetvertextype{\vertexone}(\witnessset_2)).
\]
\item \textbf{Introduce edge.} For each two distinct $\vertexone,\vertextwo\in \abag$,
\[
(\abag,\, \dpcore_1[\treewidthvalue].\introedgegeneric{\vertexone}{\vertextwo}(\witnessset_1),\, \dpcore_2[\treewidthvalue].\introedgegeneric{\vertexone}{\vertextwo}(\witnessset_2)).
\]
\item \textbf{Join.} For any other derived state $(\abag,\witnessset_1',\witnessset_2')$ with the same bag $\abag$,
\[
(\abag,\, \dpcore_1[\treewidthvalue].\jointype(\witnessset_1,\witnessset_1'),\, \dpcore_2[\treewidthvalue].\jointype(\witnessset_2,\witnessset_2')).
\]
\end{enumerate}

We expand only derived states that satisfy the premise. States failing the premise may be recorded to avoid duplicate work, but no successors are generated from them. This pruning is complete: if there exists a $\treewidthvalue$-instructive tree decomposition $\abstractdecomposition$ such that $\decompositiongraph{\abstractdecomposition}\in P_1$ but $\decompositiongraph{\abstractdecomposition}\notin P_2$, then every subterm $\sigma$ of $\abstractdecomposition$ satisfies $\decompositiongraph{\sigma}\in P_1$ by Observation~\ref{remark:closed-under-subgraph}. Hence every intermediate state in a derivation of $\abstractdecomposition$ satisfies the premise.

We explore the reachable state space by a standard breadth-first search over this transition system, starting from the initial state.
Algorithm~\ref{algorithm:PremisePrunedInclusionTest} specifies the premise-pruned exploration.

\begin{algorithm}[H]
\caption{Premise-Pruned Inclusion Test}
\label{algorithm:PremisePrunedInclusionTest}
\footnotesize
\Input{Finite coherent $\treewidthvalue$-abstract DP-cores $\dpcore_1,\dpcore_2$ with identity cleaning and with $P_1=\dpcoregraphproperty{\dpcore_1}$ closed under subgraphs}
\Output{A reachability trace $R$ ending in an inconsistent $(\treewidthvalue,\dpcore_1,\dpcore_2)$-state if there exists $G\in \allgraphstreewidth{\treewidthvalue}$ with $G\in P_1$ and $G\notin P_2$, and "Inclusion Holds", otherwise.}
$R \gets []$\;
$Seen \gets \emptyset$\;
$Init \gets (\emptyset,\dpcore_1[\treewidthvalue].\initialsetgeneric,\dpcore_2[\treewidthvalue].\initialsetgeneric)$\;
$Y \gets [Init]$\;
$Seen \gets Seen \cup \{Init\}$\;
\While{$Y\neq \emptyset$}{
		Remove the next triple $(\abag,\mathbf{\witnessset}_1,\mathbf{\witnessset}_2)$ from $Y$ and append it to $R$\;
		\If{$(\abag,\mathbf{\witnessset}_1,\mathbf{\witnessset}_2)$ is inconsistent}{
			\Return $R$\;
		}
	\If{$(\abag,\mathbf{\witnessset}_1,\mathbf{\witnessset}_2)$ satisfies the premise}{
		\tcp{Introduce-vertex successors}
		\ForEach{$\vertexone\in [\treewidthvalue+1]\setminus \abag$}{
			$X \leftarrow (\abag\cup \{\vertexone\},\dpcore_1[\treewidthvalue].\introvertextype{\vertexone}(\mathbf{\witnessset}_1),\dpcore_2[\treewidthvalue].\introvertextype{\vertexone}(\mathbf{\witnessset}_2))$\;
			\If{$X \notin Seen$}{
				$Seen \gets Seen \cup \{X\}$\;
				Append $X$ to $Y$\;
			}
		}
		\tcp{Forget-vertex successors}
		\ForEach{$\vertexone\in \abag$}{
			$X \leftarrow (\abag\setminus \{\vertexone\},\dpcore_1[\treewidthvalue].\forgetvertextype{\vertexone}(\mathbf{\witnessset}_1),\dpcore_2[\treewidthvalue].\forgetvertextype{\vertexone}(\mathbf{\witnessset}_2))$\;
			\If{$X \notin Seen$}{
				$Seen \gets Seen \cup \{X\}$\;
				Append $X$ to $Y$\;
			}
		}
		\tcp{Introduce-edge successors}
		\ForEach{$\vertexone,\vertextwo\in \abag$ with $\vertexone\neq \vertextwo$}{
			$X \leftarrow (\abag,\dpcore_1[\treewidthvalue].\introedgegeneric{\vertexone}{\vertextwo}(\mathbf{\witnessset}_1),\dpcore_2[\treewidthvalue].\introedgegeneric{\vertexone}{\vertextwo}(\mathbf{\witnessset}_2))$\;
			\If{$X \notin Seen$}{
				$Seen \gets Seen \cup \{X\}$\;
				Append $X$ to $Y$\;
			}
		}
		\tcp{Join successors with all earlier same-bag states (both argument orders)}
		\ForEach{$(\abag',\mathbf{\witnessset}'_1,\mathbf{\witnessset}'_2)\in R$ with $\abag=\abag'$}{
			\If{$(\abag',\mathbf{\witnessset}'_1,\mathbf{\witnessset}'_2)$ satisfies the premise}{
				$X \leftarrow (\abag,\dpcore_1[\treewidthvalue].\jointype(\mathbf{\witnessset}_1,\mathbf{\witnessset}'_1),\dpcore_2[\treewidthvalue].\jointype(\mathbf{\witnessset}_2,\mathbf{\witnessset}'_2))$\;
				\If{$X \notin Seen$}{
					$Seen \gets Seen \cup \{X\}$\;
					Append $X$ to $Y$\;
				}
				$X \leftarrow (\abag,\dpcore_1[\treewidthvalue].\jointype(\mathbf{\witnessset}'_1,\mathbf{\witnessset}_1),\dpcore_2[\treewidthvalue].\jointype(\mathbf{\witnessset}'_2,\mathbf{\witnessset}_2))$\;
				\If{$X \notin Seen$}{
					$Seen \gets Seen \cup \{X\}$\;
					Append $X$ to $Y$\;
				}
			}
		}
	}
}
\Return "Inclusion Holds"\;
\end{algorithm}

We prove correctness next. If the premise-pruned search reaches an inconsistent state, then we can extract a $\treewidthvalue$-instructive tree decomposition $\abstractdecomposition$ such that $\decompositiongraph{\abstractdecomposition}\in P_1$ and $\decompositiongraph{\abstractdecomposition}\notin P_2$, i.e., a counterexample to $P_1\rightarrow P_2$. Conversely, any such counterexample yields a reachable inconsistent state, and all intermediate states satisfy the premise by Observation~\ref{remark:closed-under-subgraph}.

\begin{theorem}[Premise-pruned search correctness]
\label{theorem:premise-pruned-correctness}
Assume that $\dpcore_1$ and $\dpcore_2$ are finite coherent DP-cores with identity cleaning and that $P_1$ is closed under subgraphs. Then Algorithm~\ref{algorithm:PremisePrunedInclusionTest} terminates. Moreover:
\begin{itemize}
\setlength\itemsep{0em}
\item \textnormal{\textbf{Soundness:}} If an inconsistent $(\treewidthvalue,\dpcore_1,\dpcore_2)$-state is reachable from the initial state using the deduction rules above, then there exists a $\treewidthvalue$-instructive tree decomposition $\abstractdecomposition\in \allabstractdecompositionstreewidth{\treewidthvalue}$ such that $\decompositiongraph{\abstractdecomposition}\in P_1$ and $\decompositiongraph{\abstractdecomposition}\notin P_2$.
\item \textnormal{\textbf{Completeness:}} If there exists $\abstractdecomposition\in \allabstractdecompositionstreewidth{\treewidthvalue}$ such that $\decompositiongraph{\abstractdecomposition}\in P_1$ and $\decompositiongraph{\abstractdecomposition}\notin P_2$, then an inconsistent state is reachable, and moreover there is such a reachability derivation in which every intermediate state satisfies the premise.
\end{itemize}
Consequently, Algorithm~\ref{algorithm:PremisePrunedInclusionTest} returns ``Inclusion Holds'' if and only if no such counterexample exists.
\end{theorem}
\begin{proof}
Since $\dpcore_1[\treewidthvalue].\allwitnesses$ and $\dpcore_2[\treewidthvalue].\allwitnesses$ are finite, there are only finitely many triples $(\abag,\witnessset_1,\witnessset_2)$ with $\abag\subseteq[\treewidthvalue+1]$ and $\witnessset_i\subseteq \dpcore_i[\treewidthvalue].\allwitnesses$.
Algorithm~\ref{algorithm:PremisePrunedInclusionTest} inserts a triple into $Seen$ at most once, and therefore terminates.

For soundness, let $(\abag,\witnessset_1,\witnessset_2)$ be a reachable inconsistent state. Fix a reachability derivation
$(\abag_0,\witnessset_{1,0},\witnessset_{2,0}),\dots,(\abag_m,\witnessset_{1,m},\witnessset_{2,m})=(\abag,\witnessset_1,\witnessset_2)$
starting from the initial state and applying the deduction rules.

We construct, by induction on $i\in\{0,\dots,m\}$, a $\treewidthvalue$-instructive tree decomposition $\abstractdecomposition_i$ such that
$\topbag{\abstractdecomposition_i}=\abag_i$ and
\[
\dynamizationfunction{\dpcore_1}{\treewidthvalue}(\abstractdecomposition_i)=\witnessset_{1,i}
\quad\text{and}\quad
\dynamizationfunction{\dpcore_2}{\treewidthvalue}(\abstractdecomposition_i)=\witnessset_{2,i}.
\]
For $i=0$, take $\abstractdecomposition_0=\leaftype$. For the step, if state $i$ is obtained from some earlier state $j$ by an introduce-vertex/forget-vertex/introduce-edge rule, define $\abstractdecomposition_i$ to be the corresponding constructor applied to $\abstractdecomposition_j$; if state $i$ is obtained from earlier states $j,\ell$ by the join rule, define $\abstractdecomposition_i\defeq \jointype(\abstractdecomposition_j,\abstractdecomposition_\ell)$. The side conditions in the deduction rules ensure that each constructor is legal, hence $\abstractdecomposition_i\in \allabstractdecompositionstreewidth{\treewidthvalue}$. The equalities on dynamization follow directly from Definition~\ref{definition:Dynamization} and the induction hypothesis.

Applying this at $i=m$, let $\abstractdecomposition\defeq \abstractdecomposition_m$. Since the state $(\abag,\witnessset_1,\witnessset_2)$ is inconsistent, $\witnessset_1$ contains a final witness of $\dpcore_1[\treewidthvalue]$ while $\witnessset_2$ contains no final witness of $\dpcore_2[\treewidthvalue]$. Hence $\abstractdecomposition$ is accepted by $\dpcore_1[\treewidthvalue]$ and not accepted by $\dpcore_2[\treewidthvalue]$. By definition of $P_1=\dpcoregraphproperty{\dpcore_1}$, this implies $\decompositiongraph{\abstractdecomposition}\in P_1$. Moreover, if $\decompositiongraph{\abstractdecomposition}\in P_2=\dpcoregraphproperty{\dpcore_2}$ then coherence of $\dpcore_2$ would imply that $\abstractdecomposition$ is accepted by $\dpcore_2[\treewidthvalue]$, a contradiction. Therefore $\decompositiongraph{\abstractdecomposition}\notin P_2$.

	For completeness, assume there exists $\abstractdecomposition$ with $\decompositiongraph{\abstractdecomposition}\in P_1$ and $\decompositiongraph{\abstractdecomposition}\notin P_2$. Since $\dpcore_1$ and $\dpcore_2$ are coherent, $\abstractdecomposition$ is accepted by $\dpcore_1[\treewidthvalue]$ and not accepted by $\dpcore_2[\treewidthvalue]$. Consider the bottom-up computation of $\dynamizationfunction{\dpcore_1}{\treewidthvalue}$ and $\dynamizationfunction{\dpcore_2}{\treewidthvalue}$ on $\abstractdecomposition$. For each subterm $\sigmaabstractdecomposition$ of $\abstractdecomposition$, let
	\[
	T(\sigmaabstractdecomposition)\defeq \bigl(\topbag{\sigmaabstractdecomposition},\,\dynamizationfunction{\dpcore_1}{\treewidthvalue}(\sigmaabstractdecomposition),\,\dynamizationfunction{\dpcore_2}{\treewidthvalue}(\sigmaabstractdecomposition)\bigr).
	\]
	By Definition~\ref{definition:Dynamization}, these triples form a reachability derivation to the root triple $T(\abstractdecomposition)$, which is inconsistent.

			Finally, every intermediate triple satisfies the premise: if $\sigmaabstractdecomposition$ is a subterm of $\abstractdecomposition$, then Lemma~\ref{lemma:subterm-subgraph} yields that $\decompositiongraph{\sigmaabstractdecomposition}$ is isomorphic to a subgraph of $\decompositiongraph{\abstractdecomposition}$. Since $P_1$ is closed under subgraphs, graph properties are taken up to isomorphism, and $\decompositiongraph{\abstractdecomposition}\in P_1$, we have $\decompositiongraph{\sigmaabstractdecomposition}\in P_1$. Equivalently, $\decompositiongraph{\sigmaabstractdecomposition}\in \dpcoregraphproperty{\dpcore_1}$, so coherence of $\dpcore_1$ implies that $\sigmaabstractdecomposition$ is accepted by $\dpcore_1[\treewidthvalue]$. Hence $\dynamizationfunction{\dpcore_1}{\treewidthvalue}(\sigmaabstractdecomposition)$ contains a final witness, i.e., the premise holds at every intermediate state.
			Algorithm~\ref{algorithm:PremisePrunedInclusionTest} enumerates all successors of every dequeued premise-satisfying state and records each generated state at most once. Therefore the premise-satisfying derivation above is eventually generated, and the algorithm returns an inconsistent trace. Conversely, if the algorithm returns such a trace, the soundness argument gives a counterexample. If the finite queue is exhausted, no reachable inconsistent state exists, and by completeness no counterexample exists.
\end{proof}

\subsection{Canonized Premise Search}
\label{subsection:canonized-premise-search}

The implementation used in the experiments combines the premise-pruned transition system above with the canonization procedure of Section~\ref{subsec:symmetry-search}. We formalize this combined version as follows. Assume that $\dpcore_1$ and $\dpcore_2$ are interface-respecting DP-cores with identity cleaning and with witness actions. For a relabeling function $f:\abag\to[\treewidthvalue+1]$, define the product action on triples by
\[
f\cdot(\abag,\witnessset_1,\witnessset_2)
\defeq
\bigl(f(\abag),\action_{\dpcore_1}^{\treewidthvalue}(f,\witnessset_1),
\action_{\dpcore_2}^{\treewidthvalue}(f,\witnessset_2)\bigr).
\]
Fix total orders on $\dpcore_1[\treewidthvalue].\allwitnesses$ and $\dpcore_2[\treewidthvalue].\allwitnesses$, and order triples lexicographically by their bag vector and the two witness-set vectors. For every well-formed triple, define
\[
\canonic_{\dpcore_1,\dpcore_2}^{k}(\abag,\witnessset_1,\witnessset_2)
\defeq
\min\{f\cdot(\abag,\witnessset_1,\witnessset_2)\mid f\in\relabelclass_{\treewidthvalue},\ \domain(f)=\abag\}.
\]
Because the two component actions preserve final witnesses, both premise satisfaction and inconsistency are invariant under this canonical replacement.

\begin{lemma}[Product canonical representatives are orbit-invariant]
\label{lemma:product-canonical-orbit-invariant}
Let $(\abag,\witnessset_1,\witnessset_2)$ be a well-formed triple and let $f:\abag\to[\treewidthvalue+1]$ be injective. If
\[
(\abag',\witnessset'_1,\witnessset'_2)=f\cdot(\abag,\witnessset_1,\witnessset_2),
\]
then
\[
\canonic_{\dpcore_1,\dpcore_2}^{k}(\abag',\witnessset'_1,\witnessset'_2)
=
\canonic_{\dpcore_1,\dpcore_2}^{k}(\abag,\witnessset_1,\witnessset_2).
\]
\end{lemma}
\begin{proof}
This is the componentwise version of Lemma~\ref{lemma:canonical-orbit-invariant}. For every relabeling $h:\abag'\to[\treewidthvalue+1]$, the product action of $h$ on the primed triple equals the product action of $h\circ f$ on the original triple, by composition in both component witness actions. Conversely, every relabeling $g:\abag\to[\treewidthvalue+1]$ is obtained from the primed side as $g\circ f^{-1}$, using inverse relabeling componentwise. Hence the two sets of candidate triples in the definition of $\canonic_{\dpcore_1,\dpcore_2}^{k}$ are equal, and so their minima are equal.
\end{proof}

\begin{lemma}[Product relabeled traces are realizable]
\label{lemma:product-relabeled-traces-realizable}
Let $\dpcore_1$ and $\dpcore_2$ be DP-cores with identity cleaning and witness actions.
Consider a finite well-formed product trace
\[
(\abag_0,\witnessset_{1,0},\witnessset_{2,0}),\dots,
(\abag_m,\witnessset_{1,m},\witnessset_{2,m})
\]
obtained by the componentwise relabeled rules of Definition~\ref{definition:RelabeledRefutation}, using the same constructor, indices, relabeling, and join-bag permutation in both components at each step.
Then, for every relabeling $g:\abag_m\to[\treewidthvalue+1]$, there is a $\treewidthvalue$-instructive tree decomposition $\abstractdecomposition$ such that
\[
\topbag{\abstractdecomposition}=g(\abag_m),
\]
and, for each $q\in\{1,2\}$,
\[
\dynamizationfunction{\dpcore_q}{\treewidthvalue}(\abstractdecomposition)
=
\action_{\dpcore_q}^{\treewidthvalue}(g,\witnessset_{q,m}).
\]
\end{lemma}
\begin{proof}
This is Lemma~\ref{lemma:tree decomposition-existence-sequence-relabel-semi-refutation} applied simultaneously to the two components.
The product trace uses one underlying constructor at each step, so the induction in that lemma constructs the same instructive decomposition for both components.
The well-definedness checks for relabelings and join permutations are componentwise: support containment follows in each component exactly as in Lemma~\ref{lemma:tree decomposition-existence-sequence-relabel-semi-refutation}, and the same relabeling $g$ is applied to the common final bag.
Identity cleaning ensures that the constructed subterm witness sets coincide with the dynamizations of $\dpcore_1$ and $\dpcore_2$.
\end{proof}

The \emph{ISO-BFS-premise} search is Algorithm~\ref{algorithm:PremisePrunedInclusionTest} with the following two modifications. First, the initial triple and every generated successor triple are replaced by their canonical representative under $\canonic_{\dpcore_1,\dpcore_2}^{k}$ before lookup in $Seen$ and insertion in $Y$. Second, in the join case the algorithm ranges over all earlier processed triples with the same bag, all permutations of the shared bag, and both argument orders, exactly as in Algorithm~\ref{algorithm:InclusionTestCombination}, applying the corresponding relabeling to the second join argument before the componentwise join. The pathwidth version is obtained by omitting join successors.

\begin{theorem}[Correctness of ISO-BFS-premise]
\label{theorem:iso-bfs-premise-correctness}
Assume that $\dpcore_1$ and $\dpcore_2$ are finite, coherent, interface-respecting DP-cores with identity cleaning and witness actions, and that $P_1=\dpcoregraphproperty{\dpcore_1}$ is closed under subgraphs. Then ISO-BFS-premise terminates. Moreover, it returns ``Inclusion Holds'' if and only if there is no graph $G$ of treewidth at most $\treewidthvalue$ such that $G\in P_1$ and $G\notin P_2$; otherwise it returns a canonized premise-pruned trace ending in an inconsistent triple.
\end{theorem}
\begin{proof}
Finiteness gives only finitely many canonical triples, and the algorithm inserts each of them into $Seen$ at most once, so termination follows.

Soundness is the componentwise analogue of Lemma~\ref{lemma:algorithm-returns-canonized-refutation}. Every returned trace is obtained from the initial triple by applying one of the lifted DP-transitions to the two components and then applying a product-action relabeling; in the join case, the same bag permutation is used in both components. Thus the returned trace is a product relabeled trace in the sense of Lemma~\ref{lemma:product-relabeled-traces-realizable}. Applying that lemma to the identity relabeling on the final bag gives one $\treewidthvalue$-instructive tree decomposition whose two dynamizations are the two witness sets in the final triple. Since the final triple is inconsistent, this decomposition is accepted by $\dpcore_1$ and rejected by $\dpcore_2$. Coherence then gives a graph $G$ of treewidth at most $\treewidthvalue$ with $G\in P_1$ and $G\notin P_2$.

For completeness, suppose such a graph exists and choose a $\treewidthvalue$-instructive tree decomposition $\abstractdecomposition$ representing it. The proof of Theorem~\ref{theorem:premise-pruned-correctness} gives a premise-satisfying derivation of an inconsistent triple from the subterms of $\abstractdecomposition$. Canonizing each triple in that derivation preserves premise satisfaction and inconsistency, and the product-action commutation laws ensure that every canonized successor is among the successors enumerated by ISO-BFS-premise. Whenever the enumerated successor differs from the canonized derivation only by a relabeling of the interface, Lemma~\ref{lemma:product-canonical-orbit-invariant} identifies their canonical representatives. Thus the breadth-first exploration must eventually generate an inconsistent canonical triple unless it has already returned one. Therefore, if the algorithm exhausts the search and returns ``Inclusion Holds'', no such graph exists.
\end{proof}

\section{Reed's Conjecture Parameterized by Treewidth}
\label{ReedConjecture}

In this section, we provide a concrete example of how dynamic programming algorithms can be used to provide asymptotic upper bounds on the time complexity of verifying whether a given graph-theoretic conjecture is valid on the class of graphs of width at most $\treewidthvalue$. More specifically, we analyze the following well-known conjecture due to Reed \cite{reed1998omega}, which establishes an upper bound on the chromatic number $\chi(\agraph)$ of a triangle-free graph $\agraph$ in terms of the maximum degree $\maxDegree(\agraph)$ of $\agraph$.

\begin{conjecture}
\label{conjecture:Reed-triangle-free}
For any simple, triangle-free, undirected graph $\agraph$,
$\chromaticnumber(\agraph)\leq \left\lceil \frac{\maxDegree(\agraph) + 3}{2}\right\rceil. $
\end{conjecture}

Note that graphs of treewidth at most $\treewidthvalue$ are $(\treewidthvalue+1)$-colorable \cite{bodlaender2005equitable}. The following theorem due to Dvo{\v{r}}\'{a}k and Kawarabayashi establishes a better upper bound for the chromatic number of triangle-free graphs in terms of treewidth.

\begin{theorem}[\cite{dvovrak2017triangle}]
\label{theorem:treewidthTrianglFreeChromaticNumber}
For any triangle-free graph $\agraph$ of treewidth $\leq \treewidthvalue$, $\chromaticnumber(\agraph) \leq \left\lceil \frac{\treewidthvalue + 3}{2}\right\rceil$.
\end{theorem}
Therefore, in order to prove that every graph of treewidth at most $\treewidthvalue$ satisfies
Conjecture~\ref{conjecture:Reed-triangle-free}, it is enough to consider graphs whose maximum degree is at most $\treewidthvalue-1$: if $\maxDegree(\agraph)\ge \treewidthvalue$, then Theorem~\ref{theorem:treewidthTrianglFreeChromaticNumber} implies the conjecture. More precisely, it suffices to show that for each $s\in \{0,\dots,\treewidthvalue-1\}$, every graph of treewidth at most $\treewidthvalue$ and maximum degree at most $s$ has chromatic number at most $\lceil\frac{s + 3}{2}\rceil$. For particular values of $\treewidthvalue$ this range can sometimes be shortened further; for example, when $\treewidthvalue=5$, the case $s=4$ is already discharged by Theorem~\ref{theorem:treewidthTrianglFreeChromaticNumber}, since both bounds are equal to $4$.

Now, let $\chromaticNumberAtMostProperty{\numbercolors}$ denote the graph property consisting of all graphs that are $\numbercolors$-colorable, $\maxDegreeAtLeastProperty{\maxdegreeparameter}$
denote the graph property consisting of all graphs that have maximum degree at least $\maxdegreeparameter$, $\simpleCliqueNumberAtLeastProperty{\cliqueparameter}$ be the property consisting of all graphs that have clique number at least $\cliqueparameter$, and $\hasMultiEdgeProperty$ be the property consisting of all graphs that have some multiple edges. Let
$\reedProperty(s)$ be the graph property

\begin{equation}
\label{equation:ReedProperty}
\begin{aligned}
\reedProperty(s) \equiv{}&
\bigl(\neg \hasMultiEdgeProperty
\wedge \neg \simpleCliqueNumberAtLeastProperty{3}\\
&{}\wedge \neg \maxDegreeAtLeastProperty{s+1}\bigr)\\
&{}\rightarrow \chromaticNumberAtMostProperty{\lceil (s+3)/2 \rceil}.
\end{aligned}
\end{equation}

Here, $\neg \hasMultiEdgeProperty$ enforces simplicity and $\neg \simpleCliqueNumberAtLeastProperty{3}$ enforces triangle-freeness.
Moreover, $\neg \maxDegreeAtLeastProperty{s+1}$ is equivalent to $\maxDegree(\agraph)\le s$, and the conclusion asserts $\lceil (s+3)/2 \rceil$-colorability.
The premise $\neg \hasMultiEdgeProperty \wedge \neg \simpleCliqueNumberAtLeastProperty{3}\wedge \neg \maxDegreeAtLeastProperty{s+1}$ is closed under subgraphs, so the early-pruning procedure from Section~\ref{section:EarlyPruning} applies when searching for counterexamples to $\reedProperty(s)$.
Indeed, subgraphs of simple graphs are simple, subgraphs cannot create a triangle, and taking subgraphs cannot increase the maximum degree.

Then determining whether all graphs of treewidth at most $\treewidthvalue$ satisfy Reed's conjecture is equivalent to determining whether for each $s\in \{0,1,\dots, \treewidthvalue-1\}$,
the set of all graphs of treewidth at most $\treewidthvalue$
is contained in property $\reedProperty(s)$. Now,
$\reedProperty(s)$ is a Boolean
combination of four properties
represented by finite DP-cores of polynomial bit-length in $\treewidthvalue$ (see Theorem~\ref{theorem:CoreExistence} below). The simple-clique core is used only together with the simplicity mask $\neg\hasMultiEdgeProperty$, which is the setting in which its correctness statement applies. The specification of each of the four DP-cores can be found in the appendix. Such implementations are of independent interest, because each such core may be viewed as an algorithm that takes a $\treewidthvalue$-instructive tree decomposition $\abstractdecomposition$ as input and decides whether the graph $\decompositiongraph{\abstractdecomposition}$ associated with $\abstractdecomposition$ satisfies the corresponding property represented by the core.

%
%
\begin{theorem}[DP-Cores for Reed's Conjecture]
\label{theorem:CoreExistence}
For each $\maxdegreeparameter,\numbercolors\in \Nplus$ and each $\cliqueparameter\in \Nplus$ with $\cliqueparameter\ge 2$, there exist finite, interface-respecting instructive DP-cores with the cleaning functions specified in their definitions and witness actions, denoted
\[
\begin{gathered}
\maxDegreeAtLeastCore{\maxdegreeparameter},\quad
\chromaticnumberAtMostCore{\numbercolors},\\
\simpleCliqueNumberAtLeastCore{\cliqueparameter},\quad
\hasMultiEdgeCore,
\end{gathered}
\]
satisfying the following properties.
	\begin{enumerate}
		\item $\maxDegreeAtLeastCore{\maxdegreeparameter}$ has bit-length $O(\treewidthvalue\cdot \log(\maxdegreeparameter+1))$ and
	$\dpcoregraphproperty{\maxDegreeAtLeastCore{\maxdegreeparameter}} = \maxDegreeAtLeastProperty{\maxdegreeparameter}$.
			\item $\chromaticnumberAtMostCore{\numbercolors}$ has bit-length $O(\treewidthvalue\cdot\log(\numbercolors+1))$ and
		\[
		\dpcoregraphproperty{\chromaticnumberAtMostCore{\numbercolors}}
		= \chromaticNumberAtMostProperty{\numbercolors}.
		\]
	\item $\simpleCliqueNumberAtLeastCore{\cliqueparameter}$ has bit-length  $O(\treewidthvalue\cdot \log \cliqueparameter)$ and its defined graph property agrees with $\simpleCliqueNumberAtLeastProperty{\cliqueparameter}$ after intersecting both sides with $\neg \hasMultiEdgeProperty$.
	\item $\hasMultiEdgeCore$ has bit-length  $O(\treewidthvalue^2)$ and $\dpcoregraphproperty{\hasMultiEdgeCore} = \hasMultiEdgeProperty$
	\end{enumerate}
	The cores in items 1, 2, and 4 are coherent. The simple-clique core is used below only under the mask $\neg\hasMultiEdgeProperty$, and item 3 is the coherent simple-graph property needed for Reed's formula.
	\end{theorem}
\begin{proof}
Take the four cores defined in Section~\ref{section:ChromaticNumberAtMost} and in Appendix~\ref{appendix:DPCores}. Each definition is an instructive DP-core; the cleaning function is the identity except for the simple-clique core, whose cleaning function keeps only the found witness once such a witness is present. The transition rules mention only the current active labels, remove forgotten labels from the local data, and use label sets as the support of local witnesses; hence the cores are interface-respecting.

For $\maxDegreeAtLeastCore{\maxdegreeparameter}$, Corollary~\ref{corollary:CorrectnessMaxDegree} gives
$\dpcoregraphproperty{\maxDegreeAtLeastCore{\maxdegreeparameter}}=\maxDegreeAtLeastProperty{\maxdegreeparameter}$, and Observation~\ref{observation:ComplexityMeasuresMaxDegree} gives bit-length $O(\treewidthvalue\cdot\log(\maxdegreeparameter+1))$. Lemma~\ref{lemma:witness-action-maxdegree} gives the witness action. Coherence follows from the correctness corollary because maximum degree is invariant under graph isomorphism.

For $\chromaticnumberAtMostCore{\numbercolors}$, Corollary~\ref{corollary:CorrectnessColorability} gives
\[
\dpcoregraphproperty{\chromaticnumberAtMostCore{\numbercolors}}
= \chromaticNumberAtMostProperty{\numbercolors}.
\]
Observation~\ref{observation:ComplexityMeasuresColorability} gives bit-length $O(\treewidthvalue\cdot\log(\numbercolors+1))$. Lemma~\ref{lemma:witness-action-colorability} gives the witness action. Coherence follows because $\numbercolors$-colorability is invariant under graph isomorphism.

For $\simpleCliqueNumberAtLeastCore{\cliqueparameter}$, Corollary~\ref{corollary:CorrectnessSimpleCliqueNumberAtLeast} proves correctness on simple graphs, and this is exactly the equality after intersecting both sides with $\neg\hasMultiEdgeProperty$. Observation~\ref{observation:ComplexityMeasuresClique} gives bit-length $O(\treewidthvalue\cdot\log\cliqueparameter)$, and Lemma~\ref{lemma:witness-action-clique} gives the witness action. The simple-graph restriction is essential here: the Reed formula uses this core only together with $\neg\hasMultiEdgeProperty$, so the stated property equality is the one needed in Equation~\eqref{equation:ReedProperty}. Coherence on this restricted use follows because clique existence in simple graphs is invariant under graph isomorphism.

For $\hasMultiEdgeCore$, Corollary~\ref{corollary:CorrectnessMultiEdge} gives
$\dpcoregraphproperty{\hasMultiEdgeCore}=\hasMultiEdgeProperty$, Observation~\ref{observation:ComplexityMeasuresHasMultiEdge} gives bit-length $O(\treewidthvalue^2)$, and Lemma~\ref{lemma:witness-action-multiedge} gives the witness action. Coherence follows because the existence of multiple edges is invariant under graph isomorphism. Finiteness of all four cores follows from the displayed bit-length bounds for fixed $\treewidthvalue$ and fixed parameters.
\end{proof}
\begin{lemma}[Identity-cleaning simple-clique variant]
\label{lemma:identity-cleaning-simple-clique}
Let $\simpleCliqueNumberAtLeastCore{\cliqueparameter}^{\mathrm{id}}$ be obtained from $\simpleCliqueNumberAtLeastCore{\cliqueparameter}$ by replacing its cleaning function by the identity function and leaving all transition rules unchanged.
For every $\treewidthvalue$-instructive tree decomposition $\abstractdecomposition$ whose graph is simple,
\begin{align*}
\abstractdecomposition
&\in\accepteddecompositions{\simpleCliqueNumberAtLeastCore{\cliqueparameter}[\treewidthvalue]}\\
&\Longleftrightarrow
\abstractdecomposition
\in\accepteddecompositions{\simpleCliqueNumberAtLeastCore{\cliqueparameter}^{\mathrm{id}}[\treewidthvalue]}.
\end{align*}
\end{lemma}
\begin{proof}
In the simple-clique core the only final witness is $\cliquefound$, and every transition preserves $\cliquefound$ once it appears.
The displayed cleaning function is the identity on witness sets not containing $\cliquefound$ and maps every witness set containing $\cliquefound$ to $\{\cliquefound\}$.
Therefore, by induction over the term, the cleaned and identity-cleaned computations contain $\cliquefound$ at exactly the same subterms.
Acceptance depends only on the presence of $\cliquefound$, so the two variants accept the same simple graphs.
\end{proof}

The simple-clique cleaning function above is the one used by the TreeWidzard implementation. For applications of the search theorems that are stated under identity cleaning, we use the identity-cleaning variant from Lemma~\ref{lemma:identity-cleaning-simple-clique}.

\begin{lemma}[Masked simple-clique coherence for Reed]
\label{lemma:masked-simple-clique-coherence}
For each $s\in\N$, the Boolean combination in Equation~\eqref{equation:ReedProperty} is coherent when the triangle-freeness conjunct is represented by $\simpleCliqueNumberAtLeastCore{3}$ only under the simultaneous mask $\neg\hasMultiEdgeProperty$.
\end{lemma}
\begin{proof}
Let two instructive decompositions represent isomorphic graphs. The properties
$\hasMultiEdgeProperty$, $\maxDegreeAtLeastProperty{s+1}$, and
$\chromaticNumberAtMostProperty{\lceil(s+3)/2\rceil}$ are invariant under graph
isomorphism by Corollaries~\ref{corollary:CorrectnessMultiEdge},
\ref{corollary:CorrectnessMaxDegree}, and~\ref{corollary:CorrectnessColorability}.
If the represented graph has a multiple edge, then the premise of
Equation~\eqref{equation:ReedProperty} is false because of the conjunct
$\neg\hasMultiEdgeProperty$, so the value of the simple-clique component is
irrelevant. If the represented graph has no multiple edge, then it is simple,
and Corollary~\ref{corollary:CorrectnessSimpleCliqueNumberAtLeast} identifies
the simple-clique core with the ordinary clique property on both isomorphic
graphs. Thus the whole Boolean formula has the same truth value on isomorphic
represented graphs.
\end{proof}

For the Reed searches, the four component cores are combined by the standard DP-combination construction of~\cite{de2023width}, applied to the identity-cleaning variants when a search theorem assumes identity cleaning. The construction uses product witness sets and componentwise witness actions, so it preserves finiteness, interface-respectingness, identity cleaning, and the stated bit-length bound up to the sum of the component bit-lengths. Coherence of the resulting combination follows from the component correctness statements above, Lemma~\ref{lemma:masked-simple-clique-coherence}, and the Boolean formula in Equation~\eqref{equation:ReedProperty}; the simple-clique component is evaluated only on graphs satisfying the simultaneous $\neg\hasMultiEdgeProperty$ premise.

Therefore, as a consequence of Theorem~\ref{theorem:ProvabilityPreserving}, Theorem~\ref{theorem:CoreExistence}, and Equation~\eqref{equation:ReedProperty} we have the following corollary.
\begin{corollary}
\label{corollary:ComplexityReed} For each $\treewidthvalue\in \N$, Reed's conjecture for triangle-free graphs can be tested in time double-exponential in $O(\treewidthvalue^2)$.
\end{corollary}
\begin{proof}
For a fixed $s\in\{0,\dots,\treewidthvalue-1\}$, the Reed DP-combination uses the four component cores from Theorem~\ref{theorem:CoreExistence} with
\[
\numbercolors=\left\lceil\frac{s+3}{2}\right\rceil\le \treewidthvalue+1.
\]
The total bit-length of one combined local witness is bounded by the sum of the component bounds:
\[
O(\treewidthvalue\log(s+2))+
O(\treewidthvalue\log(\numbercolors+1))+
O(\treewidthvalue\log 4)+
O(\treewidthvalue^2)
=O(\treewidthvalue^2).
\]
Hence the combined witness universe has size at most $2^{O(\treewidthvalue^2)}$, and a search state is a finite subset of this universe, so the number of states is at most $2^{2^{O(\treewidthvalue^2)}}$.
The finite alphabet operations, relabeling choices, and the factor of $\treewidthvalue$ choices for $s$ are absorbed in this bound.
Theorem~\ref{theorem:algorithm-correctness} and Theorem~\ref{theorem:iso-bfs-premise-correctness} therefore give a decision procedure within time $2^{2^{O(\treewidthvalue^2)}}$.
\end{proof}
Note that the DP-core $\hasMultiEdgeCore$ is deterministic, and therefore has multiplicity $1$ (Section~\ref{subsec:multiplicity}).

\section{Experimental Evaluation}
\label{section:Experiments}

We provide an experimental evaluation of our
approach. We have implemented our width-based automated theorem proving framework in a tool called \textsc{TreeWidzard}~\cite{treewidzard}. The tool provides an interface that facilitates the implementation of dynamic programming algorithms parameterized by treewidth and pathwidth, and their integration for width-based automated theorem proving. We evaluate the implementation on the task of producing counterexamples for false graph-theoretic statements, and on verifying Reed's conjecture on graphs of small treewidth and pathwidth.

In the implementation, the corresponding TreeWidzard property-file names are as follows:
\begin{center}
\begin{tabular}{ll}
$\chromaticnumberAtMostCore{\numbercolors}$ & \texttt{ChromaticNumber\_AtMost}\\
$\maxDegreeAtLeastCore{d}$ & \texttt{MaximumDegree\_AtLeast}\\
$\simpleCliqueNumberAtLeastCore{\omega}$ & \texttt{SimpleCliqueNumber\_AtLeast}\\
$\hasMultiEdgeCore$ & \texttt{HasMultipleEdges}
\end{tabular}
\end{center}

For reproducibility, the experimental artifacts corresponding to this full version (all generated counterexample instances, the corresponding instructive tree decompositions and refutation traces, and the scripts used to generate the experimental tables) are available at \url{https://github.com/mlgorithm/state-canonization}.

\subsection{Constructing Counterexamples to False Statements}
We evaluate the framework along three questions: \emph{Effectiveness}---can we automatically construct counterexamples for false statements under small width bounds? \emph{Efficiency}---how much do canonization and premise pruning reduce the explored state space (states, time, memory)? \emph{Bottlenecks}---which DP-cores dominate the search, measured via multiplicity? When the corresponding data are available, we compare four search strategies: BFS (baseline), ISO-BFS (canonization), BFS-premise (early pruning), and ISO-BFS-premise (both techniques combined).

\subsubsection{Chromatic Number vs Pathwidth}

As stated in Theorem~\ref{theorem:treewidthTrianglFreeChromaticNumber}, triangle-free graphs of treewidth at most $k$ have chromatic number at
most $\lceil (k+3)/2 \rceil$. Dvo{\v{r}}\'{a}k and Kawarabayashi~\cite{dvovrak2017triangle} show that this
bound is tight for every $k \ge 1$: for each $k$ there exists a triangle-free graph of treewidth at
most $k$ with chromatic number $\lceil (k+3)/2 \rceil$. Since the pathwidth of a graph is lower
bounded by its treewidth, this also implies that triangle-free graphs of pathwidth at most $k$
have chromatic number at most $\lceil (k+3)/2 \rceil$. For a fixed width bound we decrease the
right-hand side by one and search for a counterexample, producing a graph together with a
refutation trace whenever one exists.

In Figure \ref{fig:counterexample-graph-triangle-free-14-vertices}, we depict a triangle-free graph with $14$ vertices, $27$ edges and chromatic number $4$. The drawing is schematic; the adjacency list and the path decomposition are exact and were generated by the tool. The graph has pathwidth $4$ and is therefore a counterexample to the strengthened statement that triangle-free graphs of pathwidth at most $4$ have chromatic number at most $3$. Note that it follows from~\cite{dvovrak2017triangle} that this statement is true for graphs of pathwidth at most $3$, and our search terminates without counterexamples in this case.

\begin{figure}[H]
\centering
\begin{subfigure}{0.3\textwidth}
\centering
\includegraphics[width=\linewidth]{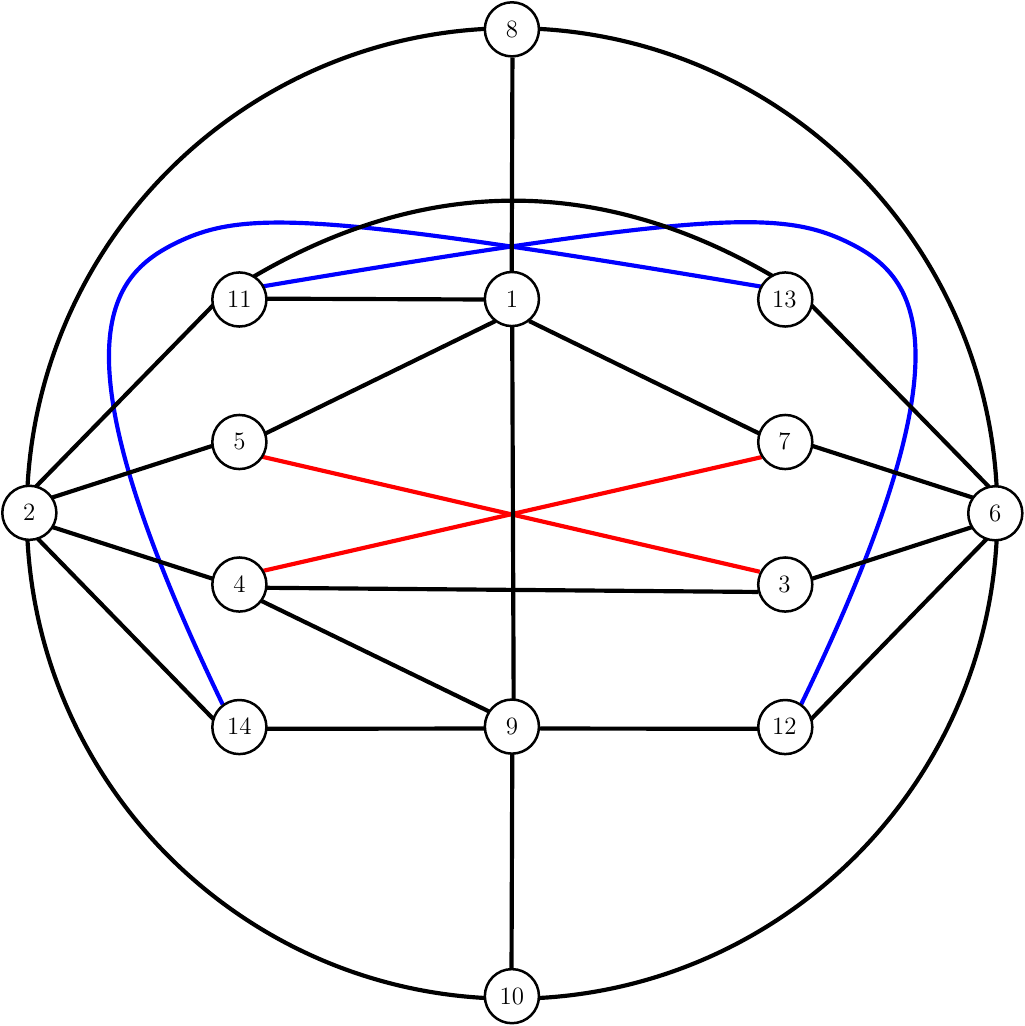}
\end{subfigure}
\hfill
\begin{subfigure}{0.3\textwidth}
\centering
{\scriptsize\setlength{\tabcolsep}{2pt}
\begin{tabular}{@{}r@{:\ }l@{}}
1  & [5,7,8,9,11] \\
2  & [4,5,8,10,11,14] \\
3  & [4,5,6] \\
4  & [2,3,7,9] \\
5  & [1,2,3] \\
6  & [3,7,8,10,12,13] \\
7  & [1,4,6] \\
8  & [1,2,6] \\
9  & [1,4,10,12,14] \\
10 & [2,6,9] \\
11 & [1,2,12,13] \\
12 & [6,9,11] \\
13 & [6,11,14] \\
14 & [2,9,13] \\
\end{tabular}}
\end{subfigure}
\hfill
\begin{subfigure}{0.3\textwidth}
{\scriptsize
\begin{tabular}{@{}l@{}}
$\{1\}$ \\
$\{1,2\}$ \\
$\{1,2,3\}$ \\
$\{1,2,3,4\}$ \\
$\{1,2,3,4,5\}$ \\
$\{1,2,3,4,6\}$ \\
$\{1,2,4,6,7\}$ \\
$\{1,2,4,6,8\}$ \\
$\{1,2,4,6,9\}$ \\
$\{1,2,6,9,10\}$ \\
$\{1,2,6,9,11\}$ \\
$\{2,6,9,11,12\}$ \\
$\{2,6,9,11,13\}$ \\
$\{2,9,13,14\}$ \\
\end{tabular}}
\end{subfigure}
\caption{A triangle-free graph with 14 vertices, pathwidth $4$, and chromatic number 4 (left), together with its adjacency list (middle), and a path decomposition of width 4 (right).}
\label{fig:counterexample-graph-triangle-free-14-vertices}
\end{figure}

\subsubsection{Chromatic Number vs Pathwidth Plus Maximum Degree}

Reed's conjecture for triangle-free graphs states the chromatic number of a graph of maximum degree $\Delta$ is at most $\lceil\frac{ \Delta+3}{2}\rceil$ \cite{reed1998omega}.
In Figure \ref{fig:counterexample-graph-triangle-free-22-vertices}, we depict a triangle-free graph of pathwidth $4$, maximum degree $4$, and chromatic number $4$. It is therefore a counterexample to the false strengthening that triangle-free graphs of maximum degree at most $4$ have chromatic number at most $3$, and it confirms that the bound in Reed's conjecture cannot be improved to $\lceil\frac{ \Delta+3}{2}\rceil-1$. Moreover, pathwidth $4$ is minimal for such a counterexample: by Theorem~\ref{theorem:treewidthTrianglFreeChromaticNumber}, every triangle-free graph of pathwidth at most $3$ is $3$-colorable.

\begin{figure}[H]
\centering
\begin{subfigure}{0.3\textwidth}
\centering

\includegraphics[width=\linewidth]{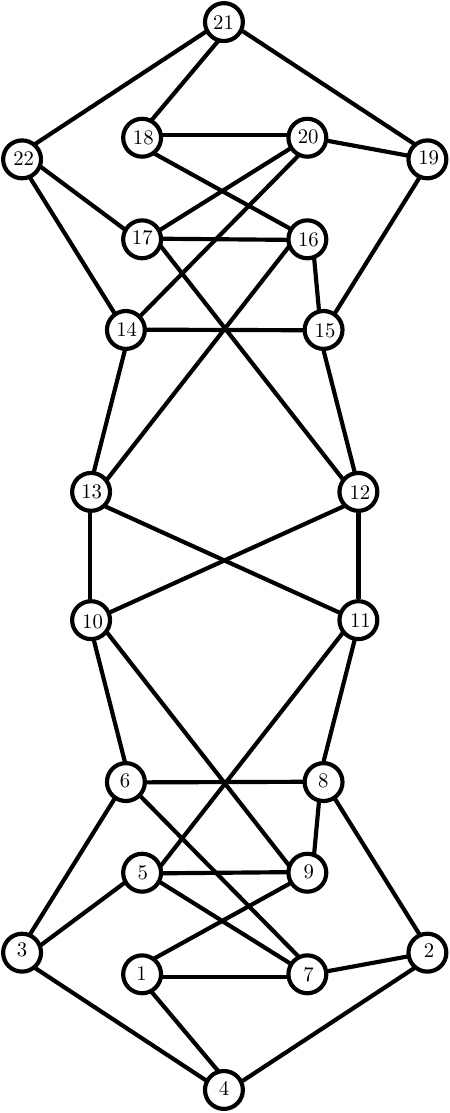}

\end{subfigure}
\hfill
\begin{subfigure}{0.3\textwidth}
\centering
{\scriptsize\setlength{\tabcolsep}{2pt}
\begin{tabular}{@{}r@{:\ }l@{}}
1  & [4,7,9] \\
2  & [4,7,8] \\
3  & [4,5,6] \\
4  & [1,2,3] \\
5  & [3,7,9,11] \\
6  & [3,7,8,10] \\
7  & [1,2,5,6] \\
8  & [2,6,9,11] \\
9  & [1,5,8,10] \\
10 & [6,9,12,13] \\
11 & [5,8,12,13] \\
12 & [10,11,15,17] \\
13 & [10,11,14,16] \\
14 & [13,15,20,22] \\
15 & [12,14,16,19] \\
16 & [13,15,17,18] \\
17 & [12,16,20,22] \\
18 & [16,20,21] \\
19 & [15,20,21] \\
20 & [14,17,18,19] \\
21 & [18,19,22] \\
22 & [14,17,21] \\
\end{tabular}}
	\end{subfigure}
	\hfill
	\begin{subfigure}{0.3\textwidth}
	\centering
	{\scriptsize
	\begin{tabular}{@{}l@{}}
	$\{1,2,3,4,5\}$ \\
	$\{1,2,3,5,6\}$ \\
	$\{1,2,5,6,7\}$ \\
	$\{1,2,5,6,8\}$ \\
	$\{1,5,6,8,9\}$ \\
	$\{5,6,8,9,10\}$ \\
	$\{5,8,10,11\}$ \\
	$\{10,11,12\}$ \\
	$\{10,11,12,13\}$ \\
	$\{12,13,14\}$ \\
	$\{12,13,14,15\}$ \\
	$\{12,13,14,15,16\}$ \\
	$\{12,14,15,16,17\}$ \\
	$\{14,15,16,17,18\}$ \\
	$\{14,15,17,18,19\}$ \\
	$\{14,17,18,19,20\}$ \\
	$\{14,17,18,19,21\}$ \\
	$\{14,17,21,22\}$ \\
	\end{tabular}}
	\end{subfigure}
\caption{A triangle-free graph with 22 vertices, pathwidth $4$, chromatic number 4, and maximum degree 4 (left), together with its adjacency list (middle), and a path decomposition of width 4 (right).
}
\label{fig:counterexample-graph-triangle-free-22-vertices}
\end{figure}

\subsection{Automatically Validating Reed's Triangle-Free Conjecture on Graphs of Small Pathwidth and Treewidth}

When testing a graph-theoretic conjecture on the class of graphs of pathwidth/treewidth at most $\treewidthvalue$, our deduction algorithm either produces a counterexample or exhausts the search space; in the latter case, the conjecture is valid on that class.

When combining both the symmetry-breaking technique introduced in Section~\ref{WBATPSymmetryBreaking} with the early-pruning technique introduced in Section~\ref{section:EarlyPruning}, our software targets Conjecture~\ref{conjecture:Reed-triangle-free} in the class of graphs of pathwidth at most $5$. This bounded case also follows from standard analytic ingredients: pathwidth at most $5$ implies treewidth at most $5$, so Theorem~\ref{theorem:treewidthTrianglFreeChromaticNumber} gives $\chromaticnumber(\agraph)\le 4$ for triangle-free graphs; for $\maxDegree(\agraph)\ge 4$ this is Reed's bound, while for $\maxDegree(\agraph)\le 3$ Brooks' theorem gives $\chromaticnumber(\agraph)\le 3$ \cite{Brooks1941}. Thus we use the pathwidth-$5$ case as a benchmark for the automated framework rather than as a novelty claim for the bounded theorem.
For the automated search in this case, it therefore suffices to consider graphs of maximum degree at most $3$.

In the rerun of the pathwidth-$5$, maximum-degree-$3$ configuration, the search using both techniques produced $746{,}187$ states, consumed about $29$ GB of memory, and terminated normally in approximately $41.4$ hours ($\sim 1.7$ days) on the Saga supercomputer, with $52$ CPUs allocated to a single task (see Table~\ref{table:conjecture-proof-experiments}). The search confirms that every triangle-free simple graph of pathwidth at most $5$ and maximum degree at most $3$ is $3$-colorable. Without canonization or premise pruning, this configuration remains prohibitively large within the budget allocated to this run.
\begin{table}[H]
    \centering
    \begin{tabular}{| c | c | c | c | c | c | c | c |}\hline
        pw & $\Delta$ & States & Memory & Time & Cores & Threads & Multiplicity\\ \hline
        5 & 3 & $746{,}187$ & $29$ GB & $41.4$ h & $52$ & $104$ & $122,1,52,1$\\ \hline
\end{tabular}
\caption{Rerun of Reed's conjecture on graphs of pathwidth at most $5$ and maximum degree $3$. The search terminated with the property satisfied. Multiplicity order: \texttt{ChromaticNumber\_AtMost}, \texttt{MaximumDegree\_AtLeast}, \texttt{SimpleCliqueNumber\_AtLeast}, \texttt{HasMultipleEdges} (see Section~\ref{subsec:multiplicity}).}
\label{table:conjecture-proof-experiments}
\end{table}

For graphs of pathwidth at most $4$, we performed a comparison between the original deduction algorithm introduced in \cite{de2023width} and the algorithm augmented only with the symmetry-breaking procedure, only with the early-pruning procedure, and with both (see Table~\ref{table:States}). For graphs of treewidth at most $4$, we report the two symmetry-aware variants used in our treewidth runs. Both improvements, when applied in isolation, significantly decrease the number of states considered during the search process. Nevertheless, when combining both approaches, the reduction of the search space was substantial. For example, in the case of pathwidth $3$ and maximum degree $2$, the method with no improvement produced about $24.5$ million states ($24{,}482{,}279$). The method with only the symmetry-breaking improvement produced about $1$ million states ($1{,}031{,}545$), the method with only the early-pruning improvement produced $1916$ states, and the method with both improvements produced only $141$ states. We provide comparative information on the number of states (Table~\ref{table:States}), time used (Table~\ref{table:Time}), and memory usage (Table~\ref{table:Memory}) for the reported algorithms.
In the tables below, \emph{States} denotes the number of distinct states inserted into the search's \texttt{Seen} set (after canonization, when enabled).
The \emph{Memory} column reports peak batch-task memory usage from the SLURM accounting logs. Entries for very small jobs with no detailed accounting block are reported below $1$ KB.

\noindent
Unless stated otherwise, experiments were run with a 24-hour timeout on the Saga supercomputer. Standard runs used about 160 GB of allocated memory. Runs that previously hit the memory cap were re-executed on expanded bigmem profiles with about 512 GB of allocated memory.

\begin{table}[H]
	    \centering
	    \resizebox{\textwidth}{!}{%
	    \begin{tabularx}{\textwidth}{|p{1.2cm}|X|X|X|X|X|X|X|X|X|X|}
	    \hline
	    \multicolumn{1}{|c}{ }& \multicolumn{ 3 }{c|}{\thead{\footnotesize{ BFS  } }}& \multicolumn{ 3 }{c|}{\thead{\footnotesize{ BFS-premise } }}\\ \hline
	    \diagbox[width=1.2cm]{\footnotesize pw}{\footnotesize $\Delta$} & \tiny{2} & \tiny{3} & \tiny{4}  & \tiny{2} & \tiny{3} & \tiny{4} \\ \hline
    \tiny{  pw = 2 } & \tiny{2503} & \tiny{4814} & \tiny{9877} & \tiny{149} & \tiny{341} & \tiny{617} \\ \hline
    \tiny{  pw = 3 } & \tiny{24482279} & \tiny{24908274} & \tiny{5193467} & \tiny{1916} & \tiny{9850} & \tiny{27720}\\ \hline
    \tiny{  pw = 4 } & \tiny{\textcolor{blue}{$>40982845^*$}} & \tiny{\textcolor{blue}{$>24741126^*$}} & \tiny{\textcolor{blue}{$>33517713^*$}} & \tiny{29184} & \tiny{1156954} & \tiny{2438694} \\ \hline
    \multicolumn{1}{|c}{ }& \multicolumn{ 3 }{c|}{\thead{\footnotesize{ ISO-BFS  } }}& \multicolumn{ 3 }{c|}{\thead{\footnotesize{ ISO-BFS-premise } }}\\ \hline
    \tiny{  pw = 2 } & \tiny{520} & \tiny{945} & \tiny{1843} & \tiny{38} & \tiny{76} & \tiny{128} \\ \hline
    \tiny{  pw = 3 } & \tiny{1031545} & \tiny{1050960} & \tiny{223920} & \tiny{141} & \tiny{579} & \tiny{1451} \\ \hline
    \tiny{  pw = 4 } & \tiny{\textcolor{blue}{$>9449990$}} & \tiny{\textcolor{blue}{$>9903864$}} & \tiny{\textcolor{blue}{$>7945250$}} & \tiny{486} & \tiny{12375} & \tiny{24494}  \\ \hline
    \multicolumn{1}{|c}{ }& \multicolumn{ 3 }{c|}{\thead{\footnotesize{ ISO-BFS  } }}& \multicolumn{ 3 }{c|}{\thead{\footnotesize{ ISO-BFS-premise } }}\\ \hline
    \tiny{  tw = 2 } & \tiny{\textcolor{blue}{$>$25351}} & \tiny{\textcolor{blue}{$>$25960}} & \tiny{\textcolor{blue}{$>$27632}} & \tiny{57} & \tiny{153} & \tiny{330} \\ \hline
	    \tiny{  tw = 3 } & \tiny{\textcolor{blue}{$>$69621}} & \tiny{\textcolor{blue}{$>$71183}} & \tiny{\textcolor{blue}{$>$72891}} & \tiny{194} & \tiny{1268} & \tiny{4080} \\ \hline
	    \tiny{  tw = 4 } & \tiny{\textcolor{blue}{$>$71858}} & \tiny{\textcolor{blue}{$>$73532}} & \tiny{\textcolor{blue}{$>$75426}} & \tiny{616} & \tiny{\textcolor{blue}{$>$4942}} & \tiny{\textcolor{blue}{$>$13438 }} \\ \hline
\end{tabularx}%
}
\caption{Number of states generated by our program when testing Reed's $(\omega,\Delta,\chi)$-conjecture for $\omega = 2$ (triangle-free case) on graphs of constant pathwidth and treewidth. Entries prefixed with $>$ indicate runs that did not terminate within the timeout, and entries marked with $^*$ indicate runs that exceeded the memory limit. }
\label{table:States}
\end{table}

Canonization and premise pruning are complementary, and their combination yields the largest reduction in explored states (e.g., for pathwidth $3$ and $\Delta=2$ the number of explored states drops from $24{,}482{,}279$ to $141$).

\begin{table}[H]
	    \centering
	    \resizebox{\textwidth}{!}{%
	    \begin{tabularx}{\textwidth}{|p{1.2cm}|X|X|X|X|X|X|X|}
	    \hline
	    \multicolumn{1}{|c|}{ }& \multicolumn{ 3 }{c|}{\thead{\footnotesize{ BFS  } }}& \multicolumn{ 3 }{c|}{\thead{\footnotesize{ BFS-premise } }}\\ \hline
	    \diagbox[width=1.2cm]{\footnotesize pw}{\footnotesize $\Delta$} & \tiny{2} & \tiny{3} & \tiny{4} & \tiny{2} & \tiny{3} & \tiny{4} \\ \hline
    \tiny{  pw = 2 } & \tiny{0.0 s} & \tiny{0.0 s} & \tiny{0.0 s}  & \tiny{0.0 s} & \tiny{0.0 s} & \tiny{0.0 s} \\ \hline
    \tiny{  pw = 3 } & \tiny{2548.0 s} & \tiny{2618.0 s} & \tiny{535.0 s} & \tiny{0.0 s} & \tiny{1.0 s} & \tiny{1.0 s}  \\ \hline
    \tiny{  pw = 4 } & \tiny{\textcolor{blue}{$>$1.2 h $^*$}} & \tiny{\textcolor{blue}{$>$1.1 h $^*$}} & \tiny{\textcolor{blue}{$>$1.2 h $^*$}} &  \tiny{4.0 s}& \tiny{178.0 s} & \tiny{476.0 s} \\ \hline
    %
    \multicolumn{1}{|c}{ }& \multicolumn{ 3 }{c|}{\thead{\footnotesize{ ISO-BFS  } }}& \multicolumn{ 3 }{c|}{\thead{\footnotesize{ ISO-BFS-premise } }}\\ \hline
    \tiny{  pw = 2 } & \tiny{0.0 s} & \tiny{0.0 s} & \tiny{0.0 s} & \tiny{0.0 s} & \tiny{0.0 s} & \tiny{0.0 s} \\ \hline
    \tiny{  pw = 3 } & \tiny{876.0 s} & \tiny{890.0 s} & \tiny{275.0 s}   & \tiny{0.0 s} & \tiny{1.0 s} & \tiny{1.0 s} \\ \hline
    \tiny{  pw = 4 } & \tiny{\textcolor{blue}{$>$24.0 h}} & \tiny{\textcolor{blue}{$>$24.0 h}} & \tiny{\textcolor{blue}{$>$24.0 h}}  & \tiny{5.0 s} & \tiny{151.0 s} & \tiny{349.0 s}  \\ \hline

    \multicolumn{1}{|c}{ }& \multicolumn{ 3 }{c|}{\thead{\footnotesize{ ISO-BFS  } }}& \multicolumn{ 3 }{c|}{\thead{\footnotesize{ ISO-BFS-premise } }}\\ \hline
	    \tiny{  tw = 2 } & \tiny{\textcolor{blue}{$>$24.0 h}} & \tiny{\textcolor{blue}{$>$24.0 h}} & \tiny{\textcolor{blue}{$>$24.0 h}} & \tiny{0.0 s} & \tiny{3.0 s} & \tiny{12.0 s} \\ \hline
	    \tiny{  tw = 3 } & \tiny{\textcolor{blue}{$>$24.0 h}} & \tiny{\textcolor{blue}{$>$24.0 h}} & \tiny{\textcolor{blue}{$>$24.0 h}} & \tiny{63.0 s} & \tiny{1.1 h} & \tiny{11.1 h} \\ \hline
	    \tiny{  tw = 4 } & \tiny{\textcolor{blue}{$>$24.0 h}} & \tiny{\textcolor{blue}{$>$24.0 h}} & \tiny{\textcolor{blue}{$>$24.0 h}} & \tiny{6.5 h} & \tiny{\textcolor{blue}{$>$24.0 h}} & \tiny{\textcolor{blue}{$>$24.0 h}} \\ \hline
	    \end{tabularx}%
	    }
\caption{Reed's Conjecture. Time required to test the conjecture for distinct values of maximum degree, pathwidth, and treewidth. Entries prefixed with $>$ indicate runs that did not terminate within the timeout, and entries marked with $^*$ indicate runs that exceeded the memory limit. Time format: h is an hour and s is a second. }
\label{table:Time}
\end{table}

The time savings mirror the reductions in state space: in several regimes, the baseline strategies hit time limits while the combined strategy terminates within seconds or minutes.

\begin{table}[H]
	    \centering
	    \resizebox{\textwidth}{!}{%
	    \begin{tabularx}{\textwidth}{|p{1.2cm}|X|X|X|X|X|X|X|}
	    \hline
	    \multicolumn{1}{|c|}{ }& \multicolumn{ 3 }{c|}{\thead{\footnotesize{ BFS  } }}& \multicolumn{ 3 }{c|}{\thead{\footnotesize{ BFS-premise } }}\\ \hline
	    \diagbox[width=1.2cm]{\footnotesize pw}{\footnotesize $\Delta$} & \tiny{2} & \tiny{3} & \tiny{4} & \tiny{2} & \tiny{3} & \tiny{4} \\ \hline
    \tiny{  pw = 2 } & \tiny{$<$ 1 KB} & \tiny{$<$ 1 KB} & \tiny{$<$ 1 KB}  & \tiny{$<$ 1 KB} & \tiny{$<$ 1 KB} & \tiny{$<$ 1 KB} \\ \hline
    \tiny{  pw = 3 } & \tiny{147.8 GB} & \tiny{151.4 GB} & \tiny{39.6 GB} & \tiny{$<$ 1 KB} & \tiny{$<$ 1 KB} & \tiny{0.2 GB}  \\ \hline
    \tiny{  pw = 4 } & \tiny{\textcolor{blue}{$>$512 GB $^*$}} & \tiny{\textcolor{blue}{$>$512 GB $^*$}} & \tiny{\textcolor{blue}{$>$512 GB $^*$}} &  \tiny{0.4 GB}& \tiny{17.6 GB} & \tiny{42.9 GB} \\ \hline
    %
    \multicolumn{1}{|c}{ }& \multicolumn{ 3 }{c|}{\thead{\footnotesize{ ISO-BFS  } }}& \multicolumn{ 3 }{c|}{\thead{\footnotesize{ ISO-BFS-premise } }}\\ \hline
    \tiny{  pw = 2 } & \tiny{$<$ 1 KB} & \tiny{$<$ 1 KB} & \tiny{$<$ 1 KB} & \tiny{$<$ 1 KB} & \tiny{$<$ 1 KB} & \tiny{$<$ 1 KB} \\ \hline
    \tiny{  pw = 3 } & \tiny{6.3 GB} & \tiny{6.4 GB} & \tiny{1.7 GB}   & \tiny{$<$ 1 KB} & \tiny{$<$ 1 KB} & \tiny{0 GB} \\ \hline
    \tiny{  pw = 4 } & \tiny{\textcolor{blue}{$>$173.6 GB}} & \tiny{\textcolor{blue}{$>$179.3 GB}} & \tiny{\textcolor{blue}{$>$178.1 GB}}  & \tiny{0 GB} & \tiny{0.2 GB} & \tiny{0.5 GB}  \\ \hline

    \multicolumn{1}{|c}{ }& \multicolumn{ 3 }{c|}{\thead{\footnotesize{ ISO-BFS  } }}& \multicolumn{ 3 }{c|}{\thead{\footnotesize{ ISO-BFS-premise } }}\\ \hline
    \tiny{  tw = 2 } & \tiny{\textcolor{blue}{$>$4.3 GB}} & \tiny{\textcolor{blue}{$>$4.3 GB}} & \tiny{\textcolor{blue}{$>$4.3 GB}} & \tiny{$<$ 1 KB} & \tiny{0 GB} & \tiny{0 GB} \\ \hline
	    \tiny{  tw = 3 } & \tiny{\textcolor{blue}{$>$76.9 GB}} & \tiny{\textcolor{blue}{$>$76.7 GB}} & \tiny{\textcolor{blue}{$>$78.7 GB}} & \tiny{0 GB} & \tiny{0 GB} & \tiny{0.1 GB} \\ \hline
	    \tiny{  tw = 4 } & \tiny{\textcolor{blue}{$>$91.3 GB}} & \tiny{\textcolor{blue}{$>$93 GB}} & \tiny{\textcolor{blue}{$>$96.1 GB}} & \tiny{0 GB} & \tiny{\textcolor{blue}{$>$0.2 GB}} & \tiny{\textcolor{blue}{$>$1 GB}} \\ \hline
	    \end{tabularx}%
	    }
\caption{Reed's Conjecture. Memory consumption to test the conjecture for distinct values of maximum degree, pathwidth, and treewidth. Entries prefixed with $>$ indicate runs that did not terminate within the timeout, and entries marked with $^*$ indicate runs that exceeded the memory limit.}
\label{table:Memory}
\end{table}

Canonization and premise pruning substantially reduce memory footprints, often turning runs that exceed memory limits into runs using less than a gigabyte.

\subsection{Multiplicity}\label{subsec:multiplicity}

The multiplicity of a DP-core with respect to a search strategy measures the maximum number of local witnesses produced by the core during the construction of states. Formally, for each explored state and each successor-generation step, we record the size of that component core's lifted transition output on the current witness set (or pair of witness sets for joins), and the table entry is the maximum of this quantity over the run. Thus the reported values are state-level lifted-output maxima, not merely $\max_w|\delta(w)|$ over single witnesses. Table~\ref{table:Multiplicity} reports these multiplicities. Each entry has four components corresponding to the multiplicities of the cores \texttt{ChromaticNumber\_AtMost}, \texttt{MaximumDegree\_AtLeast}, \texttt{SimpleCliqueNumber\_AtLeast}, and \texttt{HasMultipleEdges}, respectively. Note that \texttt{MaximumDegree\_AtLeast} and \texttt{HasMultipleEdges} are deterministic DP-cores, and therefore the multiplicity of these cores is $1$ regardless of the value of treewidth/pathwidth.
Higher multiplicity directly increases the branching factor in the state graph (more successor witnesses per step), which amplifies the benefit of canonization and premise pruning.

\begin{table}[H]
\centering
\resizebox{\textwidth}{!}{%
\begin{tabularx}{\textwidth}{|p{1.2cm}|X|X|X|X|X|X|X|}
\hline
\multicolumn{1}{|c|}{ }& \multicolumn{ 3 }{c|}{\thead{\footnotesize{ BFS  } }}& \multicolumn{ 3 }{c|}{\thead{\footnotesize{ BFS-premise } }}\\ \hline
\diagbox[width=1.2cm]{\footnotesize pw}{\footnotesize $\Delta$} & \tiny{2} & \tiny{3} & \tiny{4} & \tiny{2} & \tiny{3} & \tiny{4} \\ \hline
\tiny{  pw = 2 } & \tiny{5,1,12,1} & \tiny{5,1,12,1} & \tiny{5,1,12,1}  & \tiny{5,1,9,1} & \tiny{5,1,9,1} & \tiny{5,1,9,1} \\ \hline
\tiny{  pw = 3 } & \tiny{14,1,26,1} & \tiny{14,1,26,1} & \tiny{15,1,26,1} & \tiny{14,1,18,1} & \tiny{14,1,18,1} & \tiny{15,1,18,1}  \\ \hline
\tiny{  pw = 4 } & \tiny{\textcolor{blue}{$>$41,1,38,1 $^*$}} & \tiny{\textcolor{blue}{$>$41,1,38,1 $^*$}} & \tiny{\textcolor{blue}{$>$51,1,38,1 $^*$}} &  \tiny{41,1,30,1}& \tiny{41,1,32,1} & \tiny{51,1,32,1} \\ \hline
%
\multicolumn{1}{|c}{ }& \multicolumn{ 3 }{c|}{\thead{\footnotesize{ ISO-BFS  } }}& \multicolumn{ 3 }{c|}{\thead{\footnotesize{ ISO-BFS-premise } }}\\ \hline
\tiny{  pw = 2 } & \tiny{5,1,12,1} & \tiny{5,1,12,1} & \tiny{5,1,12,1} & \tiny{5,1,9,1} & \tiny{5,1,9,1} & \tiny{5,1,9,1} \\ \hline
\tiny{  pw = 3 } & \tiny{14,1,26,1} & \tiny{14,1,26,1} & \tiny{15,1,26,1}   & \tiny{14,1,18,1} & \tiny{14,1,18,1} & \tiny{15,1,18,1} \\ \hline
\tiny{  pw = 4 } & \tiny{\textcolor{blue}{$>$41,1,42,1}} & \tiny{\textcolor{blue}{$>$41,1,42,1}} & \tiny{\textcolor{blue}{$>$51,1,42,1}}  & \tiny{41,1,30,1} & \tiny{41,1,32,1} & \tiny{51,1,32,1}  \\ \hline

\multicolumn{1}{|c}{ }& \multicolumn{ 3 }{c|}{\thead{\footnotesize{ ISO-BFS  } }}& \multicolumn{ 3 }{c|}{\thead{\footnotesize{ ISO-BFS-premise } }}\\ \hline
\tiny{  tw = 2 } & \tiny{\textcolor{blue}{$>$5,1,12,1}} & \tiny{\textcolor{blue}{$>$5,1,12,1}} & \tiny{\textcolor{blue}{$>$5,1,12,1}} & \tiny{5,1,11,1} & \tiny{5,1,11,1} & \tiny{5,1,11,1} \\ \hline
\tiny{  tw = 3 } & \tiny{\textcolor{blue}{$>$14,1,19,1}} & \tiny{\textcolor{blue}{$>$14,1,19,1}} & \tiny{\textcolor{blue}{$>$15,1,19,1}} & \tiny{ 14,1,19,1} & \tiny{14,1,21,1} & \tiny{15,1,21,1} \\ \hline
\tiny{  tw = 4 } & \tiny{\textcolor{blue}{$>$41,1,29,1}} & \tiny{\textcolor{blue}{$>$41,1,29,1}} & \tiny{\textcolor{blue}{$>$51,1,29,1}} & \tiny{41,1,31,1} & \tiny{\textcolor{blue}{$>$41,1,32,1}} & \tiny{\textcolor{blue}{$>$51,1,35,1}} \\ \hline
\end{tabularx}%
}
\caption{Reed's Conjecture. Multiplicity: \texttt{ChromaticNumber\_AtMost}, \texttt{MaximumDegree\_AtLeast}, \texttt{SimpleCliqueNumber\_AtLeast}, \texttt{HasMultipleEdges}. Entries prefixed with $>$ indicate runs that did not terminate within the timeout, and entries marked with $^*$ indicate runs that exceeded the memory limit.}
\label{table:Multiplicity}
\end{table}

The multiplicity values show where branching occurs. It is concentrated in the chromatic-number and simple-clique cores. The maximum-degree and multi-edge cores have multiplicity $1$.

\section{Conclusion}
\label{section:Conclusion}

In this work, we took initial steps toward evaluating the width-based automated theorem proving approach introduced in~\cite{de2023width} in practice. We implemented this approach in \textsc{TreeWidzard} and introduced two complementary search reductions: state canonization and early pruning for implications with subgraph-closed premises. Together, these techniques drastically reduce the explored state space in practice. Our experiments on coloring statements for triangle-free graphs demonstrate orders-of-magnitude reductions, automatically generate counterexamples to false strengthenings, and enable automated bounded-width validation of Reed's triangle-free conjecture on graphs of pathwidth at most $5$ and treewidth at most $3$.

It is worth highlighting the modularity of our approach. While the implementation of instructive dynamic programming cores requires specialized knowledge on the part of the programmer, the use of such cores once they have been implemented is straightforward. For instance, in our framework, the critical case for testing Reed's conjecture on graphs of pathwidth at most $5$ is the case where the degree is at most $3$. More precisely, we test the conjecture that all simple, triangle-free graphs of maximum degree at most $3$ have chromatic number at most $3$. This conjecture is stated in TreeWidzard property-file syntax as follows, where the first four lines each correspond to a graph property, and the last line corresponds to the conjecture being tested.

{\small
\begin{verbatim}
x := MaximumDegree_AtLeast(4)
y := SimpleCliqueNumber_AtLeast(3)
z := HasMultipleEdges()
w := ChromaticNumber_AtMost(3)
Formula
NOT x AND NOT y AND NOT z IMPLIES w
\end{verbatim}
}

The idea is that dynamic programming cores deciding graph properties are implemented as plugins that need to be implemented only once by a specialist and then used without difficulty by graph theorists. We believe that our approach has the potential to create a productive exchange of ideas between the community of researchers working on parameterized complexity theory and researchers working in automated theorem proving. In essence, our framework shows that three decades of accumulated knowledge obtained in the development of faster width-based parameterized algorithms for model checking may be put into use in the context of automated theorem proving.

\FloatBarrier

\section*{Acknowledgements}
The computations were performed on the Saga supercomputer, on resources provided by Sigma2~--~the National Infrastructure for High-Performance Computing and Data Storage in Norway (\url{https://www.sigma2.no/}), under project \texttt{nn9535k}.

\bibliographystyle{plain}
\bibliography{bibliography}
	\clearpage
		\appendix
		\clearpage
\refstepcounter{section}
\section*{Summary of Notation}
\addcontentsline{toc}{section}{\protect\numberline{\thesection}Summary of Notation}
\label{appendix:notation-quickref}

\begin{table}[H]
\centering
\footnotesize
\begin{tabularx}{\textwidth}{@{}lX@{}}
\hline
\textbf{Symbol} & \textbf{Meaning} \\
\hline
$\N$ & natural numbers $\{0,1,2,\dots\}$ \\
$\Nplus$ & positive natural numbers $\N\setminus\{0\}$ \\
$[n]$ & $\{1,\dots,n\}$ (with $[0]=\emptyset$) \\
$\finitepowerset{X}$, $\powerset{X}$ & finite subsets / all subsets of $X$ \\
$G=(\vertexset{G},\edgeset{G},\incidencerelation{G})$ & graph with vertex set $\vertexset{G}$, edge set $\edgeset{G}$, incidence relation $\incidencerelation{G}$ \\
$\vertexset{G}$, $\edgeset{G}$, $\incidencerelation{G}$ & vertex set / edge set / incidence relation of $G$ \\
$|G|$ & $|\vertexset{G}|+|\edgeset{G}|$ \\
$\emptyset$ & empty graph $(\emptyset,\emptyset,\emptyset)$ \\
$G\uplus H$ & disjoint union (Section~\ref{section:Preliminaries}) \\
$G\sim H$ & graphs $G$ and $H$ are isomorphic \\
$H$ embeds as a subgraph of $G$ & there exists an embedding $(\phi,\nu)$ (Section~\ref{section:Preliminaries}) \\
$P\subseteq \allgraphs$ & graph property (closed under isomorphism) \\
$\isomorphismclosure{\mathcal{S}}$ & isomorphism closure of a set of graphs $\mathcal{S}$ \\
$\neg P$, $P\wedge Q$, $P\vee Q$, $P\to Q$ & Boolean operations on graph properties (Section~\ref{section:Preliminaries}) \\
\hline
\end{tabularx}
\caption{Core graph and set notation.}
\label{table:notation}
\end{table}

\begin{table}[H]
\centering
\footnotesize
\begin{tabularx}{\textwidth}{@{}lX@{}}
\hline
\textbf{Symbol} & \textbf{Meaning} \\
\hline
$(\Sigma,\arity)$ & ranked alphabet and arity function \\
$\Terms{\Sigma}$ & terms over ranked alphabet $\Sigma$ \\
$\abstractalphabet{k}$ & $k$-instructive alphabet (Equation~\eqref{equation:InstructiveAlphabet}) \\
$\leaftype$ & leaf constructor (arity 0) \\
$\introvertextype{u}$, $\forgetvertextype{u}$ & introduce/forget vertex with interface label $u$ (arity 1) \\
$\introedgetype{u}{v}$ & introduce an edge between active labels $u$ and $v$ (arity 1) \\
$\jointype$ & join constructor (arity 2) \\
$\treeautomaton_k$ & automaton recognizing legal $k$-instructive terms (Section~\ref{section:InstructiveTreeDecompositions}) \\
$\lang(\treeautomaton_k)$ & set of terms accepted by $\treeautomaton_k$ \\
$\allabstractdecompositionstreewidth{k}$ & $k$-instructive decompositions (legal terms over $\abstractalphabet{k}$) \\
$\allabstractdecompositionspathwidth{k}$ & join-free $k$-instructive decompositions (instructive path decompositions) \\
$\decompositiongraph{\tau}$ & graph constructed by a decomposition term $\tau$ \\
$B(\tau)$ & bag/interface labels active at the root of $\tau$ \\
$\theta[\tau]$ & map from active labels in $B(\tau)$ to their corresponding vertices in $\decompositiongraph{\tau}$ \\
$\allgraphstreewidth{k}$, $\allgraphspathwidth{k}$ & graphs of treewidth/pathwidth at most $k$ \\
\hline
\end{tabularx}
\caption{Terms and instructive decompositions.}
\label{table:notation-decompositions}
\end{table}

\begin{table}[H]
\centering
\footnotesize
\begin{tabularx}{\textwidth}{@{}lX@{}}
\hline
\textbf{Symbol} & \textbf{Meaning} \\
\hline
$D[k]$ & DP-core specialized to width bound $k$ \\
$W_k$ & witness universe of $D[k]$ (a decidable subset of $\{0,1\}^\ast$) \\
$D[k].\initialsetgeneric$ & initial witness set at leaves \\
$D[k].\finalwitnessgenericcore$ & final predicate on witnesses \\
$\introvertexgeneric{\cdot}$, $\forgetvertexgeneric{\cdot}$, $\introedgegeneric{\cdot}{\cdot}$, $\joingenericcore$ & DP transitions (lifted to witness sets by union) \\
$\dynamizationfunction{\dpcore}{k}(\tau)$ & witness set computed by DP-core $\dpcore[k]$ on term $\tau$ (Definition~\ref{definition:Dynamization}) \\
$\dpcoregraphproperty{D}$ & graph property defined by DP-core $D$ \\
$(b,S)$ & DP-state: bag $b\subseteq[k+1]$ and a finite witness set $S$ \\
inconsistent state & witness set contains no final witness \\
$(k,D)$-refutation & reachability trace ending in an inconsistent state (Definition~\ref{definition:DPRefutation}) \\
$\beta_D(k,n)$ & max bit-length of useful witnesses of $D$ at width $k$ and term size $n$ \\
\hline
\end{tabularx}
\caption{DP-cores, dynamization, and refutations.}
\label{table:notation-dpcores}
\end{table}

\begin{table}[H]
\centering
\footnotesize
\begin{tabularx}{\textwidth}{@{}lX@{}}
\hline
\textbf{Symbol} & \textbf{Meaning} \\
\hline
$\relabelclass_k$ & $k$-relabeling functions $f:b\to[k+1]$ (injective, with variable domain $b$) \\
$\mathrm{Lbl}_k(w)$, $\mathrm{Lbl}_k(S)$ & interface labels mentioned by witness $w$ and by witness set $S$ \\
$\action_{\dpcore}^{k}(f,S)$ & witness action of relabeling $f$ applied to witnesses/witness sets \\
$\mathrm{Perm}(b)$ & permutations of a fixed bag $b$ \\
$\canonic_{\dpcore}^{k}(b,S)$ & canonical representative of the relabeling orbit of $(b,S)$ \\
$P_1\to P_2$ & implication property; premise pruning applies when $P_1$ is subgraph-closed \\
\hline
\end{tabularx}
\caption{Relabelings, canonization, and premise pruning.}
\label{table:notation-canonization}
\end{table}

\FloatBarrier

		\section[Graph Associated with a k-Instructive Tree Decomposition]{Graph Associated with a $\treewidthvalue$-Instructive Tree Decomposition}
\label{appendix:GraphFromDecomposition}

This appendix makes explicit the semantics of the constructors in the instructive alphabet by defining, for each $\treewidthvalue$-instructive decomposition $\abstractdecomposition$, the associated graph $\decompositiongraph{\abstractdecomposition}$, the current interface bag $\topbag{\abstractdecomposition}$, and the interface map $\topmap{\abstractdecomposition}$.

\begin{definition}[Constructor semantics for $\decompositiongraph{\cdot}$, $\topbag{\cdot}$, and $\topmap{\cdot}$]
\label{definition:itd-constructor-semantics}
Let $\treewidthvalue\in \N$ and let $\abstractdecomposition\in \allabstractdecompositionstreewidth{\treewidthvalue}$.
\begin{enumerate}
\setlength\itemsep{0em}
\item[(a)] If $\abstractdecomposition = \leaftype$, then $\decompositiongraph{\abstractdecomposition}$ is the empty graph, $\topbag{\abstractdecomposition}=\emptyset$, and $\topmap{\abstractdecomposition}$ is the empty function.
\item[(b)] If $\abstractdecomposition = \introvertextype{\vertexone}(\sigmaabstractdecomposition)$, then $\vertexone\notin \topbag{\sigmaabstractdecomposition}$, $\topbag{\abstractdecomposition}=\topbag{\sigmaabstractdecomposition}\cup\{\vertexone\}$, and $\decompositiongraph{\abstractdecomposition}$ is obtained from $\decompositiongraph{\sigmaabstractdecomposition}$ by adding one new isolated vertex $x\defeq \min(\Nplus\setminus \vertexset{\decompositiongraph{\sigmaabstractdecomposition}})$. Moreover, $\topmap{\abstractdecomposition}|_{\topbag{\sigmaabstractdecomposition}}=\topmap{\sigmaabstractdecomposition}$ and $\topmap{\abstractdecomposition}(\vertexone)=x$.
\item[(c)] If $\abstractdecomposition = \forgetvertextype{\vertexone}(\sigmaabstractdecomposition)$, then $\vertexone\in \topbag{\sigmaabstractdecomposition}$, $\topbag{\abstractdecomposition}=\topbag{\sigmaabstractdecomposition}\setminus\{\vertexone\}$, $\decompositiongraph{\abstractdecomposition}=\decompositiongraph{\sigmaabstractdecomposition}$, and $\topmap{\abstractdecomposition}=\topmap{\sigmaabstractdecomposition}|_{\topbag{\abstractdecomposition}}$.
\item[(d)] If $\abstractdecomposition = \introedgetype{\vertexone}{\vertextwo}(\sigmaabstractdecomposition)$, then $\vertexone\neq\vertextwo$, $\{\vertexone,\vertextwo\}\subseteq \topbag{\sigmaabstractdecomposition}$, $\topbag{\abstractdecomposition}=\topbag{\sigmaabstractdecomposition}$, $\topmap{\abstractdecomposition}=\topmap{\sigmaabstractdecomposition}$, and $\decompositiongraph{\abstractdecomposition}$ is obtained from $\decompositiongraph{\sigmaabstractdecomposition}$ by adding one new edge $e\defeq \min(\Nplus\setminus \edgeset{\decompositiongraph{\sigmaabstractdecomposition}})$ with endpoints $\{\topmap{\sigmaabstractdecomposition}(\vertexone),\topmap{\sigmaabstractdecomposition}(\vertextwo)\}$.
\item[(e)] If $\abstractdecomposition = \jointype(\sigmaabstractdecomposition_1,\sigmaabstractdecomposition_2)$, then $\topbag{\sigmaabstractdecomposition_1}=\topbag{\sigmaabstractdecomposition_2}=\topbag{\abstractdecomposition}\defeq \abag$. Let $G_i\defeq \decompositiongraph{\sigmaabstractdecomposition_i}$ and $\theta_i\defeq \topmap{\sigmaabstractdecomposition_i}$ for $i\in\{1,2\}$, and let $H\defeq G_1\uplus G_2$ be their disjoint union (Section~\ref{section:Preliminaries}). Define $\mu:\vertexset{H}\to \N$ by
\[
		\mu(x)\defeq
		\begin{cases}
		2\theta_1(u) & \text{if } x = 2\theta_2(u)+1 \text{ for some } u\in \abag,\\
		x & \text{otherwise.}
		\end{cases}
		\]
		By Lemma~\ref{lemma:itd-interface-map-injective} below, $\theta_2$ is injective on $\abag$, hence the label $u$ is unique whenever the first case applies and $\mu$ is well-defined.
			We define $\decompositiongraph{\abstractdecomposition}$ as the graph with vertex set $\mu(\vertexset{H})$, edge set $\edgeset{H}$, and incidence relation $\{(e,\mu(v)):(e,v)\in \incidencerelation{H}\}$, and set $\topmap{\abstractdecomposition}(u)\defeq 2\theta_1(u)$ for each $u\in \abag$. In particular, each child graph embeds as a subgraph of $\decompositiongraph{\abstractdecomposition}$ (Lemma~\ref{lemma:join-child-embeddings}), and for each $u\in \abag$ the identified vertex is $\topmap{\abstractdecomposition}(u)$.
		In other words, $\jointype$ glues the two child graphs along the shared interface labels $\abag$, while keeping all non-interface vertices disjoint and preserving all edges from both children.
		\end{enumerate}
\end{definition}

\begin{lemma}[Interface map injectivity]
\label{lemma:itd-interface-map-injective}
Let $\treewidthvalue\in \N$ and let $\abstractdecomposition\in \allabstractdecompositionstreewidth{\treewidthvalue}$.
Then the interface map $\topmap{\abstractdecomposition}:\topbag{\abstractdecomposition}\to \vertexset{\decompositiongraph{\abstractdecomposition}}$ is injective.
\end{lemma}
\begin{proof}
We proceed by structural induction on $\abstractdecomposition$.
The claim is trivial for $\leaftype$, since the domain is empty.
For $\introvertextype{\vertexone}(\sigmaabstractdecomposition)$, the map $\topmap{\abstractdecomposition}$ extends $\topmap{\sigmaabstractdecomposition}$ and maps the fresh label $\vertexone$ to the fresh vertex $x\notin \vertexset{\decompositiongraph{\sigmaabstractdecomposition}}$; since $\topmap{\sigmaabstractdecomposition}(\topbag{\sigmaabstractdecomposition})\subseteq \vertexset{\decompositiongraph{\sigmaabstractdecomposition}}$, we have $x\notin \mathrm{im}(\topmap{\sigmaabstractdecomposition})$, hence injectivity is preserved.
For $\forgetvertextype{\vertexone}(\sigmaabstractdecomposition)$, the map is a restriction of $\topmap{\sigmaabstractdecomposition}$, so injectivity is preserved.
For $\introedgetype{\vertexone}{\vertextwo}(\sigmaabstractdecomposition)$, the interface map is unchanged.
Finally, for $\jointype(\sigmaabstractdecomposition_1,\sigmaabstractdecomposition_2)$ we have $\topmap{\abstractdecomposition}(u)=2\theta_1(u)$ for $u\in \topbag{\abstractdecomposition}$, where $\theta_1=\topmap{\sigmaabstractdecomposition_1}$ is injective by the induction hypothesis; multiplying by $2$ preserves injectivity, so $\topmap{\abstractdecomposition}$ is injective.
\end{proof}

\begin{lemma}[Join preserves the graph axioms]
\label{lemma:join-preserves-graph-axioms}
Let $\treewidthvalue\in \N$ and let $\abstractdecomposition=\jointype(\sigmaabstractdecomposition_1,\sigmaabstractdecomposition_2)\in \allabstractdecompositionstreewidth{\treewidthvalue}$.
Let $G_i\defeq \decompositiongraph{\sigmaabstractdecomposition_i}$ and $\theta_i\defeq \topmap{\sigmaabstractdecomposition_i}$ for $i\in\{1,2\}$, and let $H=G_1\uplus G_2$ and $\mu$ be as in Definition~\ref{definition:itd-constructor-semantics}(e).
Then for every edge $\anedge\in \edgeset{H}$, the set
\[
\{\,\mu(v) : (\anedge,v)\in \incidencerelation{H}\,\}
\]
has cardinality $2$. In particular, $\decompositiongraph{\abstractdecomposition}$ is a graph in the sense of Section~\ref{section:Preliminaries} (no loops are created by the identification).
\end{lemma}
\begin{proof}
Let $\anedge\in \edgeset{H}$. Since $H$ is the disjoint union of two renamed copies, the edge $\anedge$ belongs to exactly one child copy.
If $\anedge\in \edgeset{G_1^{(0)}}$, then both endpoints of $\anedge$ are even vertices and $\mu$ fixes all even vertices, so the two endpoints remain distinct.
Otherwise $\anedge\in \edgeset{G_2^{(1)}}$, hence both endpoints are odd vertices.
If neither endpoint is of the form $2\theta_2(u)+1$ with $u\in \topbag{\abstractdecomposition}$, then $\mu$ fixes both endpoints and they remain distinct.
If exactly one endpoint is of that form, then one endpoint maps to an even vertex and the other remains odd, so they are distinct.
If both endpoints are of that form, say $2\theta_2(u)+1$ and $2\theta_2(v)+1$, then $u\neq v$ (otherwise the endpoints would be equal), and $\mu$ maps them to $2\theta_1(u)$ and $2\theta_1(v)$.
By Lemma~\ref{lemma:itd-interface-map-injective}, $\theta_1$ is injective on the bag, hence $\theta_1(u)\neq \theta_1(v)$ and therefore $2\theta_1(u)\neq 2\theta_1(v)$.
Thus in all cases the two endpoints of $\anedge$ remain distinct under $\mu$, as claimed.
\end{proof}

\noindent
We choose the smallest unused vertex/edge identifiers in (b) and (d) to make the construction canonical (i.e., deterministic rather than only defined up to isomorphism).

The join constructor glues two independently constructed graphs along a shared interface bag.
Formally, we (i) take the disjoint union $H=G_1\uplus G_2$, and then (ii) identify, for each interface label $u$, the two copies of the interface vertex $\theta_1(u)$ and $\theta_2(u)$ via the map $\mu$.
Only interface vertices are identified, and edges are never merged; consequently, if both children contain an edge between the same pair of interface vertices then the join creates a multi-edge, while vertices outside the interface remain disjoint across children (Lemma~\ref{lemma:join-no-cross-edges}).

\begin{lemma}[Child embeddings in joins]
\label{lemma:join-child-embeddings}
Let $\treewidthvalue\in \N$ and let $\abstractdecomposition=\jointype(\sigmaabstractdecomposition_1,\sigmaabstractdecomposition_2)\in \allabstractdecompositionstreewidth{\treewidthvalue}$.
Write $G\defeq \decompositiongraph{\abstractdecomposition}$ and $G_i\defeq \decompositiongraph{\sigmaabstractdecomposition_i}$ for $i\in\{1,2\}$, and let $\mu$ be as in Definition~\ref{definition:itd-constructor-semantics}(e).
Then each child graph embeds as a subgraph of $G$:
\begin{enumerate}
\item The maps $\phi_1:\vertexset{G_1}\to \vertexset{G}$ and $\nu_1:\edgeset{G_1}\to \edgeset{G}$ defined by $\phi_1(v)\defeq 2v$ and $\nu_1(e)\defeq 2e$ form an embedding of $G_1$ into $G$.
\item The maps $\phi_2:\vertexset{G_2}\to \vertexset{G}$ and $\nu_2:\edgeset{G_2}\to \edgeset{G}$ defined by $\phi_2(v)\defeq \mu(2v+1)$ and $\nu_2(e)\defeq 2e+1$ form an embedding of $G_2$ into $G$.
\end{enumerate}
Moreover, for each $u\in \topbag{\abstractdecomposition}$ we have $\phi_1(\topmap{\sigmaabstractdecomposition_1}(u))=\phi_2(\topmap{\sigmaabstractdecomposition_2}(u))=\topmap{\abstractdecomposition}(u)$.
\end{lemma}
\begin{proof}
Let $H\defeq G_1\uplus G_2$ be the disjoint union. By the disjoint union construction (Section~\ref{section:Preliminaries}), the left copy $G_1^{(0)}$ uses even vertex/edge identifiers and the right copy $G_2^{(1)}$ uses odd ones. By Definition~\ref{definition:itd-constructor-semantics}(e), the graph $G$ has edge set $\edgeset{H}$ and incidence relation
\[
\incidencerelation{G}=\{(e,\mu(v)):(e,v)\in \incidencerelation{H}\}.
\]

\emph{(1)} The maps $\phi_1$ and $\nu_1$ are injective. Moreover, since $\mu$ fixes all even vertices, for each $(e,v)\in \incidencerelation{G_1}$ we have $(2e,2v)\in \incidencerelation{H}$ and hence $(2e,\mu(2v))=(\nu_1(e),\phi_1(v))\in \incidencerelation{G}$. Conversely, if $(\nu_1(e),\phi_1(v))=(2e,2v)\in \incidencerelation{G}$, then $(2e,2v)\in \incidencerelation{H}$, which implies $(e,v)\in \incidencerelation{G_1}$ by the definition of the disjoint union. Thus $(\phi_1,\nu_1)$ is an embedding.

\emph{(2)} The map $\nu_2$ is injective. For $\phi_2$, note that $\mu$ fixes odd vertices that are not interface vertices, and maps each odd interface vertex $2\topmap{\sigmaabstractdecomposition_2}(u)+1$ to the even vertex $2\topmap{\sigmaabstractdecomposition_1}(u)$. Distinct odd vertices in $\vertexset{G_2^{(1)}}$ therefore have distinct $\mu$-images, so $\phi_2$ is injective. Now, for each $(e,v)\in \incidencerelation{G_2}$ we have $(2e+1,2v+1)\in \incidencerelation{H}$ and hence $(2e+1,\mu(2v+1))=(\nu_2(e),\phi_2(v))\in \incidencerelation{G}$. Conversely, if $(\nu_2(e),\phi_2(v))=(2e+1,\mu(2v+1))\in \incidencerelation{G}$, then by definition of $\incidencerelation{G}$ we have $(2e+1,2v+1)\in \incidencerelation{H}$, which implies $(e,v)\in \incidencerelation{G_2}$. Thus $(\phi_2,\nu_2)$ is an embedding.

For the final claim, let $u\in \topbag{\abstractdecomposition}$. Then $\phi_1(\topmap{\sigmaabstractdecomposition_1}(u))=2\topmap{\sigmaabstractdecomposition_1}(u)=\topmap{\abstractdecomposition}(u)$ by Definition~\ref{definition:itd-constructor-semantics}(e). Also $\phi_2(\topmap{\sigmaabstractdecomposition_2}(u))=\mu(2\topmap{\sigmaabstractdecomposition_2}(u)+1)=2\topmap{\sigmaabstractdecomposition_1}(u)=\topmap{\abstractdecomposition}(u)$ by the definition of $\mu$.
\end{proof}

\begin{lemma}[No cross-edges in joins]
\label{lemma:join-no-cross-edges}
Let $\treewidthvalue\in \N$ and let $\abstractdecomposition=\jointype(\sigmaabstractdecomposition_1,\sigmaabstractdecomposition_2)\in \allabstractdecompositionstreewidth{\treewidthvalue}$.
Write $G\defeq \decompositiongraph{\abstractdecomposition}$ and $G_i\defeq \decompositiongraph{\sigmaabstractdecomposition_i}$ for $i\in\{1,2\}$, and let $\abag\defeq \topbag{\abstractdecomposition}$.
Let $\phi_i$ and $\nu_i$ be the vertex and edge maps of the embedding of $G_i$ into $G$ from Lemma~\ref{lemma:join-child-embeddings}.
Then every edge identifier of $G$ originates in exactly one child: for each $\anedge\in \edgeset{G}$, there is a unique $i\in\{1,2\}$ and a unique $\anedge_i\in\edgeset{G_i}$ such that $\anedge=\nu_i(\anedge_i)$.
For this child $i$, the endpoints $\edgeendpointsname(\anedge)$ in $G$ are contained in $\phi_i(\vertexset{G_i})$.
In particular, there is no edge of $G$ whose endpoints lie respectively in
\[
\phi_1(\vertexset{G_1}\setminus \topmap{\sigmaabstractdecomposition_1}(\abag))
\quad\text{and}\quad
\phi_2(\vertexset{G_2}\setminus \topmap{\sigmaabstractdecomposition_2}(\abag)).
\]
\end{lemma}
\begin{proof}
In the join construction of Definition~\ref{definition:itd-constructor-semantics}(e), the edge set of $G$ is the edge set of the disjoint union $H=G_1\uplus G_2$, and vertices are identified only by the map $\mu$ on interface vertices.
The disjoint union uses disjoint edge identifiers, and the join does not identify edge identifiers. Hence every edge $\anedge\in \edgeset{G}$ is uniquely of the form $\nu_i(\anedge_i)$ for one child edge $\anedge_i\in\edgeset{G_i}$.
Its incident vertices in $G$ are precisely the $\mu$-images of the incident vertices of that child edge, so both endpoints of $\anedge$ lie in the image of the corresponding child embedding. Since $\mu$ only identifies interface vertices, no edge can connect a non-interface vertex from the first child image to a non-interface vertex from the second child image.
\end{proof}

\begin{lemma}[Subterms induce subgraphs]
\label{lemma:subterm-subgraph}
Let $\treewidthvalue\in \N$, $\abstractdecomposition\in \allabstractdecompositionstreewidth{\treewidthvalue}$, and let $\sigmaabstractdecomposition$ be a subterm of $\abstractdecomposition$. Then $\decompositiongraph{\sigmaabstractdecomposition}$ is isomorphic to a subgraph of $\decompositiongraph{\abstractdecomposition}$.
\end{lemma}
\begin{proof}
We prove the statement by induction on the distance from the root of $\sigmaabstractdecomposition$ to the root of $\abstractdecomposition$ in the term tree. If $\sigmaabstractdecomposition=\abstractdecomposition$, the claim is immediate. Otherwise, let $\pi$ be the parent of $\sigmaabstractdecomposition$ in $\abstractdecomposition$. By Definition~\ref{definition:itd-constructor-semantics}, $\decompositiongraph{\sigmaabstractdecomposition}$ embeds as a subgraph of $\decompositiongraph{\pi}$ for each constructor case (and in the join case, the required embeddings are given explicitly by Lemma~\ref{lemma:join-child-embeddings}). By the induction hypothesis, $\decompositiongraph{\pi}$ embeds as a subgraph of $\decompositiongraph{\abstractdecomposition}$. Composing embeddings yields the claim.
\end{proof}

		\section{DP-Cores for Graph Properties}
\label{appendix:DPCores}

In this appendix, we define DP-cores for the graph properties used throughout the paper that are not already specified in the main text.
For each such graph property $\graphproperty$, we specify a DP-core and prove its correctness on instructive decompositions.
The concrete cores used for Reed's conjecture are denoted
$\maxDegreeAtLeastCore{\maxdegreeparameter}$,
$\simpleCliqueNumberAtLeastCore{\cliqueparameter}$, and
$\hasMultiEdgeCore$, matching the corresponding TreeWidzard property-file names.

Fix $k\in \N$ and let $\abstractdecomposition\in \allabstractdecompositionstreewidth{k}$ be a $k$-instructive decomposition (Section~\ref{section:InstructiveTreeDecompositions}).
Each DP-core is specified via a predicate of the form $\mathrm{Pred}_{\graphproperty}$ relating decompositions and local witnesses, and we prove predicate soundness.
As a consequence, for every $\abstractdecomposition\in \allabstractdecompositionstreewidth{k}$, acceptance by the corresponding width-$k$ core is equivalent to $\decompositiongraph{\abstractdecomposition}\in \graphproperty$ (DP-core soundness).

				\subsection{Maximum Degree}
{
\renewcommand{\dpcore}{\maxDegreeAtLeastCore{\maxdegreeparameter}}

Let $\agraph$ be a graph. The maximum degree $\maxDegree(\agraph)$ of $\agraph$ is the maximum degree among its vertices, i.e.,
\[
\maxDegree(\agraph)=\max_{\xvertex  \in \vertexset{\agraph}} d_\xvertex,
\]
where $d_\xvertex$ is the degree of the vertex $\xvertex$ in $\agraph$. For the empty graph, we use the convention that its maximum degree is $0$. We denote the maximum degree of $\agraph$ by $\maxDegree$ if $\agraph$ is clear from the context. Let $\maxDegreeAtLeastProperty{\maxdegreeparameter}$ be the graph property containing all graphs that have a maximum degree of at least $\maxdegreeparameter$. In the following, we define a DP-core $\dpcore$ that decides $\maxDegreeAtLeastProperty{\maxdegreeparameter}$. More specifically, for each $\treewidthvalue \in \N$, $\dpcore[\treewidthvalue]$ decides the graph property $\maxDegreeAtLeastProperty{\maxdegreeparameter} \cap \allgraphstreewidth{\treewidthvalue}$.
\\

\noindent{\textbf{Local Witness.}} We begin by defining the structure of local witnesses for $\dpcore[\treewidthvalue]$. The main idea is to track the degrees of active labels, which represent the active vertices of a graph, until we encounter a label with a degree of at least $\maxdegreeparameter$. Inactive vertices do not need to be tracked since no edges will be connected to them.

To implement this idea, we define a local witness to be either a map or a Boolean value. We define the map to record the degrees of active labels for each node of a $\treewidthvalue$-instructive tree decomposition. If a local witness is the Boolean value $\maxdegreefound$, it indicates that a label with a degree of at least $\maxdegreeparameter$ has been found.

\begin{definition}
A $\dpcore[\treewidthvalue]$ local witness is either a map $\awitness:\abag \to \{0,1,\dots,\maxdegreeparameter-1\}$, where $\abag \subseteq [\treewidthvalue+1]$, or the Boolean value $\awitness = \maxdegreefound$.
\end{definition}

Intuitively, $\abag$ represents the set of labels that are active in some node of a $\treewidthvalue$-instructive tree decomposition, and the Boolean value $\maxdegreefound$ is a Boolean flag specifying that a vertex of maximum degree at least $\maxdegreeparameter$ was found.

\noindent{\textbf{All Witnesses}.} We let $\dpcore[\treewidthvalue].\allwitnesses$ be the set of all $\dpcore[\treewidthvalue]$-local witnesses:
\[
\dpcore[\treewidthvalue].\allwitnesses = \{\awitness : \awitness \text{ is a } \dpcore[\treewidthvalue]\text{ local witness}\}.
\]

\noindent{\textbf{Initial Witnesses}.} When the instruction $\initialsetgenericcore$ is called, it creates an empty graph. The only possible local witness is the empty map, $\awitness_0:\emptyset\rightarrow \emptyset$, since the empty graph has no vertices.
\[
\dpcore[\treewidthvalue].\initialsetgeneric = \{\awitness_0\}.
\]
\\

\noindent{\textbf{Introduce Vertex}.} Let $\awitness$ be a $\dpcore[\treewidthvalue]$-local witness, and let $\vertexone$ represent a label in $[\treewidthvalue + 1]$. When a new vertex is introduced and labeled with $\vertexone$, we evaluate the situation based on the state of the local witness.

If a vertex with a degree of at least $\maxdegreeparameter$ has already been identified, i.e., $\awitness = \maxdegreefound$, there is no need to modify the local witness, as the desired condition has already been met. However, if no such vertex has been found, i.e., $\awitness$ is a map, the label $\vertexone$ is added to the map $\awitness$ with an initial value of 0. This reflects the fact that the vertex represented by $\vertexone$ currently has a degree of 0, as no edges are connected to it at this stage.

\[
\dpcore[\treewidthvalue].\introvertexgeneric{\vertexone}(\awitness) = \begin{cases}
    \{\awitness\} & \quad \text {if } \awitness = \maxdegreefound, \\[6pt]
    \emptyset & \quad \begin{array}[t]{@{}l@{}}
    \text{if } \awitness\neq \maxdegreefound\\
    \text{and } \vertexone\in\domain(\awitness),
    \end{array}\\[6pt]
    \{\awitness \cup \{\vertexone \to 0\}\} & \quad \text{otherwise.}
\end{cases}
\]
\\

\noindent{\textbf{Forget Vertex}.} Let $\awitness$ be a local witness, and let $\vertexone$ represent an active label. When the vertex labeled $\vertexone$ is forgotten, the label $\vertexone$ is freed, making it available for reuse when introducing a future vertex. If $\awitness= \maxdegreefound$, it implies that a vertex meeting the required degree condition has already been identified. In this case, the local witness $\awitness$ remains unchanged, as no further modifications are needed. On the other hand, if $\awitness$ is a map, this process involves removing the label $\vertexone$ from the map $\awitness$, since the degree of the vertex with label $\vertexone$ will not increase anymore.
\[
\dpcore[\treewidthvalue].\forgetvertextype{\vertexone}(\awitness) = \begin{cases}
\{\awitness\} & \quad \text{if } \awitness=\maxdegreefound, \\[6pt]
\{\awitness|_{\domain(\awitness)\setminus\{\vertexone\}}\} & \quad \text{otherwise.}
\end{cases}
\]

\noindent{\textbf{Introduce Edge}.} When introducing an edge between two vertices, the first step is to determine whether a vertex with a degree of at least $\maxdegreeparameter$ has already been found. Let $\awitness$ be a local witness, and let $\vertexone$ and $\vertextwo$ be distinct active labels representing two vertices. If $\awitness = \maxdegreefound$, it indicates that a qualifying vertex has already been identified, and no further modifications to the local witness $\awitness$ are needed.

Otherwise (i.e., if $\awitness$ is a map), we proceed by incrementing the degrees of the vertices represented by $\vertexone$ and $\vertextwo$ in the map $\awitness$. If either endpoint is missing from the map domain, the local witness is incompatible with introducing an edge on these active labels and the transition returns no witness. We then check whether either $\awitness(\vertexone)$ or $\awitness(\vertextwo)$ reaches $\maxdegreeparameter$ after the increment. If this condition is met, it signifies that a vertex with the required degree has been found. In this case, the local witness will change to the Boolean value $\maxdegreefound$, as further tracking of degrees is unnecessary.

If neither vertex reaches $\maxdegreeparameter$, the degrees of $\vertexone$ and $\vertextwo$ are updated in $\awitness$, and the local witness is modified accordingly.

Formally, the operation $\dpcore[\treewidthvalue].\introedgegeneric{\vertexone}{\vertextwo}(\awitness)$ produces one of the following outcomes:

\begin{enumerate}
    \item If $\awitness=\maxdegreefound$, return $\{\awitness\}$.
    \item If $\awitness\neq\maxdegreefound$ and $\{\vertexone,\vertextwo\}\nsubseteq \domain(\awitness)$, return $\emptyset$.
    \item If either $\awitness(\vertexone)=\maxdegreeparameter-1$ or $\awitness(\vertextwo)=\maxdegreeparameter-1$, return $\{\maxdegreefound\}$.
    \item Otherwise, return
    \[
    \{(\awitness|_{\domain(\awitness)\setminus\{\vertexone,\vertextwo\}}\cup\{\vertexone\to \awitness(\vertexone)+1,\vertextwo\to\awitness(\vertextwo)+1\})\}.
    \]
\end{enumerate}

\noindent{\textbf{Join}.}
Let $\awitness$ and $\awitness'$ be local witnesses associated with the left and right children of a given node, respectively. The goal of the $\joingenericcore$ operation is to merge the graphs corresponding to these children along the vertices labeled with active labels.

The process begins by checking whether either child has already identified a vertex with a degree of at least $\maxdegreeparameter$. Specifically, if $\awitness = \maxdegreefound$ for the left child or $\awitness' = \maxdegreefound$ for the right child, the join operation immediately returns the default local witness $\maxdegreefound$, as no further processing is needed.

If neither child has identified such a vertex (i.e., both $\awitness$ and $\awitness'$ are maps), we first check that the two maps have the same domain. If the domains differ, the two implementation witnesses are incompatible and the join returns the empty set. Otherwise, we proceed by checking the sum of the values of $\awitness$ and $\awitness'$ for each active label. If the sum for any active label is at least $\maxdegreeparameter$, this indicates that a vertex with the required degree has been found. In this case, we again return the local witness $\maxdegreefound$.

If no such vertex is found, we construct a new map $\awitness''$ by summing the values of $\awitness$ and $\awitness'$ for each active label. This new map $\awitness''$ represents the combined state of the left and right children.

Formally, the operation $\dpcore[\treewidthvalue].\joingenericcore(\awitness,\awitness')$ produces one of the following outcomes:
\begin{enumerate}
    \item If $\awitness=\maxdegreefound$ or $\awitness' = \maxdegreefound$, return $\{\maxdegreefound\}$.
    \item If $\domain(\awitness)\neq\domain(\awitness')$, return $\emptyset$.
    \item If there exists $\vertexone  \in \domain(\awitness)$ such that $\awitness(\vertexone)+\awitness'(\vertexone)\geq \maxdegreeparameter$, return $\{\maxdegreefound\}$.
    \item Otherwise, return $\{\awitness''\}$, where for each $\vertexone  \in \domain(\awitness)$, $\awitness''(\vertexone)=\awitness(\vertexone)+\awitness'(\vertexone)$.
\end{enumerate}

\noindent{\textbf{Final Witness Function.}}
The final witness function is used to determine whether the condition of finding a vertex with a degree of at least $\maxdegreeparameter$ has been satisfied. This function checks if $\awitness = \maxdegreefound$.
If $\awitness = \maxdegreefound$, it confirms that a vertex with the required degree has been found, and the function outputs $1$. Otherwise, the function outputs $0$, indicating that the condition has not yet been met.

Formally:
\[
\dpcore[\treewidthvalue].\finalwitnessgenericcore(\awitness) = 1 \quad \text{if and only if} \quad \awitness =  \maxdegreefound.
\]

\noindent{\textbf{Cleaning Function}.} In our example, the cleaning function acts as the identity:
\[
\dpcore[\treewidthvalue].\cleaningfunctioncore(\witnessset)=\witnessset \text{ for every }
\witnessset \subseteq \dpcore[\treewidthvalue].\allwitnesses.
\]

\noindent{\textbf{Invariant Function}.} In our example, the invariant function is trivial. It sends a set of local witnesses $\witnessset$ to $1$ if and only if it contains a final local witness.

\[
\dpcore[\treewidthvalue].\invariantCore(\witnessset)=
\begin{cases}
1 & \quad \text{ if } \witnessset \text{ has a final local witness,} \\[6pt]
0 & \quad \text{ otherwise.}
\end{cases}
\]

The next step is to prove the correctness of the DP-core we just defined. For this, we define the following predicate relating $\treewidthvalue$-instructive tree decompositions with local witnesses.

\begin{definition}
\label{definition:MaxDegreePredicate}
We let $\mxdp$ be the predicate that is true on a pair $(\abstractdecomposition,\awitness)   \in \allabstractdecompositionstreewidth{\treewidthvalue}\times \maxDegreeAtLeastCore{\maxdegreeparameter}[\treewidthvalue].\allwitnesses$ if and only if one of the following happens:
\begin{enumerate}
    \item \label{definition:MaxDegreePredicate-one} $\awitness = \maxdegreefound$ and $\maxDegree(\decompositiongraph{\abstractdecomposition})\geq \maxdegreeparameter$.
    \item $\awitness\neq \maxdegreefound$, $\domain(\awitness)=\topbag{\abstractdecomposition}$, $\maxDegree(\decompositiongraph{\abstractdecomposition})< \maxdegreeparameter$, and for each $\vertexone  \in \topbag{\abstractdecomposition}$, $\awitness(\vertexone) = d_{\decompositiongraph{\abstractdecomposition}}(\topmap{\abstractdecomposition}(\vertexone))$.
\end{enumerate}
\end{definition}

The first condition requires that if the witness $\awitness = \maxdegreefound$, the graph associated with $\abstractdecomposition$ should have a maximum degree of at least $\maxdegreeparameter$. Otherwise, the maximum degree of the graph $\decompositiongraph{\abstractdecomposition}$ should be less than $\maxdegreeparameter$, the map domain should be exactly the active bag, and for each active label in $\topbag{\abstractdecomposition}$, say $\vertexone$, the vertex associated with that label in the graph $\decompositiongraph{\abstractdecomposition}$ should have the degree $\awitness(\vertexone)$.

\begin{proposition}
\label{proposition:PredicateMaxDegree}
For each $\abstractdecomposition  \in \allabstractdecompositionstreewidth{\treewidthvalue}$, a local witness $\awitness$ belongs to
$
\dynamizationfunction{\maxDegreeAtLeastCore{\maxdegreeparameter}}{\treewidthvalue}(\abstractdecomposition)
$
if and only if $\mxdp(\abstractdecomposition,\awitness)=\truevalue$.
\end{proposition}

\begin{proof}
We proceed by structural induction on $\abstractdecomposition$. Throughout, we use that the cleaning function of $\dpcore[\treewidthvalue]$ is the identity.

\medskip
\noindent\textbf{Base case:} $\abstractdecomposition=\leaftype$. By Definition~\ref{definition:Dynamization}, $\dynamizationfunction{\dpcore}{\treewidthvalue}(\leaftype)=\dpcore[\treewidthvalue].\initialsetgeneric=\{\awitness_0\}$, where $\awitness_0$ is the empty map. Hence $\awitness\in \dynamizationfunction{\dpcore}{\treewidthvalue}(\leaftype)$ if and only if $\awitness=\awitness_0$. Since $\decompositiongraph{\leaftype}$ is the empty graph and $\topbag{\leaftype}=\emptyset$, we have $\maxDegree(\decompositiongraph{\leaftype})=0<\maxdegreeparameter$, so $\mxdp(\leaftype,\awitness_0)=\truevalue$ and $\mxdp(\leaftype,\maxdegreefound)=\falsevalue$. Conversely, if $\mxdp(\leaftype,\awitness)=\truevalue$ and $\awitness\neq\maxdegreefound$, then Condition~2 forces $\domain(\awitness)=\emptyset$, so $\awitness=\awitness_0$.

\medskip
\noindent\textbf{Induction step:} assume the claim holds for the immediate subterm(s) of $\abstractdecomposition$.

\medskip
\noindent\textbf{Case 1:} $\abstractdecomposition=\introvertextype{\vertexone}(\sigmaabstractdecomposition)$.

$(\Rightarrow)$ Suppose $\awitness\in \dynamizationfunction{\dpcore}{\treewidthvalue}\allowbreak(\abstractdecomposition)$. Then $\awitness\in \introvertextype{\vertexone}(\dynamizationfunction{\dpcore}{\treewidthvalue}\allowbreak(\sigmaabstractdecomposition))$. Hence there exists $\awitness'\in \dynamizationfunction{\dpcore}{\treewidthvalue}\allowbreak(\sigmaabstractdecomposition)$ such that $\awitness\in \introvertextype{\vertexone}(\awitness')$. If $\awitness'=\maxdegreefound$, then $\awitness=\maxdegreefound$ and by the induction hypothesis $\maxDegree(\decompositiongraph{\sigmaabstractdecomposition})\ge \maxdegreeparameter$. By Definition~\ref{definition:itd-constructor-semantics}, introducing a vertex adds an isolated vertex, so $\maxDegree(\decompositiongraph{\abstractdecomposition})\ge \maxdegreeparameter$ and $\mxdp(\abstractdecomposition,\awitness)=\truevalue$. Otherwise $\awitness'$ is a map and $\awitness=\awitness'\cup\{\vertexone\to 0\}$. By the induction hypothesis, $\maxDegree(\decompositiongraph{\sigmaabstractdecomposition})<\maxdegreeparameter$ and $\awitness'$ records the degrees of all active vertices. By Definition~\ref{definition:itd-constructor-semantics}, $\vertexone$ maps to a new isolated vertex, so $\awitness$ records the degrees of all active vertices in $\decompositiongraph{\abstractdecomposition}$ and $\maxDegree(\decompositiongraph{\abstractdecomposition})<\maxdegreeparameter$, hence $\mxdp(\abstractdecomposition,\awitness)=\truevalue$.

$(\Leftarrow)$ Suppose $\mxdp(\abstractdecomposition,\awitness)=\truevalue$. If $\awitness=\maxdegreefound$, then $\maxDegree(\decompositiongraph{\abstractdecomposition})\ge \maxdegreeparameter$, and since introducing an isolated vertex does not decrease maximum degree, we have $\maxDegree(\decompositiongraph{\sigmaabstractdecomposition})\ge \maxdegreeparameter$. By the induction hypothesis, $\maxdegreefound\in \dynamizationfunction{\dpcore}{\treewidthvalue}\allowbreak(\sigmaabstractdecomposition)$, and therefore $\maxdegreefound\in \dynamizationfunction{\dpcore}{\treewidthvalue}\allowbreak(\abstractdecomposition)$. Otherwise $\awitness$ is a map with $\domain(\awitness)=\topbag{\abstractdecomposition}$ and $\maxDegree(\decompositiongraph{\abstractdecomposition})<\maxdegreeparameter$. Let $\awitness'\defeq \awitness|_{\topbag{\sigmaabstractdecomposition}}$. By Definition~\ref{definition:itd-constructor-semantics}, $\decompositiongraph{\sigmaabstractdecomposition}$ is obtained by deleting the new isolated vertex, and $\topbag{\abstractdecomposition}=\topbag{\sigmaabstractdecomposition}\cup\{\vertexone\}$ with the new vertex having degree $0$. Hence $\domain(\awitness')=\topbag{\sigmaabstractdecomposition}$, $\awitness'$ records the degrees of active vertices in $\decompositiongraph{\sigmaabstractdecomposition}$, and $\maxDegree(\decompositiongraph{\sigmaabstractdecomposition})<\maxdegreeparameter$. Thus $\mxdp(\sigmaabstractdecomposition,\awitness')=\truevalue$, and by the induction hypothesis $\awitness'\in \dynamizationfunction{\dpcore}{\treewidthvalue}\allowbreak(\sigmaabstractdecomposition)$. Since $\awitness\in \introvertextype{\vertexone}(\awitness')$, we conclude $\awitness\in \dynamizationfunction{\dpcore}{\treewidthvalue}\allowbreak(\abstractdecomposition)$.

\medskip
\noindent\textbf{Case 2:} $\abstractdecomposition=\forgetvertextype{\vertexone}(\sigmaabstractdecomposition)$.

$(\Rightarrow)$ Suppose $\awitness\in \dynamizationfunction{\dpcore}{\treewidthvalue}\allowbreak(\abstractdecomposition)$. Then $\awitness\in \forgetvertextype{\vertexone}(\dynamizationfunction{\dpcore}{\treewidthvalue}\allowbreak(\sigmaabstractdecomposition))$, so there exists $\awitness'\in \dynamizationfunction{\dpcore}{\treewidthvalue}\allowbreak(\sigmaabstractdecomposition)$ with $\awitness\in \forgetvertextype{\vertexone}(\awitness')$. If $\awitness'=\maxdegreefound$, then $\awitness=\maxdegreefound$ and $\maxDegree(\decompositiongraph{\sigmaabstractdecomposition})\ge \maxdegreeparameter$ by induction, hence also $\maxDegree(\decompositiongraph{\abstractdecomposition})\ge \maxdegreeparameter$ by Definition~\ref{definition:itd-constructor-semantics}, so $\mxdp(\abstractdecomposition,\awitness)=\truevalue$. Otherwise $\awitness'$ is a map and $\awitness=\awitness'|_{\domain(\awitness')\setminus\{\vertexone\}}$. By induction, $\maxDegree(\decompositiongraph{\sigmaabstractdecomposition})<\maxdegreeparameter$ and $\awitness'$ records degrees of active vertices. Since forgetting does not change the graph (Definition~\ref{definition:itd-constructor-semantics}), $\maxDegree(\decompositiongraph{\abstractdecomposition})<\maxdegreeparameter$ and $\awitness$ records degrees of the remaining active vertices, hence $\mxdp(\abstractdecomposition,\awitness)=\truevalue$.

$(\Leftarrow)$ Suppose $\mxdp(\abstractdecomposition,\awitness)=\truevalue$.
If $\awitness=\maxdegreefound$, then $\maxDegree(\decompositiongraph{\abstractdecomposition})\ge \maxdegreeparameter$.
Since forgetting does not change the graph, we also have $\maxDegree(\decompositiongraph{\sigmaabstractdecomposition})\ge \maxdegreeparameter$.
By the induction hypothesis,
\[
\maxdegreefound\in \dynamizationfunction{\dpcore}{\treewidthvalue}(\sigmaabstractdecomposition).
\]
Since $\forgetvertextype{\vertexone}(\maxdegreefound)=\{\maxdegreefound\}$, we get
\[
\maxdegreefound\in \dynamizationfunction{\dpcore}{\treewidthvalue}(\abstractdecomposition).
\]

Otherwise $\awitness$ is a map with $\domain(\awitness)=\topbag{\abstractdecomposition}$ and $\maxDegree(\decompositiongraph{\abstractdecomposition})<\maxdegreeparameter$.
Define a map $\awitness'$ on $\topbag{\sigmaabstractdecomposition}$ by
\[
\awitness'(u)\defeq d_{\decompositiongraph{\sigmaabstractdecomposition}}(\topmap{\sigmaabstractdecomposition}(u)).
\]
Since $\topbag{\sigmaabstractdecomposition}=\topbag{\abstractdecomposition}\cup\{\vertexone\}$ and $\maxDegree(\decompositiongraph{\sigmaabstractdecomposition})=\maxDegree(\decompositiongraph{\abstractdecomposition})<\maxdegreeparameter$, this map is well-defined with domain $\topbag{\sigmaabstractdecomposition}$ and satisfies $\mxdp(\sigmaabstractdecomposition,\awitness')=\truevalue$.
By the induction hypothesis, $\awitness'\in \dynamizationfunction{\dpcore}{\treewidthvalue}(\sigmaabstractdecomposition)$, and since $\forgetvertextype{\vertexone}(\awitness')=\{\awitness\}$, we get $\awitness\in \dynamizationfunction{\dpcore}{\treewidthvalue}(\abstractdecomposition)$.

\medskip
\noindent\textbf{Case 3:} $\abstractdecomposition=\introedgetype{\vertexone}{\vertextwo}(\sigmaabstractdecomposition)$.

$(\Rightarrow)$ Suppose $\awitness\in \dynamizationfunction{\dpcore}{\treewidthvalue}(\abstractdecomposition)$. Then there exists $\awitness'\in \dynamizationfunction{\dpcore}{\treewidthvalue}(\sigmaabstractdecomposition)$ such that $\awitness\in \introedgegeneric{\vertexone}{\vertextwo}(\awitness')$. If $\awitness'=\maxdegreefound$, then $\awitness=\maxdegreefound$ and $\mxdp(\abstractdecomposition,\awitness)=\truevalue$ as in the previous cases. Otherwise $\awitness'$ is a map. If $\awitness=\maxdegreefound$, then by the definition of $\introedgegeneric{\vertexone}{\vertextwo}$ we have $\awitness'(\vertexone)=\maxdegreeparameter-1$ or $\awitness'(\vertextwo)=\maxdegreeparameter-1$. By the induction hypothesis, $\awitness'$ records degrees in $\decompositiongraph{\sigmaabstractdecomposition}$, so after adding the new edge (Definition~\ref{definition:itd-constructor-semantics}) one of the endpoints has degree at least $\maxdegreeparameter$, and hence $\mxdp(\abstractdecomposition,\maxdegreefound)=\truevalue$. If $\awitness\neq \maxdegreefound$, then by definition $\awitness$ is obtained from $\awitness'$ by incrementing the values of $\vertexone$ and $\vertextwo$. By the induction hypothesis and Definition~\ref{definition:itd-constructor-semantics}, $\awitness$ records degrees in $\decompositiongraph{\abstractdecomposition}$, and since neither endpoint reached $\maxdegreeparameter$, we still have $\maxDegree(\decompositiongraph{\abstractdecomposition})<\maxdegreeparameter$. Thus $\mxdp(\abstractdecomposition,\awitness)=\truevalue$.

$(\Leftarrow)$ Suppose $\mxdp(\abstractdecomposition,\awitness)=\truevalue$. If $\awitness=\maxdegreefound$, then $\maxDegree(\decompositiongraph{\abstractdecomposition})\ge \maxdegreeparameter$. Let $\anedge$ be the edge introduced at the root. Removing $\anedge$ yields $\decompositiongraph{\sigmaabstractdecomposition}$ (Definition~\ref{definition:itd-constructor-semantics}), and either $\maxDegree(\decompositiongraph{\sigmaabstractdecomposition})\ge \maxdegreeparameter$ or one of the endpoints of $\anedge$ has degree $\maxdegreeparameter-1$ in $\decompositiongraph{\sigmaabstractdecomposition}$. In the first case, $\maxdegreefound\in \dynamizationfunction{\dpcore}{\treewidthvalue}(\sigmaabstractdecomposition)$ by induction and $\maxdegreefound\in \dynamizationfunction{\dpcore}{\treewidthvalue}(\abstractdecomposition)$. In the second case, let $\awitness'$ be the unique map on $\topbag{\sigmaabstractdecomposition}$ defined by $\awitness'(u)=d_{\decompositiongraph{\sigmaabstractdecomposition}}(\topmap{\sigmaabstractdecomposition}(u))$. Then $\mxdp(\sigmaabstractdecomposition,\awitness')=\truevalue$, so by induction $\awitness'\in \dynamizationfunction{\dpcore}{\treewidthvalue}(\sigmaabstractdecomposition)$, and the intro-edge rule yields $\maxdegreefound\in \introedgegeneric{\vertexone}{\vertextwo}(\awitness')$, hence $\awitness\in \dynamizationfunction{\dpcore}{\treewidthvalue}(\abstractdecomposition)$. If $\awitness\neq \maxdegreefound$, then $\maxDegree(\decompositiongraph{\abstractdecomposition})<\maxdegreeparameter$ and $\awitness$ records degrees in $\decompositiongraph{\abstractdecomposition}$. Let $\awitness'$ be the map obtained from $\awitness$ by decrementing the values at $\vertexone$ and $\vertextwo$ by $1$ (and leaving all other values unchanged). Then $\awitness'$ records degrees in $\decompositiongraph{\sigmaabstractdecomposition}$, so $\mxdp(\sigmaabstractdecomposition,\awitness')=\truevalue$ and by induction $\awitness'\in \dynamizationfunction{\dpcore}{\treewidthvalue}(\sigmaabstractdecomposition)$. By definition of $\introedgegeneric{\vertexone}{\vertextwo}$, we have $\awitness\in \introedgegeneric{\vertexone}{\vertextwo}(\awitness')$, hence $\awitness\in \dynamizationfunction{\dpcore}{\treewidthvalue}(\abstractdecomposition)$.

\medskip
\noindent\textbf{Case 4:} $\abstractdecomposition=\jointype(\sigmaabstractdecomposition_1,\sigmaabstractdecomposition_2)$.

$(\Rightarrow)$ Suppose $\awitness\in \dynamizationfunction{\dpcore}{\treewidthvalue}(\abstractdecomposition)$. Then there exist $\awitness_1\in \dynamizationfunction{\dpcore}{\treewidthvalue}(\sigmaabstractdecomposition_1)$ and $\awitness_2\in \dynamizationfunction{\dpcore}{\treewidthvalue}(\sigmaabstractdecomposition_2)$ such that $\awitness\in \jointype(\awitness_1,\awitness_2)$. If $\awitness=\maxdegreefound$, then either $\awitness_1=\maxdegreefound$ or $\awitness_2=\maxdegreefound$ or some bag label reaches sum $\ge \maxdegreeparameter$ in the join rule. In each subcase, by the induction hypothesis and Definition~\ref{definition:itd-constructor-semantics}, the joined graph $\decompositiongraph{\abstractdecomposition}$ has maximum degree at least $\maxdegreeparameter$, so $\mxdp(\abstractdecomposition,\awitness)=\truevalue$. Otherwise $\awitness$ is a map $\awitness''$ obtained by summing the two maps pointwise, and by induction each $\awitness_i$ records degrees in $\decompositiongraph{\sigmaabstractdecomposition_i}$ with maximum degree $<\maxdegreeparameter$. By Definition~\ref{definition:itd-constructor-semantics}, the two child edge sets are preserved, active bag vertices are identified, and edge identifiers are not merged; hence parallel edges contributed by both children are both counted, active degrees add pointwise, and all non-bag vertices keep their child degrees. Hence $\maxDegree(\decompositiongraph{\abstractdecomposition})<\maxdegreeparameter$ and $\awitness''$ records degrees of active vertices in $\decompositiongraph{\abstractdecomposition}$, so $\mxdp(\abstractdecomposition,\awitness)=\truevalue$.

$(\Leftarrow)$ Suppose $\mxdp(\abstractdecomposition,\awitness)=\truevalue$. If $\awitness\neq \maxdegreefound$, then $\domain(\awitness)=\topbag{\abstractdecomposition}$, $\maxDegree(\decompositiongraph{\abstractdecomposition})<\maxdegreeparameter$, and $\awitness$ records degrees of active vertices. Let $\awitness_i$ be the unique map on $\topbag{\sigmaabstractdecomposition_i}$ given by degrees in $\decompositiongraph{\sigmaabstractdecomposition_i}$. By Definition~\ref{definition:itd-constructor-semantics}, the three active bags are equal and these degrees are bounded by $\maxdegreeparameter-1$, so $\mxdp(\sigmaabstractdecomposition_i,\awitness_i)=\truevalue$ and by induction $\awitness_i\in \dynamizationfunction{\dpcore}{\treewidthvalue}(\sigmaabstractdecomposition_i)$ for $i\in\{1,2\}$. The join rule then yields $\awitness\in \jointype(\awitness_1,\awitness_2)$, hence $\awitness\in \dynamizationfunction{\dpcore}{\treewidthvalue}(\abstractdecomposition)$.

If $\awitness=\maxdegreefound$, then $\maxDegree(\decompositiongraph{\abstractdecomposition})\ge \maxdegreeparameter$. If this is witnessed by a vertex outside the bag, then it belongs to exactly one child graph (Definition~\ref{definition:itd-constructor-semantics}(e)), so the corresponding child already has maximum degree at least $\maxdegreeparameter$, and by induction $\maxdegreefound$ belongs to the corresponding dynamization witness set; the join rule then yields $\maxdegreefound$. Otherwise the maximum degree is witnessed by an identified bag vertex. If either child graph already has maximum degree at least $\maxdegreeparameter$, the same induction argument gives $\maxdegreefound$ in that child witness set and the join rule yields $\maxdegreefound$. It remains to consider the case where both child graphs have maximum degree less than $\maxdegreeparameter$. Let $\awitness_i$ be the degree map of the $i$th child on the common bag. By induction, these maps are in the corresponding witness sets. The degree of the witnessing bag vertex in the join graph is the sum of its two child degrees, so the join rule returns $\maxdegreefound$ when this sum reaches $\maxdegreeparameter$. Therefore $\maxdegreefound\in \dynamizationfunction{\dpcore}{\treewidthvalue}(\abstractdecomposition)$.
\end{proof}

\begin{corollary}
\label{corollary:CorrectnessMaxDegree}
Let $\abstractdecomposition$ be a $\treewidthvalue$-instructive tree decomposition. Then $\decompositiongraph{\abstractdecomposition}$
has $\maxDegree\geq \maxdegreeparameter$ if and only if
\[
\abstractdecomposition \in
\accepteddecompositions{\maxDegreeAtLeastCore{\maxdegreeparameter}[\treewidthvalue]}.
\]
\end{corollary}
\begin{proof}
Suppose that $\abstractdecomposition$ is accepted. Then there is a final witness in
$\dynamizationfunction{\maxDegreeAtLeastCore{\maxdegreeparameter}}{\treewidthvalue}(\abstractdecomposition)$. Therefore, $\maxdegreefound$ belongs to
$\dynamizationfunction{\maxDegreeAtLeastCore{\maxdegreeparameter}}{\treewidthvalue}(\abstractdecomposition)$,
since by definition, $\maxdegreefound$ is the only final witness of the DP-core. By Proposition \ref{proposition:PredicateMaxDegree}, $\maxDegree(\decompositiongraph{\abstractdecomposition})\geq \maxdegreeparameter$.

Now, suppose that $\maxDegree(\decompositiongraph{\abstractdecomposition})\geq \maxdegreeparameter$. Then $\mxdp(\abstractdecomposition,\maxdegreefound)=\truevalue$ by Definition~\ref{definition:MaxDegreePredicate}. By Proposition~\ref{proposition:PredicateMaxDegree}, $\maxdegreefound\in
\dynamizationfunction{\maxDegreeAtLeastCore{\maxdegreeparameter}}{\treewidthvalue}(\abstractdecomposition)$. Since $\maxdegreefound$ is a final witness,
\[
\abstractdecomposition   \in
\accepteddecompositions{\maxDegreeAtLeastCore{\maxdegreeparameter}[\treewidthvalue]}.
\]
\end{proof}

\begin{observation}
\label{observation:ComplexityMeasuresMaxDegree}
$\maxDegreeAtLeastCore{\maxdegreeparameter}[\treewidthvalue]$ has a bit-length of
$O(\treewidthvalue\cdot \log(\maxdegreeparameter+1))$ and timing
$O(\treewidthvalue\cdot \log(\maxdegreeparameter+1))$.
\end{observation}
\begin{proof}
A $\maxDegreeAtLeastCore{\maxdegreeparameter}[\treewidthvalue]$ local witness is either a map or a Boolean flag. We can represent a map domain and its values in $O(\treewidthvalue\cdot \log(\maxdegreeparameter+1))$ bits and a Boolean in $O(1)$ bits; therefore, $\maxDegreeAtLeastCore{\maxdegreeparameter}[\treewidthvalue]$ has a bit length of $O(\treewidthvalue\cdot\log(\maxdegreeparameter+1))$. Each transition scans or updates at most one map over the active labels and returns at most one local witness, so the timing is $O(\treewidthvalue\cdot \log(\maxdegreeparameter+1))$.
\end{proof}

\begin{lemma}[Witness action for $\maxDegreeAtLeastCore{\maxdegreeparameter}$]
\label{lemma:witness-action-maxdegree}
Fix $\treewidthvalue\in \N$.
Define $\action_{\dpcore}^{\treewidthvalue}$ on local witnesses by
\[
\action_{\dpcore}^{\treewidthvalue}(\relabelingfunction,\maxdegreefound)\defeq \maxdegreefound
\]
and for a map witness $\awitness:\abag\to\{0,\dots,\maxdegreeparameter-1\}$ with $\abag\subseteq \domain(\relabelingfunction)$,
let $\action_{\dpcore}^{\treewidthvalue}(\relabelingfunction,\awitness)$ be the map on $\relabelingfunction(\abag)$ defined by
\[
\action_{\dpcore}^{\treewidthvalue}(\relabelingfunction,\awitness)\bigl(\relabelingfunction(u)\bigr)\defeq \awitness(u)
\quad\text{for each }u\in \abag.
\]
Then $\action_{\dpcore}^{\treewidthvalue}$ is a witness action for $\dpcore[\treewidthvalue]$ in the sense of Definition~\ref{definition:witnessaction}.
\end{lemma}
\begin{proof}
We verify the conditions of Definition~\ref{definition:witnessaction}.

\smallskip\noindent
\emph{Finality invariance.}
In this core, $\maxdegreefound$ is the only final witness and $\action_{\dpcore}^{\treewidthvalue}(\relabelingfunction,\maxdegreefound)=\maxdegreefound$ by definition.

\smallskip\noindent
\emph{Label support transport.}
The label support of $\maxdegreefound$ is empty and hence preserved. For a map witness $\awitness:\abag\to\{0,\dots,\maxdegreeparameter-1\}$ we have $\mathrm{Lbl}_{\treewidthvalue}(\awitness)=\abag$ and $\mathrm{Lbl}_{\treewidthvalue}(\action_{\dpcore}^{\treewidthvalue}(\relabelingfunction,\awitness))=\relabelingfunction(\abag)=\relabelingfunction(\mathrm{Lbl}_{\treewidthvalue}(\awitness))$.

\smallskip\noindent
\emph{Inverse relabeling, composition, and extension invariance.}
These follow directly from the fact that $\action_{\dpcore}^{\treewidthvalue}(\relabelingfunction,\awitness)$ renames the domain labels by $\relabelingfunction$ while preserving values: for each $u\in \abag$,
\begin{multline*}
\action_{\dpcore}^{\treewidthvalue}\!\bigl(\relabelingfunction^{-1},\action_{\dpcore}^{\treewidthvalue}(\relabelingfunction,\awitness)\bigr)(u)\\
=\action_{\dpcore}^{\treewidthvalue}(\relabelingfunction,\awitness)(\relabelingfunction(u))
=\awitness(u),
\end{multline*}
and similarly for composition. If $\relabelingfunction'$ extends $\relabelingfunction$, then $\relabelingfunction'|_{\abag}=\relabelingfunction|_{\abag}$ and the definition implies $\action_{\dpcore}^{\treewidthvalue}(\relabelingfunction',\awitness)=\action_{\dpcore}^{\treewidthvalue}(\relabelingfunction,\awitness)$.

\smallskip\noindent
\emph{DP-core consistency.}
Let $\vertexone,\vertextwo\in \domain(\relabelingfunction)$ be distinct. If $\awitness=\maxdegreefound$, then every transition returns $\{\maxdegreefound\}$ and both sides of Conditions~\ref{condition:relabeldpcore-condition-four}--\ref{condition:relabeldpcore-condition-seven} are equal.
Otherwise let $\awitness:\abag\to\{0,\dots,\maxdegreeparameter-1\}$ with $\abag\subseteq \domain(\relabelingfunction)$.
Each transition updates $\awitness$ only at specified labels (insert $\vertexone\mapsto 0$, delete $\vertexone$, increment at $\vertexone,\vertextwo$, or sum pointwise in a join), and the definition of $\action_{\dpcore}^{\treewidthvalue}$ transports these updates to the relabeled labels $\relabelingfunction(\vertexone)$, $\relabelingfunction(\vertextwo)$. The empty-output guards are domain-membership tests and are preserved because $\domain(\action_{\dpcore}^{\treewidthvalue}(\relabelingfunction,\awitness))=\relabelingfunction(\domain(\awitness))$.
Concretely, for intro-vertex we have
\[
\action_{\dpcore}^{\treewidthvalue}\!\bigl(\relabelingfunction,\awitness\cup\{\vertexone\to 0\}\bigr)
=\action_{\dpcore}^{\treewidthvalue}(\relabelingfunction,\awitness)\cup\{\relabelingfunction(\vertexone)\to 0\},
\]
and the analogous equalities for restriction, incrementing values, and pointwise addition yield the forget, intro-edge, and join axioms.
Therefore $\action_{\dpcore}^{\treewidthvalue}$ satisfies Conditions~\ref{condition:relabeldpcore-condition-four}--\ref{condition:relabeldpcore-condition-seven}.
\end{proof}

}

				\subsection{Clique Number}

	{
			\renewcommand{\dpcore}{\simpleCliqueNumberAtLeastCore{\cliqueparameter}}
			Let $\agraph$ be a graph. We say $\agraph$ has a clique of size $\cliqueparameter$ if $\agraph$ has a clique of size $\cliqueparameter$ as a subgraph. Let $\simpleCliqueNumberAtLeastProperty{\cliqueparameter}$ be the graph property consisting of all graphs that have a clique of size $\cliqueparameter$. Next, we give the specification of a DP-core $\dpcore[\treewidthvalue]$ for the property $\simpleCliqueNumberAtLeastProperty{\cliqueparameter} \cap \neg \hasMultiEdgeProperty \cap \allgraphstreewidth{\treewidthvalue}$.
		Throughout this subsection we assume $\cliqueparameter\ge 2$.
		\\

	\noindent{\textbf{Local Witness.}} We begin by defining the structure of a local witness for $\dpcore[\treewidthvalue]$. During the processing of a $\treewidthvalue$-instructive tree decomposition, we do not know at the root which set of vertices forms a clique of size $\cliqueparameter$. Therefore, we need to guess possible sets of vertices that have been introduced during the process and could form a clique. As we process edges, we update these guesses. We define a local witness to be a “potential clique,” or rather, a “partial clique” since it is not yet a complete clique. During the process of the $\treewidthvalue$-instructive tree decomposition, we update this local witness.

		To implement this idea, we consider a local witness to be either the Boolean flag $\cliquefound$ or a pair consisting of a map and an integer, $(\cliquemap, \cliquesize)$. If the local witness $\awitness = \cliquefound$, it indicates that we have found a clique of size $\cliqueparameter$. Otherwise, if $\awitness = (\cliquemap,\cliquesize)$, the map $\cliquemap$ is used to track the current degrees of each active vertex included in the partial clique, reflecting its contribution to the potential clique. The integer $\cliquesize$ represents the count of vertices currently included in the partial clique, including selected vertices that may already have been forgotten, ensuring it does not exceed the target size $\cliqueparameter$.

	\begin{definition}\label{definition:clique-witness}
	A $\dpcore[\treewidthvalue]$ local witness $\awitness$ is one of the following:
	\begin{itemize}
		\item The Boolean flag $\cliquefound$, or
				\item A pair $(\cliquemap,\cliquesize)$ where $\cliquemap:\bagset' \to \{0,1,\dots,\cliqueparameter-1\}$ is a map for some $\bagset'\subseteq [\treewidthvalue+1]$, and $\cliquesize$ is an integer in $\{0,1,\dots,\cliqueparameter\}$ with $\domain(\cliquemap)=\bagset'$ and $|\domain(\cliquemap)|\le \cliquesize$.
	\end{itemize}
	\end{definition}

			In Definition \ref{definition:clique-witness}, $\domain(\cliquemap)$ is the set of active labels that are currently selected into the potential clique, and $\cliquemap$ records their current degrees inside this selected set. The value $\cliquesize$ is the size of the selected set (including any selected vertices that may already have been forgotten), so $|\domain(\cliquemap)|\le \cliquesize$. The flag $\cliquefound$ is the normalized form of the implementation witness whose found flag is true. It is important to note that this DP-core is correct only for simple graphs, i.e., graphs with at most one edge between each pair of vertices. The following definitions would fail for multi-edge graphs.
	\\

	\noindent{\textbf{All Witnesses}.} Let $\dpcore[\treewidthvalue].\allwitnesses$ be the set of all $\dpcore[\treewidthvalue]$-local witnesses:
	\begin{multline*}
	\dpcore[\treewidthvalue].\allwitnesses
	= \{\awitness : \awitness \text{ is a } \dpcore[\treewidthvalue]\\
	 \text{ local witness}\}.
	\end{multline*}

	\noindent{\textbf{Initial Witnesses}.} The initial state of the graph is empty, meaning there are no vertices or edges. Consequently, the only possible local witness at this stage is the empty map $\cliquemap_0$ with the clique size set to $0$, indicating that no clique has been found yet. This ensures the DP-core starts with a minimal and consistent initial configuration.

	\[
	\dpcore[\treewidthvalue].\initialsetgeneric =\{(\cliquemap_0,0)\}.
	\]

	\noindent{\textbf{Introduce Vertex}.}
	Let $\awitness$ be a $\dpcore[\treewidthvalue]$-local witness, and let $\vertexone$ represent a label in $[\treewidthvalue + 1]$. When a new vertex is introduced and labeled with $\vertexone$, we evaluate the situation based on the state of the local witness.

	If a clique of size \(\cliqueparameter\) has already been found (i.e., \(\awitness = \cliquefound\)), there is no need to modify the local witness. The function simply returns the current local witness unchanged, as the condition of finding the required clique has already been satisfied.

		Otherwise, let $\awitness = (\cliquemap,\cliquesize)$. If $\vertexone$ is already in $\domain(\cliquemap)$, the witness is incompatible with introducing a fresh vertex at label $\vertexone$, and the transition returns no witness. If the number of vertices currently included in the potential clique equals \(\cliqueparameter\) (i.e., \(\cliquesize = \cliqueparameter\)), no further vertices can be added to the clique. In this case, the local witness is also left unchanged, reflecting that no additional possibilities need to be explored.

		If neither of the above conditions applies, the function always keeps the option of not including the new vertex \(\vertexone\) in the potential clique. The implementation adds \(\vertexone\) as a selected clique vertex only when no selected clique vertex has already been forgotten, that is, when \(\cliquesize=|\domain(\cliquemap)|\). In that case, \(\vertexone\)'s degree is initialized to 0 in the map \(\cliquemap\), and the size of the potential clique is incremented by 1.

	This function is crucial for systematically exploring all configurations of vertices that might form a clique, while respecting the constraint of not exceeding the target clique size \(\cliqueparameter\).

		The transition $\dpcore[\treewidthvalue].\introvertextype{\vertexone}(\awitness)$ is defined by the following cases.
		\begin{enumerate}
		\item If $\awitness=\cliquefound$, then return $\{\awitness\}$.
		\item If $\awitness=(\cliquemap,\cliquesize)$ and $\vertexone\in\domain(\cliquemap)$, then return $\emptyset$.
		\item If $\awitness=(\cliquemap,\cliquesize)$, $\vertexone\notin\domain(\cliquemap)$, and $\cliquesize=\cliqueparameter$, then return $\{\awitness\}$.
		\item If $\awitness=(\cliquemap,\cliquesize)$, $\vertexone\notin\domain(\cliquemap)$, $\cliquesize<\cliqueparameter$, and $\cliquesize=|\domain(\cliquemap)|$, then return
		\[
		\{(\cliquemap,\cliquesize),(\cliquemap\cup\{\vertexone\to 0\},\cliquesize+1)\}.
		\]
			\item If $\awitness=(\cliquemap,\cliquesize)$, $\vertexone\notin\domain(\cliquemap)$, $\cliquesize<\cliqueparameter$, and $\cliquesize>|\domain(\cliquemap)|$, then return $\{(\cliquemap,\cliquesize)\}$.
			\end{enumerate}

			\noindent{\textbf{Forget Vertex}.} Let \(\awitness\) be a \(\dpcore[\treewidthvalue]\)-local witness, and let
		\(\vertexone\) represent a label in $[\treewidthvalue + 1]$. When the vertex labeled \(\vertexone\) is forgotten, the label \(\vertexone\) is freed, making it available for reuse when introducing a future vertex. The behavior of this function depends on the current state of \(\awitness\).

		If \(\awitness = \cliquefound\), this means a clique of size $\cliqueparameter$ has already been found, and as a consequence, the same local witness would be returned.

		Let \(\awitness = (\cliquemap,\cliquesize)\). If $\vertexone\notin \domain(\cliquemap)$, then $\vertexone$ is not currently selected as part of the partial clique, and forgetting it does not affect the witness. Otherwise, if $\vertexone\in \domain(\cliquemap)$ and $\cliquesize < \cliqueparameter$, then forgetting $\vertexone$ would make it impossible to complete the clique, so the witness becomes invalid and the function returns the empty set. Finally, if $\vertexone\in \domain(\cliquemap)$ and $\cliquesize = \cliqueparameter$, then $\vertexone$ can be forgotten only if it is adjacent (within the partial clique) to all other selected vertices, i.e., if $\cliquemap(\vertexone)=\cliqueparameter-1$; otherwise the witness is invalid.

				The transition $\dpcore[\treewidthvalue].\forgetvertextype{\vertexone}(\awitness)$ is defined by the following cases.
				\begin{enumerate}
				\item If $\awitness=\cliquefound$, then return $\{\awitness\}$.
				\item If $\awitness=(\cliquemap,\cliquesize)$ and $\vertexone\notin\domain(\cliquemap)$, then return $\{(\cliquemap,\cliquesize)\}$.
				\item If $\awitness=(\cliquemap,\cliquesize)$, $\vertexone\in\domain(\cliquemap)$, and $\cliquesize<\cliqueparameter$, then return $\emptyset$.
				\item If $\awitness=(\cliquemap,\cliqueparameter)$, $\vertexone\in\domain(\cliquemap)$, and $\cliquemap(\vertexone)\neq\cliqueparameter-1$, then return $\emptyset$.
				\item If $\awitness=(\cliquemap,\cliqueparameter)$, $\vertexone\in\domain(\cliquemap)$, and $\cliquemap(\vertexone)=\cliqueparameter-1$, then return
				\[
				\{(\cliquemap|_{\domain(\cliquemap)\setminus \{\vertexone\}},\cliqueparameter)\}.
				\]
				\end{enumerate}

	\noindent{\textbf{Introduce Edge}.} Let $\awitness$ be a local witness, and suppose a new edge is introduced between two distinct labeled vertices $\vertexone$ and $\vertextwo$. This operation updates the local witness \(\awitness\) to reflect the addition of the new edge, with the behavior depending on the current state of \(\awitness\).

	If \(\awitness = \cliquefound\), then the local witness is not modified since a clique has already been found.

	Let \(\awitness = (\cliquemap,\cliquesize)\). The operation considers two cases:
	1. If either \(\vertexone\) or \(\vertextwo\) is not in the domain of \(\cliquemap\), the function leaves the witness unchanged. This indicates that at least one of the vertices is not currently part of the potential clique, so adding an edge between them does not affect the witness.
	2. If both \(\vertexone\) and \(\vertextwo\) are in the domain of \(\cliquemap\), the function updates the local witness to reflect the addition of the edge. The degrees of \(\vertexone\) and \(\vertextwo\) are incremented in \(\cliquemap\), as the edge increases their connectivity within the potential clique. After updating the degrees, the function checks whether the partial clique is now complete. If all vertices in the domain of \(\cliquemap\) have degrees of \(\cliqueparameter - 1\) and the clique size (\(\cliquesize\)) equals \(\cliqueparameter\), the local witness changes to the Boolean flag \(\cliquefound\), indicating that a valid clique of the required size has been formed. Otherwise, the updated local witness is returned.

	Formally, the operation $\dpcore[\treewidthvalue].\introedgegeneric{\vertexone}{\vertextwo}(\awitness)$ produces one of the following outcomes:

		\begin{enumerate}
		\item If $\awitness=\cliquefound$, return $\{\awitness\}$.
		\item If either $\vertexone$ or $\vertextwo$ is not in $\domain(\cliquemap)$, return $\{\awitness\}$.
		\item If $\vertexone,\vertextwo\in \domain(\cliquemap)$ and either $\cliquemap(\vertexone)=\cliqueparameter-1$ or $\cliquemap(\vertextwo)=\cliqueparameter-1$, return $\emptyset$.
		\item If $\vertexone,\vertextwo\in \domain(\cliquemap)$, create a new map
	\[
	\cliquemap' = \cliquemap|_{\domain(\cliquemap)\setminus \{\vertexone,\vertextwo\}}\cup\{\vertexone\to\cliquemap(\vertexone)+1,\vertextwo\to\cliquemap(\vertextwo)+1\}
	\]
	\begin{itemize}
		\item If $\cliquesize=\cliqueparameter$ and for all $l \in \domain(\cliquemap')$, $\cliquemap'(l) = \cliqueparameter-1$ return
		\[ \{\cliquefound\}\]
		\item otherwise, return
		\[\{ (\cliquemap',\cliquesize)\}\]
		\end{itemize}
		\end{enumerate}

	\noindent{\textbf{Join}.} Let $\awitness$ and $\awitness'$ be local witnesses associated with the left and right children of a given node, respectively. The goal of the $\joingenericcore$ operation is to merge the graphs corresponding to these children along the vertices labeled with active labels. This operation combines the information from the two witnesses to evaluate whether a clique of size \(\cliqueparameter\) can be formed when their corresponding subgraphs are joined. The behavior of the function depends on several conditions.

	If either witness already indicates that a clique of size \(\cliqueparameter\) has been found (i.e., \(\awitness = \cliquefound\) or \(\awitness' = \cliquefound\)), the function immediately returns \(\{\cliquefound\}\) without further processing. This is because the desired clique condition has already been met.

	Let $\awitness = (\cliquemap,\cliquesize)$ and $\awitness' = (\cliquemap',\cliquesize')$.
	Consider the situation where $|\domain(\cliquemap)| < \cliquesize$ and $|\domain(\cliquemap')| < \cliquesize'$. In this case, the function returns the empty set. This situation occurs when some vertices have been forgotten in both local witnesses, making it impossible to form a complete clique of size \(\cliqueparameter\): in a join, there are no cross-edges between vertices outside the interface that originate from different children (Lemma~\ref{lemma:join-no-cross-edges}), so a clique cannot contain non-interface vertices from both sides. Note that we need the size of the domain of the map to be equal to the number of included vertices in at least one local witness.

	If the domains of the two maps \(\cliquemap\) and \(\cliquemap'\) are not identical, the function also returns the empty set. A lack of a shared domain implies that the two partial cliques cannot be merged meaningfully, as there is no common structure to form a combined clique.

		Otherwise, the function proceeds to merge the two witnesses into a new witness. The combined map, \(\cliquemap''\), is created by summing the degrees of each vertex in the shared domain of \(\cliquemap\) and \(\cliquemap'\). If some sum exceeds \(\cliqueparameter-1\), the pair is rejected because the result would not be a local witness. Otherwise the size of the combined clique, \(\cliquesize''\), is calculated as \(\cliquesize + \cliquesize' - |\domain(\cliquemap)|\), which accounts for the overlap of vertices in the shared domain.

		Finally, the operation checks whether the merged witness is a complete clique: if $\cliquesize''=\cliqueparameter$ and each label in $\domain(\cliquemap'')$ has value $\cliqueparameter-1$, the operation returns $\cliquefound$, indicating that a clique has been found. Otherwise, the map together with the new size is returned as a local witness.

	Formally, the operation $\dpcore[\treewidthvalue].\joingenericcore(\awitness, \awitness')$ produces one of the following outcomes:

	\begin{enumerate}
	\item If either $\awitness = \cliquefound$ or $\awitness' = \cliquefound$, return $\{\cliquefound\}$.
	\item If $|\domain(\cliquemap)| < \cliquesize$ and $|\domain(\cliquemap')| < \cliquesize'$, return $\emptyset$.
	\item If $\domain(\cliquemap)\neq \domain(\cliquemap')$, return $\emptyset$.
		\item If there exists $\vertexone\in\domain(\cliquemap)$ such that
		$\cliquemap(\vertexone)+\cliquemap'(\vertexone)>\cliqueparameter-1$, return $\emptyset$.
		\item Otherwise, create $\cliquemap''$ where for all $\vertexone \in \domain(\cliquemap)$\footnote{Note that $\domain(\cliquemap) = \domain(\cliquemap')$}, $\cliquemap''(\vertexone) = \cliquemap(\vertexone) + \cliquemap'(\vertexone)$, and let $\cliquesize'' = \cliquesize + \cliquesize' - |\domain(\cliquemap)|$.

		\begin{itemize}
		\item If $\cliquesize''=\cliqueparameter$ and for all $\vertexone \in \domain(\cliquemap'')$, $\cliquemap''(\vertexone) = \cliqueparameter -1 $, return
		\[\{\cliquefound\}\]
		\item Otherwise return
		\[\{(\cliquemap'',\cliquesize'')\}\]
	\end{itemize}
	\end{enumerate}

	\noindent{\textbf{Final Witness Function.}} The final witness function checks if the local witness \(\awitness\) has successfully found a clique of size \(\cliqueparameter\). If \(\awitness = \cliquefound \), meaning a clique of the required size has been identified, the function returns 1. Otherwise, if \(\awitness \neq \cliquefound\), it returns 0, indicating that the desired clique has not been found.

	Formally:
	\[ \dpcore[\treewidthvalue].\finalwitnessgeneric(\awitness) = \begin{cases}
			1 & \quad \text{if } \awitness = \cliquefound  \\
			0 & \quad \text{otherwise }
		\end{cases}
	\]
			\noindent{\textbf{Cleaning Function}.} The implementation removes all other witnesses once a found witness is present:
			\[
			\dpcore[\treewidthvalue].\cleaningfunctioncore(\witnessset)=
			\begin{cases}
			\{\cliquefound\} & \quad \text{if } \cliquefound\in \witnessset,\\[4pt]
			\witnessset & \quad \text{otherwise.}
			\end{cases}
			\]

	\noindent{\textbf{Invariant Function}.} In our example, the invariant function is the trivial one, which sends a set of local witnesses $\witnessset$ to $1$ if and only if it contains a final local witness.

	\[
		\dpcore[\treewidthvalue].\invariantCore(\witnessset)=
		\begin{cases}
			1 & \quad \text{ if } \witnessset \text{ has a final witness,} \\
			0 & \quad \text{ otherwise.}
		\end{cases}
	\]

    Note that the DP-core $\dpcore$ works only for simple graphs because it assumes that a vertex's degree accurately represents the number of unique neighbors it has. In graphs with multiple edges (multigraphs), the degree would count repeated connections, making it impossible to reliably track clique membership based on degrees alone.

	The next step is to prove the correctness of the DP-core we just defined. For this, we define the following predicate relating $\treewidthvalue$-instructive tree decompositions with local witnesses.

				\begin{definition}
					\label{definition:PredicateSimpleCliqueNumberAtLeast}
					We let $\predicateSimpleCliqueNumberAtLeast{\cliqueparameter}{\treewidthvalue}$ be the predicate that is true on a pair $(\abstractdecomposition,\awitness) \in \allabstractdecompositionstreewidth{\treewidthvalue}\times \simpleCliqueNumberAtLeastCore{\cliqueparameter}[\treewidthvalue].\allwitnesses$ if and only if there exists a subset $\vertexsub$ of $\vertexset{\decompositiongraph{\abstractdecomposition}}$ such that it satisfies one of the following conditions.
					Here, for a graph $G$ and a vertex set $S$, we write $G[S]$ for the induced subgraph of $G$ on $S$.
					\begin{enumerate}
						\item If $\awitness = \cliquefound$, the set $\vertexsub$ forms a clique of size $\cliqueparameter$.
						\item If $\awitness = (\cliquemap,\cliquesize)$, then the following conditions are satisfied.
						\begin{enumerate}
							\item $|\vertexsub| = \cliquesize$,
							\item $\vertexsub\cap \topmap{\abstractdecomposition}(\topbag{\abstractdecomposition}) = \topmap{\abstractdecomposition}(\domain(\cliquemap))$, and
							\item for each $\vertexone\in \domain(\cliquemap)$, we have $\cliquemap(\vertexone) = \vdegree_{\decompositiongraph{\abstractdecomposition}[\vertexsub]}(\topmap{\abstractdecomposition}(\vertexone))$, and for each $\xvertex\in \vertexsub\setminus \topmap{\abstractdecomposition}(\domain(\cliquemap))$, $\vdegree_{\decompositiongraph{\abstractdecomposition}[\vertexsub]}(\xvertex) = \cliqueparameter - 1$.
							\item $\cliquesize < \cliqueparameter$ or there exists $\vertexone\in \domain(\cliquemap)$ such that $\cliquemap(\vertexone) < \cliqueparameter-1$.
						\end{enumerate}
					\end{enumerate}
				\end{definition}

	The first condition indicates that if $\awitness = \cliquefound$, then there should be a vertex set $\vertexsub$ that forms a clique of size $\cliqueparameter$. The second condition means that if a local witness $\awitness = (\cliquemap,\cliquesize)$ is not a complete clique, then there should be a vertex subset $\vertexsub$ of size $\cliquesize$ which satisfies the condition for a partial clique.

			\begin{proposition}
				\label{proposition:PredicateSimpleCliqueNumberAtLeast}
				For each $\abstractdecomposition\in \allabstractdecompositionstreewidth{\treewidthvalue}$ such that $\decompositiongraph{\abstractdecomposition}$ is a simple graph, the following hold.
				\begin{enumerate}
				\item $\cliquefound\in \dynamizationfunction{\simpleCliqueNumberAtLeastCore{\cliqueparameter}}{\treewidthvalue}(\abstractdecomposition)$ if and only if $\predicateSimpleCliqueNumberAtLeast{\cliqueparameter}{\treewidthvalue}(\abstractdecomposition,\cliquefound)=\truevalue$.
					\item If $\cliquefound\notin \dynamizationfunction{\simpleCliqueNumberAtLeastCore{\cliqueparameter}}{\treewidthvalue}(\abstractdecomposition)$, then for every map witness $\awitness=(\cliquemap,\cliquesize)$,
					\begin{multline*}
					\awitness\in \dynamizationfunction{\simpleCliqueNumberAtLeastCore{\cliqueparameter}}{\treewidthvalue}(\abstractdecomposition)
					\\
					\text{if and only if}\\
					\predicateSimpleCliqueNumberAtLeast{\cliqueparameter}{\treewidthvalue}(\abstractdecomposition,\awitness)=\truevalue.
					\end{multline*}
				\end{enumerate}
				In particular, if the dynamization contains $\cliquefound$, then the specified cleaning function makes the cleaned witness set equal to $\{\cliquefound\}$.
				\end{proposition}

				\begin{proof}
				First ignore the cleaning function and let $\widehat D(\abstractdecomposition)$ denote the witness set obtained by applying only the transition rules. We prove by structural induction on $\abstractdecomposition$ that, for every local witness $\awitness$,
				\[
				\awitness\in \widehat D(\abstractdecomposition)
				\quad\text{if and only if}\quad
				\predicateSimpleCliqueNumberAtLeast{\cliqueparameter}{\treewidthvalue}(\abstractdecomposition,\awitness)=\truevalue.
				\]
				The witness set $\vertexsub$ in Definition~\ref{definition:PredicateSimpleCliqueNumberAtLeast} is the selected set. The map domain records exactly the selected vertices that are still active, the integer $\cliquesize$ records the total number of selected vertices including already forgotten ones, and every selected forgotten vertex is required to have degree $\cliqueparameter-1$ inside the selected induced subgraph. This is the invariant maintained below.

				The leaf case gives the single witness $(\cliquemap_0,0)$, witnessed by the empty selected set. For an introduce-vertex node, legality gives $\vertexone\notin\topbag{\sigmaabstractdecomposition}$, so the fresh-label guard is not triggered on reachable witnesses. If the new isolated vertex is not selected, the selected set and the predecessor witness are unchanged. If it is selected, then no selected vertex can already have been forgotten: otherwise the forgotten vertex would need degree $\cliqueparameter-1$ in the selected set, but the new isolated vertex is not adjacent to it. Hence selection of the new vertex is possible exactly in the branch $\cliquesize=|\domain(\cliquemap)|$, and the transition adds $\vertexone$ to the domain with counter $0$ and increments $\cliquesize$.

				For a forget-vertex node, if $\vertexone$ is not selected, the same selected set witnesses the unchanged map. If $\vertexone$ is selected and $\cliquesize<\cliqueparameter$, then after forgetting it the selected set cannot still be extended to a clique of size $\cliqueparameter$ through active vertices, so the transition rejects. If $\cliquesize=\cliqueparameter$, then $\vertexone$ may become inactive only when it already has degree $\cliqueparameter-1$ in the selected induced subgraph; this is exactly the predicate requirement for selected non-active vertices. The transition therefore deletes $\vertexone$ from the domain precisely in that case.

				For an introduce-edge node, if at least one endpoint is not selected, the selected induced subgraph does not change and the witness is kept. If both endpoints are selected, simplicity of $\decompositiongraph{\abstractdecomposition}$ implies that the new edge contributes exactly one new adjacency between them. The transition rejects if either counter is already $\cliqueparameter-1$, since incrementing would leave the declared witness universe. Otherwise it increments the two counters. If after this update the selected set has size $\cliqueparameter$ and every active selected counter is $\cliqueparameter-1$, then all selected active and forgotten vertices have degree $\cliqueparameter-1$ in the selected induced subgraph, so the selected set is a clique of size $\cliqueparameter$ and the transition returns $\cliquefound$. In all other cases the updated map is exactly the predicate witness.

					For a join node, the two child top bags are identified. For the forward direction, let $\vertexsub$ be a selected set witnessing the parent predicate. Its restrictions to the two child graphs have the same selected active labels, giving the equal-domain condition in the transition. If both restrictions had selected vertices outside the interface, then $\vertexsub$ would contain non-interface selected vertices from both children; by Lemma~\ref{lemma:join-no-cross-edges} there is no edge between such vertices, so they cannot both have degree $\cliqueparameter-1$ inside $\decompositiongraph{\abstractdecomposition}[\vertexsub]$. Thus the transition correctly rejects exactly this incompatible case. Otherwise, $\vertexsub$ is the union of the two child selected sets, with the active selected vertices counted once, so its size is $\cliquesize+\cliquesize'-|\domain(\cliquemap)|$. Active degrees add pointwise across the children. Since the parent graph is simple, no active-active edge is contributed by both children, so this sum does not double-count an edge. The overflow guard rejects exactly the cases in which the summed map would not be a local witness. If the merged selected set has size $\cliqueparameter$ and all active counters reach $\cliqueparameter-1$, then all selected active vertices and all selected forgotten vertices have degree $\cliqueparameter-1$ in a selected set of size $\cliqueparameter$, so the transition returns $\cliquefound$; otherwise it returns the merged map, which satisfies the predicate.

					Conversely, suppose a join transition combines child witnesses satisfying the predicate. If either child witness is $\cliquefound$, the child clique embeds into the joined graph, so $\cliquefound$ satisfies the parent predicate. Otherwise, take selected sets witnessing the two child map predicates. The transition requires equal active selected domains and rejects the case where both children contain forgotten selected vertices. Hence their union is well-defined after identifying the common active vertices, has size $\cliquesize+\cliquesize'-|\domain(\cliquemap)|$, and has no pair of selected non-interface vertices coming from different children. For selected active vertices, the parent degree is the sum of the two child counters; for selected forgotten vertices, the required degree $\cliqueparameter-1$ already holds in the child where the vertex occurs, and all selected vertices outside that child are active interface vertices already represented in that child. Thus the merged witness satisfies the parent predicate. If the transition returns $\cliquefound$, the same degree argument shows that the selected set has size $\cliqueparameter$ and every selected vertex has degree $\cliqueparameter-1$, hence it is a clique of size $\cliqueparameter$.

				The actual dynamization applies the specified cleaning function after each transition. This cleaning function is the identity on witness sets not containing $\cliquefound$, and maps every witness set containing $\cliquefound$ to $\{\cliquefound\}$. Since $\cliquefound$ is preserved by all transitions, the cleaned dynamization contains $\cliquefound$ exactly in the cases characterized above by the predicate. If $\cliquefound$ is absent, cleaning is the identity throughout the computation, so the characterization of non-found witnesses is the same as for $\widehat D$.
				\end{proof}

			\begin{corollary}
				\label{corollary:CorrectnessSimpleCliqueNumberAtLeast}
				Let $\abstractdecomposition$ be a $\treewidthvalue$-instructive tree
				decomposition such that $\decompositiongraph{\abstractdecomposition}$ is simple. Then $\decompositiongraph{\abstractdecomposition}$
				has a clique of size $\cliqueparameter$ if and only if $\abstractdecomposition\in
					\accepteddecompositions{\simpleCliqueNumberAtLeastCore{\cliqueparameter}[\treewidthvalue]}$.
			\end{corollary}
			\begin{proof}
			Assume that $\decompositiongraph{\abstractdecomposition}$ is simple and let
			$S\defeq \dynamizationfunction{\simpleCliqueNumberAtLeastCore{\cliqueparameter}}{\treewidthvalue}(\abstractdecomposition)$.
			By definition, $\abstractdecomposition$ is accepted if and only if $S$ contains a final witness, i.e., $\cliquefound\in S$.
			By Proposition~\ref{proposition:PredicateSimpleCliqueNumberAtLeast}, this holds if and only if
			$\predicateSimpleCliqueNumberAtLeast{\cliqueparameter}{\treewidthvalue}(\abstractdecomposition,\cliquefound)=\truevalue$, which is equivalent to
			$\decompositiongraph{\abstractdecomposition}$ having a clique of size $\cliqueparameter$.
			\end{proof}

		\begin{observation}
			\label{observation:ComplexityMeasuresClique}
			$\simpleCliqueNumberAtLeastCore{\cliqueparameter}[\treewidthvalue]$ has bit-length
			$O(\treewidthvalue\cdot \log \cliqueparameter)$ and timing
			$O(\treewidthvalue\cdot \log \cliqueparameter)$.
		\end{observation}

		\begin{proof}
			A $\simpleCliqueNumberAtLeastCore{\cliqueparameter}[\treewidthvalue]$ local witness is either the Boolean flag $\cliquefound$ or a pair $(\cliquemap,\cliquesize)$, where $\cliquemap$ is a map whose domain is a subset of $[\treewidthvalue+1]$ and whose values are bounded by $\cliqueparameter$. Therefore, such a witness can be encoded using $O(\treewidthvalue\cdot \log \cliqueparameter)$ bits.
			Each transition of $\simpleCliqueNumberAtLeastCore{\cliqueparameter}[\treewidthvalue]$ performs a constant number of map updates and/or scans over the current domain (of size at most $\treewidthvalue+1$), so its running time is $O(\treewidthvalue\cdot \log \cliqueparameter)$.
		\end{proof}

			\begin{lemma}[Witness action for $\simpleCliqueNumberAtLeastCore{\cliqueparameter}$]
		\label{lemma:witness-action-clique}
		Fix $\treewidthvalue\in \N$.
		Define $\action_{\dpcore}^{\treewidthvalue}$ on local witnesses by
		\[
		\action_{\dpcore}^{\treewidthvalue}(\relabelingfunction,\cliquefound)\defeq \cliquefound
		\]
		and for a map witness $\awitness=(\cliquemap,\cliquesize)$ with $\mathrm{Lbl}_{\treewidthvalue}(\awitness)\subseteq \domain(\relabelingfunction)$, let
		\[
		\action_{\dpcore}^{\treewidthvalue}(\relabelingfunction,\awitness)
		\defeq
		(\cliquemap',\cliquesize),
		\]
			where
			\[
			\domain(\cliquemap')=\relabelingfunction(\domain(\cliquemap))
			\quad\text{and}\quad
			\cliquemap'(\relabelingfunction(u))=\cliquemap(u)
			\]
			for each $u\in \domain(\cliquemap)$.
		Then $\action_{\dpcore}^{\treewidthvalue}$ is a witness action for $\dpcore[\treewidthvalue]$ in the sense of Definition~\ref{definition:witnessaction}.
		\end{lemma}
		\begin{proof}
		We verify the conditions of Definition~\ref{definition:witnessaction}.

		\smallskip\noindent
		\emph{Finality invariance.}
		The only final witness is $\cliquefound$, and it is fixed by the action.

		\smallskip\noindent
		\emph{Label support transport.}
		The label support of $\cliquefound$ is empty. For a map witness $\awitness=(\cliquemap,\cliquesize)$, the only labels mentioned in $\awitness$ are those in $\domain(\cliquemap)$, and by definition $\domain(\cliquemap')=\relabelingfunction(\domain(\cliquemap))$, so
		\[
		\mathrm{Lbl}_{\treewidthvalue}(\action_{\dpcore}^{\treewidthvalue}(\relabelingfunction,\awitness))
		=\domain(\cliquemap')
		=\relabelingfunction(\domain(\cliquemap))
		=\relabelingfunction(\mathrm{Lbl}_{\treewidthvalue}(\awitness)).
		\]

		\smallskip\noindent
		\emph{Inverse relabeling, composition, and extension invariance.}
		For map witnesses, the action renames the domain labels by $\relabelingfunction$ and preserves all integer values as well as $\cliquesize$, so applying $\relabelingfunction^{-1}$ recovers the original map, and composition holds labelwise.
		If $\relabelingfunction'$ extends $\relabelingfunction$, then $\relabelingfunction'$ agrees with $\relabelingfunction$ on $\domain(\cliquemap)$ and hence induces the same renamed map.

		\smallskip\noindent
		\emph{DP-core consistency.}
		If $\awitness=\cliquefound$, then every transition returns $\{\cliquefound\}$ and commutation is immediate.
		Otherwise let $\awitness=(\cliquemap,\cliquesize)$.
		Each transition of $\dpcore[\treewidthvalue]$ is defined by predicates of the form $\vertexone\in \domain(\cliquemap)$, comparisons involving $\cliquesize$, and pointwise updates/restrictions/sums of $\cliquemap$ at named labels.
		Because $\cliquemap'(\relabelingfunction(u))=\cliquemap(u)$ for every $u\in \domain(\cliquemap)$, we have
		\begin{align*}
		\vertexone\in \domain(\cliquemap) &\iff \relabelingfunction(\vertexone)\in \domain(\cliquemap'),\\
		\cliquemap(\vertexone)=t &\iff \cliquemap'(\relabelingfunction(\vertexone))=t,
		\end{align*}
		so each branch condition in $\introvertextype{\cdot}$, $\forgetvertextype{\cdot}$, $\introedgegeneric{\cdot}{\cdot}$, and $\joingenericcore$ is preserved under relabeling.
			Moreover, the corresponding updates (adding/removing a label from the domain, incrementing/summing counters, and updating $\cliquesize$) are transported by renaming labels via $\relabelingfunction$.
			Hence $\action_{\dpcore}^{\treewidthvalue}$ commutes with all four transitions and satisfies Conditions~\ref{condition:relabeldpcore-condition-four}--\ref{condition:relabeldpcore-condition-seven}.
				The cleaning function is equivariant as well: $\cliquefound$ is fixed by the action, and
				\begin{multline*}
				\action_{\dpcore}^{\treewidthvalue}(\relabelingfunction,
				\dpcore[\treewidthvalue].\cleaningfunctioncore(\witnessset))\\
				=
				\dpcore[\treewidthvalue].\cleaningfunctioncore(
				\action_{\dpcore}^{\treewidthvalue}(\relabelingfunction,\witnessset))
				\end{multline*}
				for every witness set $\witnessset$ on which the action is defined.
			Indeed, both sides are $\{\cliquefound\}$ if $\cliquefound\in\witnessset$, and otherwise cleaning is the identity on both sides.
			\end{proof}

}

				\subsection{Multiple Edges}
{
	\renewcommand{\dpcore}{\hasMultiEdgeCore}

	Let \(\agraph=(\vertexsetname,\edgesetname,\incidencerelationname)\) be a graph. \emph{Multiple edges} refer to two or more edges that connect the same pair of vertices. A graph \(\agraph\) is called a \emph{multigraph} if it contains multiple edges. Let \(\hasMultiEdgeProperty\) denote the property of graphs that have multiple edges. We define a DP-core \(\hasMultiEdgeCore[\treewidthvalue]\) for the property \(\hasMultiEdgeProperty \cap \allgraphstreewidth{\treewidthvalue}\).

	\noindent\textbf{Local Witness.}
	We begin by describing the structure of local witnesses for \(\dpcore[\treewidthvalue]\). To verify whether a graph contains multiple edges, it is necessary to track its edges. Specifically, a local witness is represented either as an unordered set of sets or as the Boolean flag \(\multiedgefound\). The unordered set is used to store pairs of vertex labels for which an edge exists during the processing of a \(\treewidthvalue\)-instructive tree decomposition. The Boolean flag indicates that multiple edges have been found, so no further checks are necessary.

	\begin{definition}
	A \(\dpcore[\treewidthvalue]\) local witness \(\awitness\) is one of the following:

	\begin{itemize}
		\item The Boolean flag \(\multiedgefound\), or
		\item An unordered set containing sets of size \(2\) indicating existing edges between two distinct labels associated with active vertices.
	\end{itemize}
	\end{definition}

	\noindent\textbf{All Witnesses.}
	We define \(\dpcore[\treewidthvalue].\allwitnesses\) as the set of all \(\dpcore[\treewidthvalue]\)-local witnesses:
	\[
	\dpcore[\treewidthvalue].\allwitnesses = \{\awitness : \awitness \text{ is a } \dpcore[\treewidthvalue] \text{ local witness}\}.
	\]

	\noindent\textbf{Initial Witnesses.}
	When the instruction \(\initialsetgenericcore\) is invoked, it initializes an empty graph. For the empty graph, the only valid local witness is the empty set since no edge has been introduced.
	\[
	\dpcore[\treewidthvalue].\initialsetgeneric = \{\emptyset\}.
	\]

	\noindent\textbf{Introduce Vertex.}
	Let \(\awitness\) be a \(\dpcore[\treewidthvalue]\)-local witness, and let \(\vertexone\) represent a label in \([\treewidthvalue + 1]\). When a new vertex is introduced and labeled with \(\vertexone\), there are no changes to the edges of the graph. Consequently, the same local witness is returned without modification. Note that the local witness can be either the Boolean flag or a set.

	\[
	\dpcore[\treewidthvalue].\introvertexgeneric{\vertexone}(\awitness) = \{\awitness\}.
	\]

	\noindent\textbf{Forget Vertex.}
	Let \(\awitness\) be a local witness, and let \(\vertexone\) represent an active label. When the vertex labeled \(\vertexone\) is forgotten, the label \(\vertexone\) is released, making it available for future use when introducing another vertex. If multiple edges have been found, i.e., \(\awitness = \multiedgefound\), no modifications are required and the same local witness is returned. If multiple edges have not been found (\(\awitness\neq \multiedgefound\)), then the edges connected to the vertex labeled \(\vertexone\) are no longer relevant and are removed.

		\[
		\dpcore[\treewidthvalue].\forgetvertextype{\vertexone}(\awitness) =
		\begin{cases}
			\{\awitness\} & \quad \text{if } \awitness = \multiedgefound, \\
			\Bigl\{\bigl\{\{\vertextwo, \vertextwo'\} : \substack{\vertextwo,\vertextwo'\neq \vertexone\\ \{\vertextwo, \vertextwo'\}\in \awitness}\bigr\}\Bigr\} & \quad \text{otherwise.}
		\end{cases}
		\]

	\noindent\textbf{Introduce Edge.}
	Let \(\awitness\) be a local witness. When introducing an edge between two distinct vertices labeled \(\vertexone\) and \(\vertextwo\), there are three possible outcomes:
	\begin{enumerate}
	    \item If multiple edges have already been found (\(\awitness = \multiedgefound\)), the local witness remains unchanged.
	    \item If \(\awitness \neq \multiedgefound\) and the new edge creates a multiple edge (i.e., \(\{\vertexone, \vertextwo\} \in \awitness\)), then the returned local witness is the Boolean flag \(\multiedgefound\).
	    \item Otherwise, if \(\{\vertexone, \vertextwo\} \notin \awitness\), the edge pair \(\{\vertexone, \vertextwo\}\) is added to \(\awitness\).
	\end{enumerate}

	\[
	\dpcore[\treewidthvalue].\introedgegeneric{\vertexone}{\vertextwo}(\awitness) =
		\begin{cases}
			\{\multiedgefound\} & \quad \text{if } \awitness = \multiedgefound, \\
			\{\awitness \cup \{\{\vertexone, \vertextwo\}\}\} & \quad \text{if } \awitness\neq\multiedgefound \text{ and } \{\vertexone, \vertextwo\} \notin \awitness, \\
			\{\multiedgefound\} & \quad \text{if } \awitness\neq\multiedgefound \text{ and } \{\vertexone, \vertextwo\} \in \awitness.
		\end{cases}
		\]

	\noindent\textbf{Join.}
	Let \(\awitness \) and \(\awitness'\) be local witnesses associated with the left and right children of a given node. The purpose of the \(\joingenericcore\) operation is to combine the graphs corresponding to these children, focusing on the vertices labeled with active labels.

	If multiple edges have already been found in either child (\(\awitness = \multiedgefound\) or \(\awitness'=\multiedgefound\)), the join operation immediately returns \(\{\multiedgefound\}\). Otherwise, if the intersection \(\awitness \cap \awitness'\) is non-empty, this indicates the presence of multiple edges spanning both children, and the operation returns \(\{\multiedgefound\}\). If neither of these conditions is met, meaning no multiple edges exist, the join operation returns a combined local witness \(\awitness \cup \awitness'\).

	\[
	\dpcore[\treewidthvalue].\joingenericcore(\awitness, \awitness') =
		\begin{cases}
			\{\multiedgefound\} & \quad \text{if } \awitness = \multiedgefound \text{ or } \awitness' = \multiedgefound, \\
			\{\multiedgefound\} & \quad \text{if } \awitness,\awitness'\neq\multiedgefound \text{ and } \awitness \cap \awitness' \neq \emptyset, \\
			\{\awitness \cup \awitness'\} & \quad \text{if } \awitness,\awitness'\neq\multiedgefound \text{ and } \awitness \cap \awitness' = \emptyset.
		\end{cases}
		\]

	\noindent\textbf{Final Witness Function.}
	The final witness function is used to determine whether the condition of finding multiple edges has been satisfied.

	\[
	\dpcore[\treewidthvalue].\finalwitnessgenericcore(\awitness) = 1 \iff \awitness = \multiedgefound.
	\]

	\noindent\textbf{Cleaning Function.}
	In our example, the cleaning function acts as the identity.
	\[
	\dpcore[\treewidthvalue].\cleaningfunctioncore(\witnessset)=\witnessset \text{ for every }
		\witnessset \subseteq \dpcore[\treewidthvalue].\allwitnesses.
	\]

	\noindent\textbf{Invariant Function.}
	In our example, the invariant function is the trivial one, which sends a set of local witnesses \(\witnessset\) to \(1\) if and only if it contains a final witness.

	\[
	\dpcore[\treewidthvalue].\invariantCore(\witnessset)=
		\begin{cases}
			1 & \quad \text{ if } \witnessset \text{ contains a final witness}, \\
			0 & \quad \text{ otherwise}.
		\end{cases}
	\]

	The next step is to prove the correctness of the DP-core we just defined. For this, we define the following predicate relating \(\treewidthvalue\)-instructive tree decompositions with local witnesses.

		\begin{definition}
			\label{definition:PredicateHasMultiEdge}
			We let \(\predicateHasMultiEdge{\treewidthvalue}\) be the predicate that is true on a pair \((\abstractdecomposition,\awitness) \in \allabstractdecompositionstreewidth{\treewidthvalue} \times \hasMultiEdgeCore[\treewidthvalue].\allwitnesses\) if and only if one of the following is satisfied:
			\begin{enumerate}
				\item If \(\awitness = \multiedgefound\), \(\decompositiongraph{\abstractdecomposition}\) is a multigraph.
				\item If \(\awitness \neq \multiedgefound\), \(\decompositiongraph{\abstractdecomposition}\) is a simple graph and
				\[
					\awitness
					=
					\Bigl\{\{u,v\}\in \powersetchoosek{\topbag{\abstractdecomposition}}{2} : \exists \anedge \in \edgeset{\decompositiongraph{\abstractdecomposition}} \text{ with }
					\edgeendpointsname(\anedge)=\{\topmap{\abstractdecomposition}(u),\topmap{\abstractdecomposition}(v)\}\Bigr\}.
				\]
			\end{enumerate}
		\end{definition}

		The first condition requires that if the witness is the Boolean flag \(\multiedgefound\), the graph associated with \(\abstractdecomposition\) should have multiple edges. The second condition ensures that if the witness is a set of label pairs, the graph is simple and \(\awitness\) encodes \emph{exactly} those edges whose endpoints are currently active.

	\begin{proposition}
		\label{proposition:PredicateHasMultiEdgeCorrectness}
		For each \(\abstractdecomposition \in \allabstractdecompositionstreewidth{\treewidthvalue}\),
		a local witness \(\awitness\) belongs to \(\dynamizationfunction{\hasMultiEdgeCore}{\treewidthvalue}(\abstractdecomposition)\)
		if and only if \(\predicateHasMultiEdge{\treewidthvalue}(\abstractdecomposition,\awitness) = \truevalue\).
		\end{proposition}

		\begin{proof}
		We proceed by structural induction on \(\abstractdecomposition\). Throughout, we use that the cleaning function of \(\dpcore[\treewidthvalue]\) is the identity.

		\medskip
		\noindent\textbf{Base case:} \(\abstractdecomposition=\leaftype\). Then \(\dynamizationfunction{\dpcore}{\treewidthvalue}(\leaftype)=\{\emptyset\}\). Since \(\decompositiongraph{\leaftype}\) is the empty (simple) graph, \(\predicateHasMultiEdge{\treewidthvalue}(\leaftype,\emptyset)=\truevalue\) and \(\predicateHasMultiEdge{\treewidthvalue}(\leaftype,\multiedgefound)=\falsevalue\).

		\medskip
		\noindent\textbf{Induction step:} assume the claim holds for the immediate subterm(s) of \(\abstractdecomposition\).

		\medskip
			\noindent\textbf{Case 1:} \(\abstractdecomposition=\introvertextype{\vertexone}(\sigmaabstractdecomposition)\).

			By Definition~\ref{definition:Dynamization}, a witness \(\awitness\) belongs to \(\dynamizationfunction{\dpcore}{\treewidthvalue}(\abstractdecomposition)\) if and only if there exists \(\awitness'\in \dynamizationfunction{\dpcore}{\treewidthvalue}(\sigmaabstractdecomposition)\) with \(\awitness\in \introvertextype{\vertexone}(\awitness')\). Since \(\introvertextype{\vertexone}(\awitness')=\{\awitness'\}\) for this core, we have \(\awitness=\awitness'\). Moreover, introducing a vertex does not change the edge set (Definition~\ref{definition:itd-constructor-semantics}), and the introduced vertex is isolated at this point. Hence \(\decompositiongraph{\abstractdecomposition}\) has multiple edges if and only if \(\decompositiongraph{\sigmaabstractdecomposition}\) has multiple edges, and the set of edges with both endpoints active is unchanged. The equivalence then follows directly from the induction hypothesis.

		\medskip
		\noindent\textbf{Case 2:} \(\abstractdecomposition=\forgetvertextype{\vertexone}(\sigmaabstractdecomposition)\).

			$(\Rightarrow)$ Suppose \(\awitness\in \dynamizationfunction{\dpcore}{\treewidthvalue}(\abstractdecomposition)\). Then there exists \(\awitness'\in \dynamizationfunction{\dpcore}{\treewidthvalue}(\sigmaabstractdecomposition)\) with \(\awitness\in \forgetvertextype{\vertexone}(\awitness')\). If \(\awitness'=\multiedgefound\), then \(\awitness=\multiedgefound\), and by the induction hypothesis \(\decompositiongraph{\sigmaabstractdecomposition}\) is a multigraph. Since forgetting does not change the graph (Definition~\ref{definition:itd-constructor-semantics}), also \(\decompositiongraph{\abstractdecomposition}\) is a multigraph, so \(\predicateHasMultiEdge{\treewidthvalue}(\abstractdecomposition,\awitness)=\truevalue\). Otherwise \(\awitness'\) is a set of label pairs, and \(\awitness\) is obtained from \(\awitness'\) by removing all pairs that contain \(\vertexone\). By the induction hypothesis, \(\decompositiongraph{\sigmaabstractdecomposition}\) is simple and \(\awitness'\) is exactly the set of label pairs corresponding to edges between active vertices. Since forgetting only removes \(\vertexone\) from the interface and does not change the graph, the remaining pairs in \(\awitness\) are exactly those edges whose endpoints remain active in \(\abstractdecomposition\). Hence \(\predicateHasMultiEdge{\treewidthvalue}(\abstractdecomposition,\awitness)=\truevalue\).

				$(\Leftarrow)$ Suppose \(\predicateHasMultiEdge{\treewidthvalue}(\abstractdecomposition,\awitness)=\truevalue\).
					If \(\awitness=\multiedgefound\), then \(\decompositiongraph{\abstractdecomposition}\) is a multigraph, hence so is \(\decompositiongraph{\sigmaabstractdecomposition}\).
					By the induction hypothesis,
					\[
					\multiedgefound\in \dynamizationfunction{\dpcore}{\treewidthvalue}(\sigmaabstractdecomposition).
					\]
					Since \(\forgetvertextype{\vertexone}(\multiedgefound)=\{\multiedgefound\}\), we get
					\[
					\awitness\in \dynamizationfunction{\dpcore}{\treewidthvalue}(\abstractdecomposition).
					\]

				Otherwise \(\awitness\) is a set of pairs and \(\decompositiongraph{\abstractdecomposition}\) is simple.
				Let \(G\defeq \decompositiongraph{\sigmaabstractdecomposition}=\decompositiongraph{\abstractdecomposition}\) and \(\theta\defeq \topmap{\sigmaabstractdecomposition}\).
				Define \(\awitness'\) to be \(\awitness\) together with all pairs \(\{\vertexone,\vertextwo\}\) such that there is an edge between \(\theta(\vertexone)\) and \(\theta(\vertextwo)\) in \(G\).
				Then \(\predicateHasMultiEdge{\treewidthvalue}(\sigmaabstractdecomposition,\awitness')=\truevalue\), so by the induction hypothesis \(\awitness'\in \dynamizationfunction{\dpcore}{\treewidthvalue}\allowbreak(\sigmaabstractdecomposition)\).
				By definition of the forget transition, \(\forgetvertextype{\vertexone}(\awitness')=\{\awitness\}\), hence \(\awitness\in \dynamizationfunction{\dpcore}{\treewidthvalue}\allowbreak(\abstractdecomposition)\).

		\medskip
		\noindent\textbf{Case 3:} \(\abstractdecomposition=\introedgetype{\vertexone}{\vertextwo}(\sigmaabstractdecomposition)\).

			$(\Rightarrow)$ Suppose \(\awitness\in \dynamizationfunction{\dpcore}{\treewidthvalue}(\abstractdecomposition)\). Then there exists \(\awitness'\in \dynamizationfunction{\dpcore}{\treewidthvalue}(\sigmaabstractdecomposition)\) with \(\awitness\in \introedgegeneric{\vertexone}{\vertextwo}(\awitness')\). If \(\awitness'=\multiedgefound\), then \(\awitness=\multiedgefound\) and the claim follows by induction. Otherwise \(\awitness'\) is a set. If \(\{\vertexone,\vertextwo\}\in \awitness'\), then the transition returns \(\multiedgefound\), and by the induction hypothesis \(\decompositiongraph{\sigmaabstractdecomposition}\) already has an edge between the corresponding vertices. Since \(\introedgetype{\vertexone}{\vertextwo}\) adds another edge between the same endpoints (Definition~\ref{definition:itd-constructor-semantics}), \(\decompositiongraph{\abstractdecomposition}\) is a multigraph and \(\predicateHasMultiEdge{\treewidthvalue}(\abstractdecomposition,\multiedgefound)=\truevalue\). If \(\{\vertexone,\vertextwo\}\notin \awitness'\), then \(\awitness=\awitness'\cup\{\{\vertexone,\vertextwo\}\}\). In this case, \(\decompositiongraph{\sigmaabstractdecomposition}\) is simple by the induction hypothesis and does not contain an edge between the corresponding active vertices. Therefore \(\decompositiongraph{\abstractdecomposition}\) is still simple, and the set of label pairs corresponding to edges between active vertices is exactly \(\awitness\). Hence \(\predicateHasMultiEdge{\treewidthvalue}(\abstractdecomposition,\awitness)=\truevalue\).

			$(\Leftarrow)$ Suppose \(\predicateHasMultiEdge{\treewidthvalue}(\abstractdecomposition,\awitness)=\truevalue\). If \(\awitness=\multiedgefound\), then \(\decompositiongraph{\abstractdecomposition}\) is a multigraph. Removing the last introduced edge yields \(\decompositiongraph{\sigmaabstractdecomposition}\), which is either already a multigraph or has an edge between the same endpoints as the removed edge. In either case, by the induction hypothesis there exists \(\awitness'\in \dynamizationfunction{\dpcore}{\treewidthvalue}(\sigmaabstractdecomposition)\) such that \(\introedgegeneric{\vertexone}{\vertextwo}(\awitness')\) produces \(\multiedgefound\), and hence \(\multiedgefound\in \dynamizationfunction{\dpcore}{\treewidthvalue}(\abstractdecomposition)\). Otherwise \(\awitness\) is a set and \(\decompositiongraph{\abstractdecomposition}\) is simple. Since \(\introedgetype{\vertexone}{\vertextwo}\) introduces an edge between the active vertices labeled \(\vertexone\) and \(\vertextwo\), we have \(\{\vertexone,\vertextwo\}\in \awitness\). Let \(\awitness'\defeq \awitness\setminus\{\{\vertexone,\vertextwo\}\}\). Then \(\decompositiongraph{\sigmaabstractdecomposition}\) is obtained from \(\decompositiongraph{\abstractdecomposition}\) by deleting this edge, hence it is simple and \(\awitness'\) is exactly the set of label pairs corresponding to edges between active vertices in \(\decompositiongraph{\sigmaabstractdecomposition}\). Therefore \(\predicateHasMultiEdge{\treewidthvalue}(\sigmaabstractdecomposition,\awitness')=\truevalue\), so by the induction hypothesis \(\awitness'\in \dynamizationfunction{\dpcore}{\treewidthvalue}(\sigmaabstractdecomposition)\). Since \(\{\vertexone,\vertextwo\}\notin \awitness'\), the intro-edge rule returns \(\awitness'\cup\{\{\vertexone,\vertextwo\}\}=\awitness\), hence \(\awitness\in \dynamizationfunction{\dpcore}{\treewidthvalue}(\abstractdecomposition)\).

		\medskip
		\noindent\textbf{Case 4:} \(\abstractdecomposition=\jointype(\sigmaabstractdecomposition_1,\sigmaabstractdecomposition_2)\).

				$(\Rightarrow)$ Suppose \(\awitness\in \dynamizationfunction{\dpcore}{\treewidthvalue}(\abstractdecomposition)\). Then there exist \(\awitness_1\in \dynamizationfunction{\dpcore}{\treewidthvalue}(\sigmaabstractdecomposition_1)\) and \(\awitness_2\in \dynamizationfunction{\dpcore}{\treewidthvalue}(\sigmaabstractdecomposition_2)\) such that \(\awitness\in \jointype(\awitness_1,\awitness_2)\). If \(\awitness=\multiedgefound\), then either \(\awitness_1=\multiedgefound\) or \(\awitness_2=\multiedgefound\) or \(\awitness_1\cap \awitness_2\neq\emptyset\). In the first two cases, by the induction hypothesis one child graph already has multiple edges, hence so does the join graph. In the third case, there is a label pair \(\{\vertexone,\vertextwo\}\) that corresponds to an edge in both children; by Definition~\ref{definition:itd-constructor-semantics}(e), joining identifies the corresponding endpoints, producing two distinct edges between the same pair of vertices in \(\decompositiongraph{\abstractdecomposition}\). Thus \(\decompositiongraph{\abstractdecomposition}\) is a multigraph and \(\predicateHasMultiEdge{\treewidthvalue}(\abstractdecomposition,\multiedgefound)=\truevalue\). If \(\awitness\neq \multiedgefound\), then \(\awitness=\awitness_1\cup \awitness_2\) and \(\awitness_1\cap \awitness_2=\emptyset\). By the induction hypothesis, both child graphs are simple and \(\awitness_1,\awitness_2\) encode exactly the edges between active vertices in the respective child graphs. By Lemma~\ref{lemma:join-no-cross-edges}, every edge between active vertices in \(\decompositiongraph{\abstractdecomposition}\) comes from exactly one child, and therefore \(\awitness\) is exactly the set of label pairs corresponding to edges between active vertices in the join graph. Hence \(\predicateHasMultiEdge{\treewidthvalue}(\abstractdecomposition,\awitness)=\truevalue\).

			$(\Leftarrow)$ Suppose \(\predicateHasMultiEdge{\treewidthvalue}(\abstractdecomposition,\awitness)=\truevalue\). If \(\awitness=\multiedgefound\), then \(\decompositiongraph{\abstractdecomposition}\) is a multigraph. This can happen either because one child already has multiple edges or because the join identifies two distinct edges between the same labeled endpoints from the two children. In both cases, by the induction hypothesis there exist witnesses \(\awitness_1\in \dynamizationfunction{\dpcore}{\treewidthvalue}\allowbreak(\sigmaabstractdecomposition_1)\) and \(\awitness_2\in \dynamizationfunction{\dpcore}{\treewidthvalue}\allowbreak(\sigmaabstractdecomposition_2)\) such that \(\multiedgefound\in \jointype(\awitness_1,\awitness_2)\), hence \(\multiedgefound\in \dynamizationfunction{\dpcore}{\treewidthvalue}\allowbreak(\abstractdecomposition)\). Otherwise \(\awitness\) is a set and \(\decompositiongraph{\abstractdecomposition}\) is simple. Let \(\awitness_i\) be the set of label pairs that correspond to edges between active vertices in \(\decompositiongraph{\sigmaabstractdecomposition_i}\). Then \(\predicateHasMultiEdge{\treewidthvalue}(\sigmaabstractdecomposition_i,\awitness_i)=\truevalue\), so by induction \(\awitness_i\in \dynamizationfunction{\dpcore}{\treewidthvalue}\allowbreak(\sigmaabstractdecomposition_i)\) for \(i\in\{1,2\}\). Since \(\decompositiongraph{\abstractdecomposition}\) is simple, no label pair can correspond to edges in both children, so \(\awitness_1\cap \awitness_2=\emptyset\). Moreover, every edge between active vertices in the join graph belongs to exactly one child, so \(\awitness=\awitness_1\cup \awitness_2\). Therefore \(\jointype(\awitness_1,\awitness_2)=\{\awitness\}\), and hence \(\awitness\in \dynamizationfunction{\dpcore}{\treewidthvalue}\allowbreak(\abstractdecomposition)\).
		\end{proof}

		\begin{corollary}
			\label{corollary:CorrectnessMultiEdge}
			Let \(\abstractdecomposition\) be a \(\treewidthvalue\)-instructive tree decomposition. Then \(\decompositiongraph{\abstractdecomposition}\) has multiple edges
			if and only if \(\abstractdecomposition \in
				\accepteddecompositions{\dpcore[\treewidthvalue]}\).
		\end{corollary}
		\begin{proof}
			By definition, \(\abstractdecomposition\) is accepted by \(\dpcore[\treewidthvalue]\) if and only if \(\dynamizationfunction{\dpcore}{\treewidthvalue}(\abstractdecomposition)\) contains a final witness, i.e., \(\multiedgefound\). By Proposition~\ref{proposition:PredicateHasMultiEdgeCorrectness}, this holds if and only if \(\decompositiongraph{\abstractdecomposition}\) has multiple edges (equivalently, is not simple).
		\end{proof}

	\begin{observation}
		\label{observation:ComplexityMeasuresHasMultiEdge}
		A \(\dpcore[\treewidthvalue]\) local witness has a bit-length
		\(O(\treewidthvalue^2)\) and timing \(O(\treewidthvalue^2)\).
	\end{observation}
	\begin{proof}
		A \(\dpcore[\treewidthvalue]\) witness is either the Boolean flag \(\multiedgefound\) or a set \(S\subseteq \powersetchoosek{[\treewidthvalue+1]}{2}\) of label pairs.
		Since \(|\powersetchoosek{[\treewidthvalue+1]}{2}|=\binom{\treewidthvalue+1}{2}=O(\treewidthvalue^2)\), we can encode such a set \(S\) as a bitvector indexed by \(\powersetchoosek{[\treewidthvalue+1]}{2}\), using \(O(\treewidthvalue^2)\) bits.
		The operations of the core perform set insertions/deletions and unions/intersections over this universe, which can be implemented in \(O(\treewidthvalue^2)\) time.
	\end{proof}

		\begin{lemma}[Witness action for $\hasMultiEdgeCore$]
	\label{lemma:witness-action-multiedge}
	Fix $\treewidthvalue\in \N$.
	Define $\action_{\dpcore}^{\treewidthvalue}$ on local witnesses by
	\[
	\action_{\dpcore}^{\treewidthvalue}(\relabelingfunction,\multiedgefound)\defeq \multiedgefound
	\]
	and for a set witness $\awitness\subseteq \powersetchoosek{[\treewidthvalue+1]}{2}$ with $\mathrm{Lbl}_{\treewidthvalue}(\awitness)\subseteq \domain(\relabelingfunction)$, let
	\[
	\action_{\dpcore}^{\treewidthvalue}(\relabelingfunction,\awitness)
	\defeq
	\bigl\{\{\relabelingfunction(u),\relabelingfunction(v)\} : \{u,v\}\in \awitness\bigr\}.
	\]
	Then $\action_{\dpcore}^{\treewidthvalue}$ is a witness action for $\dpcore[\treewidthvalue]$ in the sense of Definition~\ref{definition:witnessaction}.
	\end{lemma}
	\begin{proof}
	We verify the conditions of Definition~\ref{definition:witnessaction}.

	\smallskip\noindent
	\emph{Finality invariance.}
	The only final witness is $\multiedgefound$, and it is fixed by the action.

	\smallskip\noindent
	\emph{Label support transport.}
	The label support of $\multiedgefound$ is empty. For a set witness $\awitness\subseteq \powersetchoosek{[\treewidthvalue+1]}{2}$, the only labels appearing in $\awitness$ are those occurring in one of its pairs, so
	\[
	\mathrm{Lbl}_{\treewidthvalue}(\awitness)=\bigcup_{\{u,v\}\in \awitness}\{u,v\}.
		\]
		By definition of the action, each pair $\{u,v\}$ is renamed to $\{\relabelingfunction(u),\relabelingfunction(v)\}$, and therefore
		\[
		\mathrm{Lbl}_{\treewidthvalue}(\action_{\dpcore}^{\treewidthvalue}(\relabelingfunction,\awitness))
		=\relabelingfunction(\mathrm{Lbl}_{\treewidthvalue}(\awitness)).
		\]

	\smallskip\noindent
	\emph{Inverse relabeling, composition, and extension invariance.}
	Since $\relabelingfunction$ is injective, applying $\relabelingfunction^{-1}$ maps each renamed pair back to the original pair, giving inverse relabeling; composition is immediate from $\relabelingfunction(\relabelingfunction'(u))=(\relabelingfunction\circ\relabelingfunction')(u)$ applied elementwise to pairs.
	If $\relabelingfunction'$ extends $\relabelingfunction$, then $\relabelingfunction'$ agrees with $\relabelingfunction$ on $\mathrm{Lbl}_{\treewidthvalue}(\awitness)$, so the action produces the same set of renamed pairs.

	\smallskip\noindent
	\emph{DP-core consistency.}
	If $\awitness=\multiedgefound$, then every transition returns $\{\multiedgefound\}$ and the commutation identities are trivial.
	Otherwise, each transition is defined by set membership and by inserting/removing a single pair or taking union/intersection, and these operations commute with renaming pairs:
	for instance, for $\vertexone,\vertextwo\in \domain(\relabelingfunction)$,
	\[
	\{\vertexone,\vertextwo\}\in \awitness
	\quad\Longleftrightarrow\quad
	\{\relabelingfunction(\vertexone),\relabelingfunction(\vertextwo)\}\in \action_{\dpcore}^{\treewidthvalue}(\relabelingfunction,\awitness),
	\]
	and renaming distributes over set difference and union.
	Thus $\action_{\dpcore}^{\treewidthvalue}$ satisfies Conditions~\ref{condition:relabeldpcore-condition-four}--\ref{condition:relabeldpcore-condition-seven}.
	\end{proof}
}

\end{document}